\documentclass[a4paper, 11pt]{article}
\usepackage[utf8]{inputenc}
\usepackage[english]{babel}
\usepackage{amssymb}
\usepackage{amsthm}
\usepackage{mathrsfs}
\usepackage{graphicx}
\usepackage{epstopdf}
\usepackage{commath}
\usepackage{babelbib}
\usepackage{hyperref}
\usepackage[top=2.5cm, left=2.5cm, right=2.5cm, bottom=2.5cm]{geometry}

\theoremstyle{plain}
\newtheorem{thm}{Theorem}[section]
\newtheorem{lem}[thm]{Lemma}

\newtheorem{cor}[thm]{Corollary}

\theoremstyle{definition}
\newtheorem{defn}[thm]{Definition}
\newtheorem{exmp}[thm]{Example}

\theoremstyle{remark}
\newtheorem{rem}[thm]{Remark}

\newtheorem*{out}{Outline}

\title{\textsc{The Corolla Polynomial for spontaneously broken Gauge Theories}}
\author{David Prinz\footnote{\href{mailto:prinz@physik.hu-berlin.de}{prinz@physik.hu-berlin.de}},\\Institut für Mathematik und Institut für Physik,\\Humboldt-Universität zu Berlin}
\date{\today}

\begin{document}

\maketitle

\begin{abstract}
In \cite{Kreimer_Yeats,Kreimer_Sars_Suijlekom,Sars} the Corolla Polynomial \( \mathcal C (\Gamma) \in \mathbb C [a_{h_1}, \ldots, a_{h_{\left \vert \Gamma^{[1/2]} \right \vert}}] \) was introduced as a graph polynomial in half-edge variables \( \set{a_h}_{h \in \Gamma^{[1/2]}} \) over a 3-regular scalar quantum field theory (QFT) Feynman graph \( \Gamma \). It allows for a covariant quantization of pure Yang-Mills theory without the need for introducing ghost fields, clarifies the relation between quantum gauge theory and scalar QFT with cubic interaction and translates back the problem of renormalizing quantum gauge theory to the problem of renormalizing scalar QFT with cubic interaction (which is super renormalizable in 4 dimensions of spacetime). Furthermore, it is, as we believe, useful for computer calculations. In \cite{Prinz} on which this paper is based\footnote{For reasons of clarity and comprehensibility \cite{Prinz} will not be given explicitly as a reference every time below, even if most passages are literally quoted.} the formulation of \cite{Kreimer_Yeats,Kreimer_Sars_Suijlekom,Sars} gets slightly altered in a fashion specialized in the case of the Feynman gauge. It is then formulated as a graph polynomial \( \mathcal C ( \Gamma ) \in \mathbb C [a_{h_{1 \pm}}, \ldots, a_{h_{\left \vert \Gamma^{[1/2]} \right \vert} \vphantom{h}_\pm}, b_{h_1}, \ldots, b_{h_{\left \vert \Gamma^{[1/2]} \right \vert}}] \) in three different types of half-edge variables \( \set{a_{h_+} , a_{h_-} , b_h}_{h \in \Gamma^{[1/2]}} \). This formulation is also suitable for the generalization to the case of spontaneously broken gauge theories (in particular all bosons from the Standard Model), as was first worked out in \cite{Prinz} and gets reviewed here.
\end{abstract}

\section{Introduction}

\subsection{Motivation}

One of the curiosities of non-abelian gauge theories is the observation that gauge bosons in the quantized theory do not only possess the two experimentally verified transversal degrees of freedom. In addition, there is also a longitudinal one. In quantum electrodynamics photons also possess a longitudinal degree of freedom. The Ward–Takahashi identity then ensures that this longitudinal degree of freedom is cancelled with the corresponding timelike degree of freedom. Therefore, the complete photon amplitude is purely transversal. In non-abelian gauge theories this problem was solved by the introduction of unphysical particles\footnote{In the sense of contradicting the Spin-Statistic theorem, i.e. being scalar particles (having spin 0) but obeying Fermi-Dirac-statistics \cite{Kreimer,Peskin_Schroeder}.}, existing only in closed loops, the so-called ghosts. Then, the Slavnov-Taylor identity ensures that all gauge boson amplitudes are purely transversal. Although it was shown that this could be done in a self-consistent way, it remained unsatisfying since there was no convenient argument despite adjusting the theory to experimental facts. However, this question is solved in the definition of the Corolla polynomial \cite{Kreimer_Sars_Suijlekom} by the introduction of cycle homology and allows, as such, a covariant quantization of gauge fields without the need of introducing ghost fields \cite{Kreimer_Sars_Suijlekom,Kreimer,Peskin_Schroeder,Itzykson_Zuber}.

Furthermore, the introduction of the Corolla polynomial also clarifies the relation between scalar QFT with cubic interaction and quantum gauge theory. This is done using the parametric representation with its two Kirchhoff or Symanzik polynomials and the creation of a Corolla differential out of the Corolla polynomial making implicit use of graph homology \cite{Kreimer_Sars_Suijlekom}. In particular, the parametric integrand for all gauge theory Feynman graphs, which can be created from a 3-regular scalar QFT Feynman graph \(\Gamma\) by shrinking edges, can be obtained whilst acting with the Corolla differential \(\mathcal D_{\text{QCD}} (\Gamma)\) on the parametric integrand \(I (\Gamma)\) such that the parametric integrand for gauge theory \(\tilde{I}_{\mathcal F} (\Gamma)\) \cite{Kreimer_Sars_Suijlekom,Sars} reads\footnote{Actually the different quantum gauge theory graph contributions created by the Corolla differential are hidden in the parametric integrand \(\tilde{I}_{\mathcal F} (\Gamma)\) in its regular part and residues along the Schwinger parameters, cf. \mbox{Subsection \ref{ssec:corolla_polynomial_pure_yang-mills}} (mainly \thmref{thm:gauge_theory_amplitude_pure_yang-mills}).} \cite{Kreimer_Sars_Suijlekom}
\begin{equation}
	\tilde{I}_{\mathcal F} (\Gamma) = \mathcal D_{\text{QCD}} (\Gamma) I (\Gamma) \, .
\end{equation}

It is then possible to receive the renormalized quantum gauge theory amplitude by replacing \(I (\Gamma)\) by its renormalized analogue \(I^R (\Gamma)\), i.e. the problem of renormalizing gauge theory gets translated back to the renormalization of scalar QFT with cubic interaction (which is super renormalizable in 4 dimensions of spacetime). This is a simple consequence of Leibniz's rule for differentiation under the integral sign \cite{Kreimer_Sars_Suijlekom,Kreimer}.

The aim of this paper is now to review the generalization worked out in \cite{Prinz} to the approach of \cite{Kreimer_Yeats,Kreimer_Sars_Suijlekom,Sars} to include the gauge bosons of the electroweak sector of the Standard Model (cf. \mbox{Subsection \ref{ssec:corolla_polynomial_gauge_bosons_electroweak}}) as well as its scalar particles (cf. \mbox{Subsection \ref{ssec:corolla_polynomial_scalar_particles}}). This was achieved by first working out the combinatorics of labeling a 3-regular scalar QFT Feynman graph with labels of the gauge bosons of the electroweak sector of the Standard Model and then by working out the additional tensor structures arising from the inclusion of the Feynman rules for the scalar particles of the electroweak sector of the Standard Model \cite{Prinz,Romao_Silva}.

\subsection{Organization of this Paper}

This paper is organized as follows: First we set up the notations and conventions in \mbox{Section \ref{sec:notation_and_convention}}. Then, we introduce the necessary graph theoretic notions in \mbox{Section \ref{sec:graph_theoretic_notions}}, thereby allowing us to define the parametric representation in a slightly different way which suits our purposes in \mbox{Section \ref{sec:parametric_representations_of_scalar_quantum_field_theories}}. With this knowledge, we are then able to introduce the several Corolla Polynomials and Differentials in \mbox{Section \ref{sec:corolla_polynomial_and_differential}}. We then end, in \mbox{Section \ref{sec:conclusion}}, with a short conclusion.

\section{Notations and Conventions} \label{sec:notation_and_convention}

\subsection{Einstein summation Convention} \label{ssec:einstein_summation_convention}
Important in this paper and also in \cite{Kreimer_Sars_Suijlekom,Sars,Prinz} is that Einstein summation convention is even assumed (if not indicated otherwise) whenever two Lorentz indices are the same --- independent of if they correspond to the space of covariant or contravariant vectors:
\begin{equation}
	\eta^{\mu \nu} \eta^{\nu \rho} := \eta^{\mu \nu_1} \eta_{\nu_1 \nu_2} \eta^{\nu_2 \rho} = \eta^{\mu \rho}
\end{equation}
This at first sight inconvenient looking expression allows to define the Corolla polynomial in a more elegant way, cf. \remref{rem:feynman_gauge}.

\subsection{Feynman Rules, chosen Gauge and Dimension of Spacetime}
We use the definitions for the Feynman rules given in \cite{Romao_Silva} with all appearing signs chosen positive, i.e. \(\eta_G = \eta_s = \eta_e = \eta = \eta_Z = 1\). Furthermore, we use the Feynman gauge throughout this paper --- this allows a more compact notation and avoids unnecessary applications of the Leibniz rule, i.e. \(\xi_G = \xi_A = \xi_W = \xi_Z = 1\), cf. \remref{rem:feynman_gauge}. The relevant Feynman rules with this conventions are explicitly listed in \cite[Appendix A]{Prinz}. Furthermore, we work in a 4-dimensional spacetime for concreteness, even though the Corolla polynomial is not limited to that choice.

\subsection{Feynman Graphs with oriented Edges} \label{ssec:oriented_edges}
In the Standard Model some particle types have oriented edges (fermions, ghosts, \(W^\pm\)-particles, \(\varphi^\pm\)-particles). We work in this paper with unoriented edges, where a graph with unoriented edges is understood as the sum of all graphs with all possible edge orientations times their corresponding symmetry factor (cf. \defnref{defn:symmetry_factor}).

\subsection{Coupling Constants}
In order to avoid confusion between Euler's number and the electric charge, we underline all coupling constants.

\section{Graph theoretic Notions} \label{sec:graph_theoretic_notions}
A QFT is characterized by its Lagrangian density which dictates the sets \(\mathcal R_V\) and \(\mathcal R_E\) of all possible vertex- and edge-types, respectively \cite{Kreimer}. This states all allowed particle interactions and particle types out of which Feynman graphs can be built of. Now, we provide all necessary graph theoretic notions:

\vspace{\baselineskip}

\begin{defn}[Feynman graph, in parts literally quoted from {\cite[Definition 1]{Suijlekom}}] \label{defn:Feynman_graph}
A Feynman graph \(\Gamma\) is characterized by a set of vertices \(\Gamma^{[0]}\) and a set of edges \(\Gamma^{[1]} = \Gamma^{[1]}_{\text{ext}} \cup \Gamma^{[1]}_{\text{int}}\) whose elements are part of \(\mathcal R_V\) and \(\mathcal R_E\), respectively, and a set of maps
\begin{equation}
	\partial_j : \quad \Gamma^{[1]} \mapsto \Gamma^{[0]} \cup \set{1, 2, \ldots, N} \, , \qquad j \in \set{0, 1} \, ,
\end{equation}
respecting the vertex and edge types given by \(\mathcal R _V\) and \(\mathcal R _E\). Furthermore, the case \(\partial_0\) and \(\partial_1\) being both in \(\set{1, 2, \ldots, N}\) is excluded. The set \(\set{1, 2, \ldots, N}\) labels external edges of \(\Gamma\), such that \(\sum_{j=0}^1 \operatorname{card} \partial_j^{-1} (v) = 1\) for all \(v \in \set{1, 2, \ldots, N}\). The set of external edges is therefore defined as \(\Gamma^{[1]}_{\text{ext}} := \bigcup_{j = 0}^1 \partial_j \set{1, 2, \ldots, N}\) and the set of internal edges as its complement with respect to the set of edges of \(\Gamma\), i.e. \(\Gamma^{[1]}_{\text{int}} := \Gamma^{[1]} \setminus \Gamma^{[1]}_{\text{ext}}\). Therefore, external edges can be labeled by \(e_1, e_2, \ldots, e_N \in \Gamma^{[1]}_{\text{ext}}\), with \(e_k := \bigcup_{j=0}^1 \partial_j (k)\).

Furthermore, we omit scalar graphs with edges which form self-loops (so-called tadpoles). This is possible without loss of generality since their amplitudes vanish during the renormalization process\footnote{Speaking in Hopf-algebraic language, the graphs with tadpoles create an Hopf ideal \(I_\text{tad}\) in the Hopf-algebra \(H_{\text{FG}}\) of Feynman graphs and we can effectively work in the quotient space \(H_{\text{FG}} / I_\text{tad}\) \cite{Kreimer_Sars_Suijlekom}.}. In the gauge theory amplitudes created by the Corolla polynomial, amplitudes from gauge theory graphs with tadpoles will show up again by the use of graph-homology, cf. \mbox{Section \ref{sec:corolla_polynomial_and_differential}} and \cite{Kreimer_Sars_Suijlekom}.
\end{defn}

\vspace{\baselineskip}

The two Symanzik polynomials (\defnref{defn:first_symanzik_polynomial} and \defnref{defn:second_symanzik_polynomial}) are polynomials in edge-variables \(\set{A_e}_{e \in \Gamma^{[1]}}\), whereas the Corolla polynomial (\defnref{defn:corolla_polynomial}) is a polynomial in half-edge-variables \(\set{a_{h_+}, a_{h_-} , b_h}_{h \in \Gamma^{[1/2]}}\) which are defined as follows:

\vspace{\baselineskip}

\begin{defn}[Half-edge \cite{Kreimer_Sars_Suijlekom}] \label{defn:half-edge}
Let \(\Gamma\) be a Feynman graph, \(\Gamma^{[0]}\) the set of its vertices, \(\Gamma^{[1]}\) the set of its edges and \(n(v) \subset \Gamma^{[1]}\) the set of edges adjacent to the vertex \(v\). Then the tuple
\begin{equation}
	h := (v, e) \, , \qquad v \in \Gamma^{[0]} \, , \; e \in n (v) \, ,
\end{equation}
is called a half-edge of \(\Gamma\). The set of all half-edges of \(\Gamma\) is denoted by \(\Gamma^{[1/2]}\). Note that each internal edge defines two half-edges in a unique way, since we do not allow tadpoles (cf. \defnref{defn:Feynman_graph}).
\end{defn}

\vspace{\baselineskip}

Next, we define automorphisms of a Feynman graph \(\Gamma\) and its symmetry factor \(\operatorname{sym} (\Gamma)\):

\vspace{\baselineskip}

\begin{defn}[Automorphisms and symmetry factors of a Feynman graph, in parts literally quoted from {\cite[Definition 2]{Suijlekom}}] \label{defn:symmetry_factor}
Let \(\Gamma\) be a Feynman graph. An automorphism \(g\) of \(\Gamma\) is given by an automorphism \(g^{[0]}\) of \(\Gamma^{[0]}\) and an automorphism \(g^{[1]}\) of \(\Gamma^{[1]}\) that is the identity on \(\Gamma^{[1]}_{\text{ext}}\) and fulfilling for all \(e \in \Gamma^{[1]}\)
\begin{equation}
	\bigcup_{j=0}^1 g^{[0]} \left ( \partial_j (e) \right ) = \bigcup_{j=0}^1 \partial_j \left ( g^{[1]} (e) \right ) \, .
\end{equation}
Additionally, we require \(g^{[0]}\) and \(g^{[1]}\) to respect the vertex and edge types given by the sets \(\mathcal R_V\) and \(\mathcal R_E\), respectively.

The automorphism group of \(\Gamma\) is denoted by \(\operatorname{aut} (\Gamma)\) and consists of all such automorphisms of \(\Gamma\). The order of the automorphism group of \(\Gamma\) is called the symmetry factor of \(\Gamma\) and denoted by \(\operatorname{sym} (\Gamma)\), i.e.
\begin{equation}
	\operatorname{sym} (\Gamma) := \left \vert \operatorname{aut} (\Gamma) \right \vert \, .
\end{equation}
\end{defn}

\vspace{\baselineskip}

\begin{defn}[Paths and cycles \cite{Diestel}] \label{defn:paths_and_cycles}
Let \(\Gamma\) be a graph, \(\Gamma^{[0]}\) its vertex set, \(\Gamma^{[1]}\) its edge set and \(\gamma\) a subgraph of \(\Gamma\). Then:
\begin{enumerate}
\item \(\gamma\) is called a path (in \(\Gamma\)) if it is non-empty with vertex set \(\gamma^{[0]} = \set{v_1, v_2, \ldots, v_v} (\subset \Gamma^{[0]})\) and edge set \(\gamma^{[1]} = \set{v_1 v_2, v_2 v_3, \ldots, v_{v-1} v_v} (\subset \Gamma^{[1]})\) (where no repeated vertices are allowed, i.e. \(v_i \neq v_j\) for \(i \neq j\)). In particular, a path connects its two endpoints in a unique way and every internal vertex of the path has precisely two edges attached to it. The set of all paths of \(\Gamma\) is denoted by \(\mathscr P (\Gamma)\) and paths therein by \(P\).
\item \(\gamma\) is called a cycle in mathematics or a loop in physics\footnote{Be aware that a loop in mathematics is what is called a self-loop or a tadpole in physics. We use the terms cycle and loop interchangeably in the above defined sense, depending if the context is more motivated from a mathematical or a physical point of view.} (of \(\Gamma\)) if it is non-empty with vertex set \(\gamma^{[0]} = \set{v_1, v_2, \ldots, v_v} (\subset \Gamma^{[0]})\) and edge set \(\gamma^{[1]} = \set{v_1 v_2, v_2 v_3, \ldots, v_{v-1} v_v, v_v v_1} (\subset \Gamma^{[1]})\) (where no repeated vertices are allowed, i.e. \(v_i \neq v_j\) for \(i \neq j\)). In particular, a cycle can be created from the union of two disjoint paths having the same endpoints. The set of all cycles of a graph \(\Gamma\) is denoted by \(\mathscr C (\Gamma) \), a basis of cycles of \(\Gamma\) by \(\mathfrak C (\Gamma) \) and cycles therein by \(C\).
\end{enumerate}
\end{defn}

\vspace{\baselineskip}

\begin{defn}[Trees, (\(n\))-forests and spanning \(n\)-forests \cite{Diestel}] \label{defn:trees_forests_spanning_n-trees}
Let \(\Gamma\) be a graph and \(\gamma\) a subgraph of \(\Gamma\). Then:

\begin{enumerate}
\item \(\gamma\) is called a forest if it is non-empty and does not contain any cycles. Sets of forests are denoted by \(\mathscr F \) and forests therein by \(F\). If \(F\) has \(n\) connected components, then \(F\) is also called an \(n\)-forest. Sets of \(n\)-forests are denoted by \(\mathscr F_n\).
\item If a forest \(F\) is connected (i.e. a 1-forest) it is also called a tree. Trees are denoted by \(T\) and sets of trees by \(\mathscr T := \mathscr F_1\).
\item If a \mbox{\(n\)-}\-for\-est or a tree is a subgraph of \(\Gamma\) and covers all vertices of \(\Gamma\), then it is called a spanning \mbox{\(n\)-}\-for\-est or a spanning tree of \(\Gamma\), respectively. The sets of all spanning \mbox{\(n\)-}\-for\-ests and spanning trees of a graph \(\Gamma\) are denoted by \(\mathscr F_n (\Gamma)\) and \(\mathscr T (\Gamma)\), respectively. In this paper we are only interested in spanning trees and spanning 2-forests of scalar QFT Feynman graphs.
\end{enumerate}
\end{defn}

\vspace{\baselineskip}

\begin{defn}[\(k\)-Factor \cite{Diestel}] \label{defn:factor}
Let \(\Gamma\) be a graph and \(\gamma\) a subgraph of \(\Gamma\). Then:

\begin{enumerate}
\item \(\gamma\) is called a factor of \(\Gamma\), if \(\gamma\) covers all vertices of \(\Gamma\). The set of all factors of \(\Gamma\) is denoted by \(\mathfrak F (\Gamma)\) and factors therein by \(\mathfrak f\)\footnote{Please note the notational difference between a forest of a graph and a factor of a graph.}.
\item If a factor \( \mathfrak f \in \mathfrak F (\Gamma)\) is \(k\)-regular, i.e. every vertex of \( \mathfrak f \) has valence \(k\), then \( \mathfrak f \) is called a \(k\)-factor of \(\Gamma\). The set of \(k\)-factors for a graph \(\Gamma\) is denoted by \(\mathfrak F_k (\Gamma)\).
\item In this paper we are only interested in 2-factors of scalar QFT Feynman graphs. Therefore, we extend the definition of a \(k\)-factor to the case of graphs with external edges such that external edges also contribute to the valence of the external vertices. Thus, we consider disjoint unions of cycles and paths whose endpoints are external vertices, such that all vertices of the scalar QFT Feynman graph are covered.
\end{enumerate}
\end{defn}

\vspace{\baselineskip}

\begin{defn}[Incidence matrix, in parts literally quoted from {\cite[Page 7]{Kreimer_Sars_Suijlekom}}] \label{defn:incidence_matrix}
We choose an (arbitrary) ordering for the edges of our graph \(\Gamma\). Then we define the incidence matrix \(\varepsilon (\Gamma)\) of a graph \(\Gamma\) componentwise as follows:
\begin{equation}
	\varepsilon_{ve} (\Gamma) = \begin{cases} +1 & \text{if the vertex \(v\) is the endpoint of the edge \(e\)} \\ -1 & \text{if the vertex \(v\) is the starting point of the edge \(e\)} \\ 0 & \text{if \(e\) is not incident to the vertex \(v\)} \end{cases}
\end{equation}
\end{defn}

\vspace{\baselineskip}

\begin{defn}[Assigning 4-vectors to a Feynman graph, in parts literally quoted from  {\cite[Page 7]{Kreimer_Sars_Suijlekom}}] \label{defn:4-vectors}
We assign a 4-vector \(\xi_e^{\prime \mu}\) to each half-edge \(h \in \Gamma^{[1/2]}\) of a Feynman graph \(\Gamma\) in the following way: First we choose a basis of loops \(\mathfrak C_\Gamma \subset \mathscr C_\Gamma\) for \(\Gamma\) and we choose for each \(C \in \mathfrak C_\Gamma\) an orientation \(\varepsilon_{ve}^C\) (where \(\varepsilon_{ve}^C\) is defined in such a way that \(\varepsilon_{v e_1}^C = - \varepsilon_{v e_2}^C\) with \(e_1\) and \(e_2\) being two edges adjacent to the vertex \(v\) and inside the loop \(C\)). Then we assign to each half-edge \(h \equiv (v,e)\) the 4-vector\footnote{The notion of the half-edge \(h\) is here only important to clarify the orientation of the 4-vector \(\xi_e^{\prime \mu}\).}
\begin{equation}
	\varepsilon_{ve} \xi_e^{\prime \mu} := \varepsilon_{ve} \xi_e^\mu + \sum_{C \in \mathfrak C_\Gamma} \sum_{e \in C^{[1]}} \varepsilon_{ve}^C k_C^\mu \, ,
\end{equation}
where the \(\set{\xi_e^\mu}_{e \in \Gamma^{[1]}}\) are independent in the sense that momentum conservation is not assumed until the end of all calculations (cf. \remref{rem:standard_second_Symanzik_polynomial}), and the \(\set{k_C^\mu}_{C \in \mathfrak C_\Gamma}\) are the loop momenta which are to be integrated out.
\end{defn}

\vspace{\baselineskip}

\begin{defn}[Genus of a graph, in parts literally quoted from  {\cite[Page 4]{Kreimer_Sars_Suijlekom}}]
Let \(\Gamma\) be a graph and \(\mathcal M_k\) an oriented 2-dimensional manifold of genus \(k\). Then \(\Gamma\) is said to be \(k\)-compatible, if it can be drawn on \(\mathcal M_k\) without self-intersections. Furthermore, \(\Gamma\) is said to be of genus \(k\), if it can be drawn on \(\mathcal M_k\) without self intersections, but not on \(\mathcal M_l\) with \(l<k\). Planar graphs are of genus \(0\).
\end{defn}

\vspace{\baselineskip}

\begin{defn}[Orientation of a \(3\)-regular graph, in parts literally quoted from  {\cite[Page 4]{Kreimer_Sars_Suijlekom}}] \label{defn:orientation_of_a_3-regular_graph}
Let \(\Gamma\) be a \(3\)-regular \(k\)-compatible Feynman graph, drawn on an oriented 2-dimensional manifold \(\mathcal M_k\) of genus \(k\). Then \(\Gamma\) inherits an orientation by \(\mathcal M_k\) which defines to each half-edge \(h\) incident to a vertex \(v\) a unique successor \(h_+\) and predecessor \(h_-\).
\end{defn}

\section{Parametric Representation of scalar Quantum Field Theories} \label{sec:parametric_representations_of_scalar_quantum_field_theories}

The parametric representation for scalar QFTs can be obtained by the use of the so-called Schwinger trick\footnote{A similar result can be obtained using the so-called Feynman trick \cite{Kreimer}.} \cite{Kreimer_Sars_Suijlekom,Sars,Kreimer,Itzykson_Zuber}:
\begin{equation}
\frac{1}{x} = \int_{\mathbb R_+} e^{-Ax} \dif A
\end{equation}
Using this trick, the product of propagators in the numerator from the standard Feynman rules can be turned into a sum of an exponential function (where an Euclidean spacetime is assumed\footnote{This can be obtained from the Minkowski spacetime using the so-called Wick rotation \cite{Kreimer,Peskin_Schroeder,Itzykson_Zuber}.} and all appearing constants are collected in \(\alpha\)) \cite{Sars,Kreimer,Itzykson_Zuber}:
\begin{equation} \label{eqn:product_propagators}
\begin{split}
	\alpha \prod_{e \in \Gamma^{[1]}_{\text{int}}}  \frac{1}{\left ( \xi_e^{\prime 2} + m_e^2 \right )} = & \alpha \prod_{e \in \Gamma^{[1]}_{\text{int}}} \int_{\mathbb R_+} e^{- A_e \left ( \xi_e^{\prime 2} + m_e^2 \right )} \dif A_e \\ = & \alpha \int_{\mathbb R_+^{\left \vert \Gamma^{[1]}_{\text{int}} \right \vert}} e^{- \left ( \sum_{e \in \Gamma^{[1]}_{\text{int}}} A_e \left ( \xi_e^{\prime 2} + m_e^2 \right ) \right )} \left ( \prod_{e \in \Gamma^{[1]}_{\text{int}}} \dif A_e \right )
\end{split}
\end{equation}

\begin{rem}
For our purposes in defining the Corolla polynomial in \mbox{Section \ref{sec:corolla_polynomial_and_differential}} we alter the standard definition of the parametric representation of a scalar QFT in two ways: First, we will also include Schwinger variables \(A_e\) for external edges and secondly we assign 4-vectors \(\xi_e^{\prime \mu}\) to each edge \(e\) of \(\Gamma\) which we define to consist of the sum of independent variables \(\xi_e^\mu\) and the corresponding loop momenta \(k_C^\mu\) for \(C\) a loop in the basis of loops \(\mathfrak C_\Gamma\) of the Feynman graph \(\Gamma\) (cf. \defnref{defn:4-vectors}).
\end{rem}

\vspace{\baselineskip}

Therefore, we define the following simplified notation:

\vspace{\baselineskip}

\begin{defn}[Abbreviations \cite{Kreimer_Sars_Suijlekom}] We denote: \label{defn:abbreviations}
\begin{enumerate}
\item The simplex of our parametric integration domain by \(\sigma\), i.e. \[ \sigma := \mathbb R_+^{\left \vert \Gamma^{[1]} \right \vert} \, . \]
\item The measure of our extended parametric space by \(\dif \underline{A}_\Gamma\), i.e. \[ \dif \underline{A}_\Gamma := \bigwedge_{e \in \Gamma^{[1]}} \dif A_e \, . \]
\item The space of all loop momenta by \(\rho\) (recall \(\left \vert \mathfrak C_\Gamma \right \vert\) to be the dimension of the basis of loops of \(\Gamma\) from \defnref{defn:paths_and_cycles} and \defnref{defn:4-vectors}), i.e.\footnote{Note that the dimension is here actually the dimension of the basis of loops of \(\Gamma\) times the dimension of spacetime, but we're working in a 4-dimensional spacetime throughout this paper \cite{Kreimer,Peskin_Schroeder,Itzykson_Zuber}.} \[ \rho := \mathbb R^{4 \left \vert \mathfrak C_\Gamma \right \vert} \, . \]
\item The measure of the loop momenta integral by \(\dif \underline{k}_\Gamma\) (recall that we choose a basis of loops \(\mathfrak C_\Gamma \subset \mathscr C_\Gamma\) of \(\left \vert \mathscr C_\Gamma \right \vert\) independent loops of \(\Gamma\) from \defnref{defn:4-vectors}, and \(\operatorname{d}^4 \! k_C\) being the usual Lorentz invariant loop momentum measure for the loop \(C\) \cite{Kreimer,Peskin_Schroeder,Itzykson_Zuber}), i.e.\footnote{The power 4 in the measure \(\operatorname{d}^4 \! k_C\) is actually the dimension of spacetime, but again we're working in a 4-dimensional spacetime throughout this paper \cite{Kreimer,Peskin_Schroeder,Itzykson_Zuber}.} \[ \dif \underline{k}_\Gamma := \bigwedge_{C \in \mathfrak C_\Gamma} \operatorname{d}^4 \! k_C \, . \]
\item The extended universal quadric by \(\underline{Q}_\Gamma\) (where the case \(m_e = 0\) for some \(e \in \Gamma^{[1]}\) is allowed), i.e. \[ \underline{Q}_\Gamma := \sum_{e \in \Gamma^{[1]}} \left ( \xi_e^{\prime 2} + m_e^2 \right ) A_e \, . \]
\item The reduced universal quadric by \(\underline{q}_\Gamma\) (where again the case \(m_e = 0\) for some \(e \in \Gamma^{[1]}\) is allowed), i.e. \[ \underline{q}_\Gamma := \left ( \sum_{e \in \Gamma^{[1]}_{\text{ext}}} \xi_e^2 A_e \right ) + \left ( \sum_{e \in \Gamma^{[1]}} m_e^2 A_e \right ) \, . \]
\item The product over the inverse external propagators by \(P_\Gamma\), i.e. \[ P_\Gamma := \prod_{e \in \Gamma^{[1]}_\text{ext}} \left ( \xi_e^2 + m_e^2 \right ) \, . \]
\item The differential form concerning the parametric space by \(I (\Gamma)\), i.e. \[ I (\Gamma) := \left ( \alpha P_\Gamma \int_\rho e^{- \underline{Q}_\Gamma} \prod_{v \in \Gamma^{[0]}} \delta^{(4)} \left ( \sum_{e \in \Gamma^{[1]}} \varepsilon_{ve} k_e^\mu \right ) \dif \underline{k}_\Gamma \right ) \dif \underline{A}_\Gamma \, . \]
\end{enumerate}
\end{defn}

\vspace{\baselineskip}

Then, \mbox{Equation \eqref{eqn:product_propagators}} reads
\begin{equation}
	\alpha \prod_{e \in \Gamma^{[1]}_{\text{int}}}  \frac{1}{\left ( \xi_e^{\prime 2} + m_e^2 \right )} = \alpha P_\Gamma \int_\sigma e^{- \underline{Q}_\Gamma} \dif \underline{A}_\Gamma \, .
\end{equation}

One of the advantages of passing to the parametric space is that now the loop momentum integrals can be carried out after changing the order of integration\footnote{The change of the integration order for Minkowski spacetime or massless particles (i.e. ill-defined integral expressions) is formally only allowed if regulators \(i \epsilon\) are introduced before in each such propagators and whose limits to 0 are understood to be taken after the integrations are carried out \cite{Kreimer,Itzykson_Zuber}.}. In doing so, the two so-called Kirchhoff- or Symanzik-polynomials \(\psi_\Gamma\) and \(\phi_\Gamma\)\footnote{We will slightly alter the standard definition of the second Symanzik polynomial for our purposes and denote it by \(\varphi_\Gamma\), cf. \defnref{defn:second_symanzik_polynomial}.} come into play:

\vspace{\baselineskip}

\begin{defn}[First Symanzik polynomial \cite{Kreimer_Sars_Suijlekom,Sars,Kreimer}] \label{defn:first_symanzik_polynomial}
Let \(\Gamma\) be a scalar QFT Feynman graph and \(\mathscr T (\Gamma)\) the set of its spanning trees \(T\). Then we define the first Symanzik polynomial as (external edges are excluded from the product)
\begin{equation}
	\psi_\Gamma := \sum_{T \in \mathscr T (\Gamma)} \prod_{e \notin T} A_e \, .
\end{equation}
\end{defn}

\vspace{\baselineskip}

\begin{exmp} \label{exmp:1L-first_symanzik_polynomial}
We consider the one-loop self-energy graph
\begin{equation}
\Gamma := 
\vcenter{\hbox{\includegraphics[height=2.5cm]{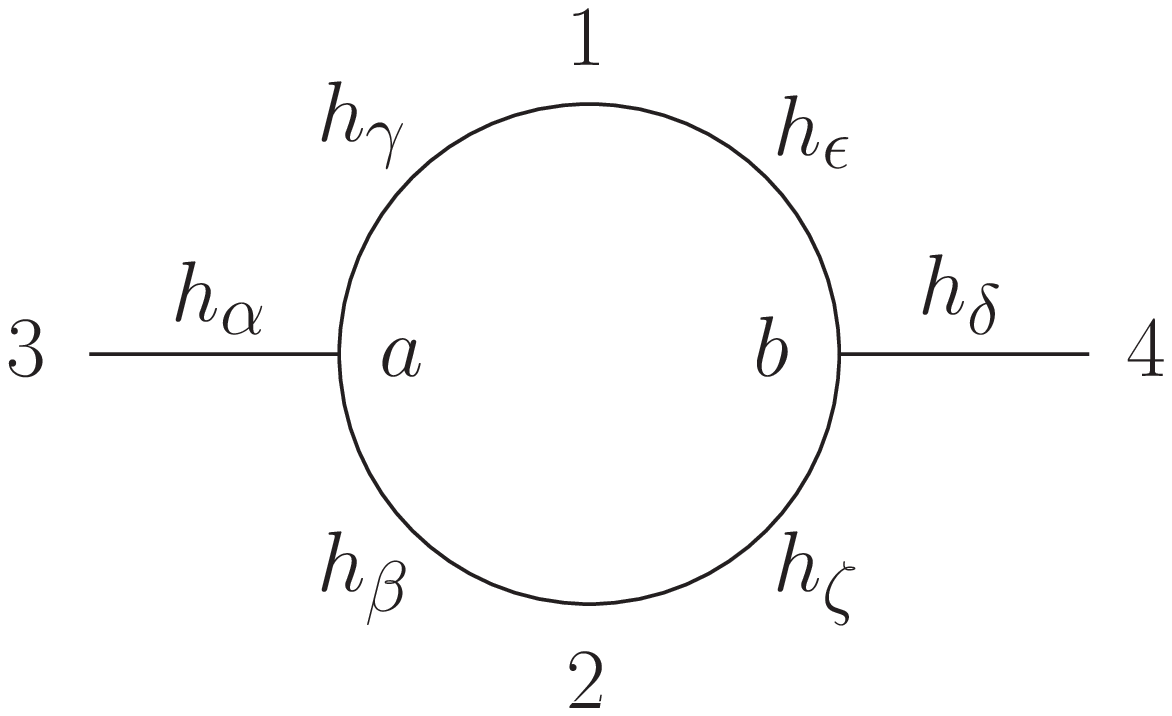}}} \, .
\end{equation}
We have \(\mathscr T (\Gamma) = \set{1, 2}\), and so
\begin{equation}
	\psi_\Gamma = A_1 + A_2 \, .
\end{equation}
\end{exmp}

\vspace{\baselineskip}

\begin{defn}[Second Symanzik polynomial, non-standard definition \cite{Kreimer_Sars_Suijlekom,Sars,Kreimer}] \label{defn:second_symanzik_polynomial}
Let \(\Gamma\) be a scalar QFT Feynman graph, \(\mathscr F_2\) the set of its spanning 2-forests \(F\), which consist of the two components \(F^{(1)}\) and \(F^{(2)}\) (i.e. \(F = \bigcup_{i=1}^2 F^{(i)}\)), and let \(\varepsilon (\Gamma)\) be its incidence matrix (cf. \defnref{defn:incidence_matrix}). Then we define the second Symanzik polynomial as (again, external edges are excluded from the sum and the product)
\begin{equation}
	\varphi_\Gamma := \sum_{F \in \mathscr F_2 (\Gamma)} \left ( \sum_{e \notin F} \tau_F (e) \xi_e^\mu \right ) ^2 \prod_{e \notin F} A_e \, ,
\end{equation}
with\footnote{Note that if we fix the two components \(F^{(1)}\) and \(F^{(2)}\) of a spanning 2-forest \(F \in \mathscr F_2 \left ( \Gamma \right ) \) and apply momentum conservation as described in \remref{rem:standard_second_Symanzik_polynomial}, the standard second Symanzik polynomial becomes independent of that choice since the resulting expressions are squares.}
\begin{equation}
	\tau_F (e) := - \sum_{v \in {F^{(1)}}^{[0]}} \varepsilon_{v e} \left ( \Gamma \right ) \equiv \begin{cases} +1 & \text{if \(e\) is oriented from \(F^{(1)}\) to \(F^{(2)}\)} \\ -1 & \text{if \(e\) is oriented from \(F^{(2)}\) to \(F^{(1)}\)} \end{cases} \, .
\end{equation}
\end{defn}

\vspace{\baselineskip}

\begin{rem} \label{rem:standard_second_Symanzik_polynomial}
Note that the usual expression \(\phi_\Gamma\) for the second Symanzik polynomial of a Feynman graph \(\Gamma\) can be obtained by setting the \(\set{\xi_e^\mu}_{e \in \Gamma^{[1]}}\) in accordance with the physical momenta \(\set{p_e^\mu}_{e \in \Gamma^{[1]}}\) flowing through the graph (without the loop momenta \(\set{k_C^\mu}_{C \in \mathfrak C_\Gamma}\)) \cite[Page 9]{Kreimer_Sars_Suijlekom}, i.e.
\begin{equation}
	Q : \quad \xi_e^\mu \mapsto q_e^\mu, \qquad \forall e \in \Gamma^{[1]} \, .
\end{equation}
\end{rem}

\vspace{\baselineskip}

\begin{exmp} \label{exmp:1L-second_symanzik_polynomial}
Again, consider the one-loop self-energy graph with incoming external momenta \(p^{\mu_3}\) and \(p^{\mu_4}\)
\begin{equation}
\Gamma := 
\vcenter{\hbox{\includegraphics[height=2.5cm]{1L_Sc_labeled.eps}}} \, .
\end{equation}
We have \(\mathscr F_2 \left ( \Gamma \right ) = \set{\set{a, b}}\), and so
\begin{equation}
	\varphi_\Gamma = \left ( \xi_1^\mu - \xi_2^\mu \right )^2 A_1 A_2 \, .
\end{equation}
\end{exmp}

\vspace{\baselineskip}

\begin{thm}[Parametric integrand with non-standard second Symanzik polynomial \cite{Kreimer_Sars_Suijlekom}] \label{thm:parametric_integrand_non-standard}
Integrating out the loop momenta in the parametric representation gives rise to the two Symanzik polynomials\footnote{The square of \(\psi_\Gamma\) in the numerator is actually a \(d/2\), with \(d\) being the dimension of spacetime, but once more we're working in a 4-dimensional spacetime throughout this paper \cite{Sars,Kreimer}.}:
\begin{equation}
	\int_\rho e^{- \underline{Q}_\Gamma} \prod_{v \in \Gamma^{[0]}} \delta^{(4)} \left ( \sum_{e \in \Gamma^{[1]}} \varepsilon_{ve} k_e^\mu \right ) \dif \underline{k}_\Gamma = \frac{e^{- \frac{\varphi_\Gamma}{\psi_\Gamma} - \underline{q}_\Gamma}}{\psi_\Gamma^2}
\end{equation}

In particular, the parametric integrand \(I (\Gamma)\) with non-standard second Symanzik polynomial can be written as
\begin{equation}
\begin{split}
	I (\Gamma) = & \left ( \alpha P_\Gamma \int_\rho e^{- \underline{Q}_\Gamma} \prod_{v \in \Gamma^{[0]}} \delta^{(4)} \left ( \sum_{e \in \Gamma^{[1]}} \varepsilon_{ve} k_e^\mu \right ) \dif \underline{k}_\Gamma \right ) \dif \underline{A}_\Gamma \\ = & \left ( \alpha P_\Gamma \frac{e^{- \frac{\varphi_\Gamma}{\psi_\Gamma} - \underline{q}_\Gamma}}{\psi_\Gamma^2} \right ) \dif \underline{A}_\Gamma
\end{split}
\end{equation}

\end{thm}

\begin{proof}
We refer to \cite[Pages 294 -- 299]{Itzykson_Zuber} for a proof of the parametric integrand with standard second Symanzik polynomial and to \cite[Page 10]{Kreimer_Sars_Suijlekom} for the notational difference concerning the non-standard second Symanzik polynomial.
\end{proof}

\begin{exmp} \label{exmp:1L-parametric_integrand}
Continuing \exref{exmp:1L-first_symanzik_polynomial} and \exref{exmp:1L-second_symanzik_polynomial} for the one-loop self-energy graph
\begin{equation}
\Gamma := 
\vcenter{\hbox{\includegraphics[height=2.5cm]{1L_Sc_labeled.eps}}} \, ,
\end{equation}
we have
\begin{equation}
	I (\Gamma) = \left ( \left ( \xi_3^2 + m_3^2 \right ) \left ( \xi_4^2 + m_4^2 \right ) \frac{e^{- \frac{\left ( \xi_1^\mu - \xi_2^\mu \right ) ^2 A_1 A_2}{A_1 + A_2} - \xi_3^2 A_3 - \xi_4^2 A_4 - \sum_{e = 1}^4 m_e^2 A_e}}{\left ( A_1 + A_2 \right )^2} \right ) \dif A_1 \wedge \dif A_2 \, .
\end{equation}
\end{exmp}

\vspace{\baselineskip}

\begin{defn}[Parametric integrand for gauge theory amplitudes with non-standard second Symanzik polynomial, in parts literally quoted from {\cite[Page 10]{Kreimer_Sars_Suijlekom}}]
In the following we're in particular interested in gauge theory amplitudes. They can be represented in the parametric space by a slightly generalization of \(I (\Gamma)\) via
\begin{subequations}
\begin{equation}
	\tilde{I}_{\mathcal F} (\Gamma) := \mathcal F I (\Gamma) \, ,
\end{equation}
with
\begin{equation}
	\mathcal F := \frac{\mathcal F_N (\set{A_e}_{e \in \Gamma^{[1]}})}{\mathcal F_D (\set{A_e}_{e \in \Gamma^{[1]}})} \, ,
\end{equation}
\end{subequations}
a rational function in the Schwinger parameters \(A_e\) and possible matrix structure. The rational function \(F\) can be created by acting with suitable differential operators on the parametric integrand \(I (\Gamma)\) and possibly multiply it with the required matrices and the contraction of their respective Lorentz indices \cite{Kreimer_Sars_Suijlekom,Sars}.
\end{defn}

\vspace{\baselineskip}

\begin{rem}
Before \(I (\Gamma)\) (or \(\tilde{I}_{\mathcal F} (\Gamma)\)) can be integrated against the simplex \(\sigma\) to yield the Feynman amplitude \(\mathcal I^R (\Gamma)\) (or \(\tilde{\mathcal I}^R_{\mathcal F} (\Gamma)\)) of \(\Gamma\) it needs to be renormalized first (cf. \cite{Kreimer_Sars_Suijlekom,Sars,Kreimer,Peskin_Schroeder,Itzykson_Zuber,Suijlekom,Brown_Kreimer} for the huge topic of renormalization). The corresponding renormalized differential form \(I^R (\Gamma)\) (or \(\tilde{I}_{\mathcal F}^R (\Gamma)\)) can be obtained from \(I (\Gamma)\) (or \(\tilde{I}_{\mathcal F} (\Gamma)\)) by the use of the Forest formula \cite{Kreimer_Sars_Suijlekom,Sars,Kreimer,Itzykson_Zuber,Brown_Kreimer}.
\end{rem}

\section{Corolla Polynomial and Differential} \label{sec:corolla_polynomial_and_differential}

\subsection{Pure Yang-Mills Theory} \label{ssec:corolla_polynomial_pure_yang-mills}

Now we define the Corolla polynomial for pure Yang-Mills theory:

\vspace{\baselineskip}

\begin{defn}[Corolla polynomial, in parts literally quoted from {{\cite[Definition 1]{Kreimer_Yeats} and \cite[Page 27]{Kreimer_Sars_Suijlekom}}}, cf. \cite{Sars}] \label{defn:corolla_polynomial}
Let \(\Gamma\) be a 3-regular scalar QFT Feynman graph. Then:
\begin{enumerate}
\item Associate to each half-edge \(h\) of \(\Gamma\) variables \(a_{h_+}\), \(a_{h_-}\) and \(b_h\).
\item Recall from \defnref{defn:half-edge}: For a vertex \(v\) of \(\Gamma\) we denote the set of the three half-edges incident to \(v\) by \(n(v)\).
\item Recall from \defnref{defn:orientation_of_a_3-regular_graph}: For a vertex \(v\) of \(\Gamma\) and \(h \in n(v)\) a half-edge of \(\Gamma\), we denote according to the orientation of \(\Gamma\) its successor by \(h_+\) and its predecessor by \(h_-\).
\item We denote the edge \(e\) of which \(h\) is part of by \(e(h)\).
\item Recall from \defnref{defn:paths_and_cycles}: We denote the set of all cycles of \(\Gamma\) by \(\mathscr C_\Gamma\).
\item For \(C \in \mathscr C_\Gamma\) a cycle and \(v\) a vertex in \(C\), since \(\Gamma\) is 3-regular, there is a unique half-edge of \(\Gamma\) incident to \(v\) and not in \(C\). We denote this half-edge by \(h(C,v)\).
\item Assign to \(\Gamma\) a factor \(\operatorname{color} (\Gamma)\).
\item We denote the combination of half-edge variables which will create in the Corolla differential the Feynman rules for the 3-valent gluon vertex \(v\) (and also the 4-valent gluon vertex as residues in Schwinger parameters if applied twice with special combinations of half-edge variables due to the Leibniz rule, cf. \thmref{thm:gauge_theory_amplitude_pure_yang-mills} and \cite{Kreimer_Sars_Suijlekom,Sars}) as \[\mathscr V_v := \sum_{h \in n(v)} b_h \left ( a_{h_{+}} + a_{h_{-}} \right )\, . \]
\item We denote the combination of half-edge variables which will create in the Corolla differential the Feynman rules for the closed ghost loop \(C_j\) as (recall that we work with unoriented ghost edges which means that we sum over both orientations, cf. \mbox{Subsection \ref{ssec:oriented_edges}}) \[\mathscr G_{C_j} := \sum_{k \in \set{+,-}} \prod_{v \in C_j^{[0]}} a_{h(C_j,v)_{k}} \, . \]
\end{enumerate}
Then we can define the various summands of the Corolla polynomial for pure Yang-Mills theory by
\begin{subequations}
\begin{equation}
	\mathcal C ^0 (\Gamma) := \prod_{v \in \Gamma^{[0]}} \mathscr V_v
\end{equation}
and for \(i \geq 1\) by
\begin{equation} \label{eqn:C^i}
	\mathcal C ^i (\Gamma) := \sum_{\substack{C_1, C_2, \ldots, C_i \in \mathscr{C}_\Gamma,\\C_j \text{ pairwise disjoint}}} \left [ \left ( \prod_{j = 1}^i \mathscr G_{C_j} \right ) \left ( \prod_{\substack{v \in \Gamma^{[0]},\\v \notin \bigcup_{k=1}^i C_k^{[0]}}} \mathscr V_v \right ) \right ] \, ,
\end{equation}
where we define \(\prod_{v \in \emptyset} \mathscr V_v := 1\), i.e. if \(\set{v \in \Gamma^{[0]} \vert v \notin \bigcup_{k=1}^i C_k} = \emptyset\) (all vertices are attached to edges that will be turned into a ghost edges). Finally, we introduce the Corolla polynomial as the alternating sum over its summands:
\begin{equation}
	\mathcal C (\Gamma) := \sum_{i=1}^\infty (-1)^i \mathcal C ^i (\Gamma)
\end{equation}
\end{subequations}
Furthermore, we define
\begin{equation}
	\mathcal C_{\text{QCD}} (\Gamma) := i^{\left \vert \Gamma^{[1]} \right \vert} \underline{g_s}^{\left \vert \Gamma^{[0]} \right \vert} \operatorname{color} (\Gamma) \mathcal C (\Gamma) \, .
\end{equation}
\end{defn}

\vspace{\baselineskip}

\begin{rem} \label{rem:feynman_gauge}
Note that our definition of the Corolla polynomial in the half-edge variables \( \set{a_{h_+} , a_{h_-} , b_h}_{h \in \Gamma^{[1/2]}} \) differs from the ones given in \cite{Kreimer_Yeats,Kreimer_Sars_Suijlekom,Sars}, but we will also define different differential operators \(\set{\mathscr A_{h_+} , \mathscr A_{h_-}}_{h \in \Gamma^{[1/2]}}\) and metric tensors \(\set{\mathscr B_h}_{h \in \Gamma^{[1/2]}}\) such that the Corolla differentials \(\mathcal D (\Gamma)\) coincide again (despite from factors of \(1/4^i\) which were missing in each summand \(\mathcal D^i (\Gamma)\) in \cite{Kreimer_Yeats,Kreimer_Sars_Suijlekom}, cf. \cite{Sars} and the fact that our definition does not create 4-valent ghost vertex contributions from the start which need to be eliminated in the formulation of \cite{Kreimer_Yeats,Kreimer_Sars_Suijlekom,Sars} while working in a linear covariant gauge otherwise, cf. \cite{Kreimer_Sars_Suijlekom}). Furthermore, the above formulation, originating from \cite{Prinz}, is only valid for the Feynman gauge (which is in particular linear). This simplifies the numerator of gauge boson propagators to metric tensors, cf. \cite{Kreimer_Sars_Suijlekom} for the general case. Together with our extended definition of the Einstein summation convention in \mbox{Subsection \ref{ssec:einstein_summation_convention}} we can effectively work with scalar QFT propagators since all relevant quantum gauge theory tensor structures are created from the corresponding vertex Feynman rules.
\end{rem}

\vspace{\baselineskip}

\begin{rem}
\(\mathcal C (\Gamma)\) is a polynomial since \(\mathcal C ^i (\Gamma) = 0\) for all \(i > \vert \mathfrak C_\Gamma \vert\) \cite{Kreimer_Yeats,Kreimer_Sars_Suijlekom}.
\end{rem}

\vspace{\baselineskip}

\begin{exmp} \label{exmp:1L_Sc_Corolla_polynomial}
Consider the one-loop self-energy graph
\begin{equation}
\Gamma := 
\vcenter{\hbox{\includegraphics[height=2.5cm]{1L_Sc_labeled.eps}}}
\end{equation}
with the six half-edges \(\set{h_{\alpha} := (a, 3), h_{\beta} = (a, 2), h_{\gamma} = (a, 1), h_{\delta} = (b, 4), h_{\epsilon} = (b, 1), h_{\zeta} = (b, 2)}\).

Then, we have
\begin{subequations}
\begin{equation}
\begin{split}
	\mathcal C^0 (\Gamma) = & \left ( b_{h_{\alpha}} \left ( a_{h_{\alpha +}} + a_{h_{\alpha -}} \right ) + b_{h_{\beta}} \left ( a_{h_{\beta +}} + a_{h_{\beta -}} \right ) + b_{h_{\gamma}} \left ( a_{h_{\gamma +}} + a_{h_{\gamma -}} \right ) \right ) \\ & \times \left ( b_{h_{\delta}} \left ( a_{h_{\delta +}} + a_{h_{\delta -}} \right ) + b_{h_{\epsilon}} \left ( a_{h_{\epsilon +}} + a_{h_{\epsilon -}} \right ) + b_{h_{\zeta}} \left ( a_{h_{\zeta +}} + a_{h_{\zeta -}} \right )  \right ) \, ,
\end{split}
\end{equation}
\begin{equation}
	\mathcal C^1 (\Gamma) = a_{h_{\alpha +}} a_{h_{\delta +}} + a_{h_{\alpha -}} a_{h_{\delta -}}
\end{equation}
and
\begin{equation}
	\mathcal C^i (\Gamma) = 0 \, , \qquad \forall i > 1 \, .
\end{equation}
\end{subequations}
So, in total we get:
\begin{equation}
\begin{split}
	\mathcal C (\Gamma) = & \mathcal C^0 (\Gamma) - \mathcal C^1 (\Gamma) \\ = & \left ( b_{h_{\alpha}} \left ( a_{h_{\alpha +}} + a_{h_{\alpha -}} \right ) + b_{h_{\beta}} \left ( a_{h_{\beta +}} + a_{h_{\beta -}} \right ) + b_{h_{\gamma}} \left ( a_{h_{\gamma +}} + a_{h_{\gamma -}} \right ) \right ) \\ & \phantom{ ( } \times \left ( b_{h_{\delta}} \left ( a_{h_{\delta +}} + a_{h_{\delta -}} \right ) + b_{h_{\epsilon}} \left ( a_{h_{\epsilon +}} + a_{h_{\epsilon -}} \right ) + b_{h_{\zeta}} \left ( a_{h_{\zeta +}} + a_{h_{\zeta -}} \right )  \right ) \\ & - a_{h_{\alpha +}} a_{h_{\delta +}} - a_{h_{\alpha -}} a_{h_{\delta -}}
\end{split}
\end{equation}
\end{exmp}

\vspace{\baselineskip}

\begin{defn}[Corolla differential, in parts literally quoted from {\cite[Page 29]{Kreimer_Sars_Suijlekom}}] \label{defn:corolla_differential}
Let \(\Gamma\) be a 3-regular scalar QFT Feynman graph. Then:
\begin{enumerate}
\item Assign to each external and internal edge a variable \(A_e\) and a 4-vector \(\xi_e^{\prime \mu}\) (as in \defnref{defn:4-vectors}) and to each edge \(e\) a Lorentz index \(\mu_{e}\).
\item Define for each half-edge \(h_k\) the following differential operator (where \(k \in \set{\pm}\)): \[ \mathscr A_{h_{k}} := - k \varepsilon_{h_k} \frac{1}{2 A_{e(h_k)}} \frac{\partial}{\partial \xi_{e(h_k) \mu_{e(h)}}} \]
\item Define for each half-edge \(h\) the following metric tensor: \[ \mathscr B_{h} := \eta^{\mu_{e \left ( h_+ \right )} \mu_{e \left ( h_- \right )}} \]
\end{enumerate}
Then, the summands of the Corolla differential \(\mathcal D^i (\Gamma)\) are defined via the summands of the Corolla polynomial \(\mathcal C^i (\Gamma)\) by replacing each half-edge variable \(a_{h_k}\) by the corresponding differential operator \(\mathscr A _{h_k}\) (denoted by \(a_{h_k} \mapsto \mathscr A _{h_k}\)) and each half-edge variable \(b_{h}\) by the corresponding metric tensor \(\mathscr B _h\) (denoted by \(b_h \mapsto \mathscr B _h\)):
\begin{subequations}
\begin{equation}
	\mathcal D^i (\Gamma) := \eval{\mathcal C^i (\Gamma)}_{\substack{a_{h_k} \mapsto \mathscr A _{h_k} , \, \forall h_k \in \Gamma^{[1/2]}\\b_{h} \mapsto \mathscr B _{h} , \, \forall h \in \Gamma^{[1/2]}}} \, ,
\end{equation}
and similarly the Corolla differential \(\mathcal D (\Gamma)\) is defined like the Corolla polynomial \(\mathcal C (\Gamma)\) as
\begin{equation}
	\mathcal D (\Gamma) := \sum_{i=1}^\infty (-1)^i \mathcal D ^i (\Gamma) \equiv \eval{\mathcal C (\Gamma)}_{\substack{a_{h_k} \mapsto \mathscr A _{h_k} , \, \forall h_k \in \Gamma^{[1/2]}\\b_{h} \mapsto \mathscr B _{h} , \, \forall h \in \Gamma^{[1/2]}}} \, .
\end{equation}
\end{subequations}
Likewise, the Corolla differentials \(\mathcal D_{\text{QCD}} (\Gamma)\) and \(\mathcal D_{\text{EW}} (\Gamma)\) are defined via the replacement of the half-edge variables \(a_{h_k}\) by the corresponding differential operators \(\mathscr A_{h_k}\) and the half-edge variables \(b_h\) by the corresponding metric tensors \(\mathscr B_h\) in their corresponding Corolla polynomials \(\mathcal C_{\text{QCD}} (\Gamma)\) and \(\mathcal C_{\text{EW}} (\Gamma)\) (the latter one will be defined in \defnref{defn:corolla_polynomial-electroweak_scalar_bosons}), respectively.
\end{defn}

\vspace{\baselineskip}

\begin{rem}
\begin{enumerate}
\item Expressions of the form \(\mathcal D (\Gamma) I (\Gamma)\) are to be understood as\footnote{Where the \(\set{q_e^\mu}_{e \in \Gamma^{[1]}}\) are the physical 4-momenta of the scalar QFT Feynman graph \(\Gamma\), i.e. the physical momenta flowing through the graph (without the loop momenta \(\set{k_C^\mu}_{C \in \mathfrak C_\Gamma}\))} \[ \mathcal D (\Gamma) I (\Gamma) := \eval{ \left ( \mathcal D (\Gamma) I (\Gamma) \right )}_{\xi_e^\mu \mapsto q_e^\mu , \, \forall e \in \Gamma^{[1]}} \, , \] such that the Corolla differential acts on the parametric integrand with non-standard second Symanzik polynomial (cf. \defnref{defn:second_symanzik_polynomial} and \thmref{thm:parametric_integrand_non-standard}) and after the differentiation the standard second Symanzik polynomial is obtained by setting the \(\set{\xi_e^\mu}_{e \in \Gamma^{[1]}}\) in accordance with the external momenta (cf. \remref{rem:standard_second_Symanzik_polynomial}). Note that due to the Leibniz rule we get also contributions from differentiating the inverse external propagators \(P_\Gamma\) (cf. 7. from \defnref{defn:abbreviations}) but since the differential operators for the external edges are of order 1, these contributions vanish again whilst setting the \(\set{\xi_e^\mu}_{e \in \Gamma^{[1]}}\) in accordance with the external momenta (again, cf. \remref{rem:standard_second_Symanzik_polynomial}).
\item We choose the 4-vectors \(\set{\xi_e^\mu}_{e \in \Gamma^{[1]}}\) assigned to each edge of the graph independently, i.e. we do not assume momentum conservation until the Corolla differential acted on the scalar integrand (cf. \defnref{defn:4-vectors}). Therefore, they can be seen as ``temporary dummy variables''. After applying the Corolla differential, we assume momentum conservation and they acquire the meaning of physical 4-momentum vectors of the momenta flowing through the graph (without the loop momenta \(\set{k_C^\mu}_{C \in \mathfrak C_\Gamma}\)) and are then denoted as \(\set{q_e^\mu}_{e \in \Gamma^{[1]}}\).
\item \(\mathcal D ^i (\Gamma)\) creates, whilst acting on the corresponding parametric scalar QFT integrand \(I (\Gamma)\), all possible pure Yang-Mills theory Feynman graphs with \(\left \vert \mathfrak C_\Gamma \right \vert\) loops and \(i\) ghost loops. The alternating sum in the Corolla polynomial \(\mathcal C (\Gamma)\) and hence also in the Corolla differential \(\mathcal D (\Gamma)\) takes care of the minus sign for each closed ghost loop (cf. \thmref{thm:gauge_theory_amplitude_pure_yang-mills}).
\end{enumerate}
\end{rem}

\vspace{\baselineskip}

\begin{thm}[\cite{Kreimer_Sars_Suijlekom}] \label{thm:gauge_theory_amplitude_pure_yang-mills}
Using the above definitions, the renormalized amplitude \(\tilde{\mathcal I}^R_{\mathcal F} (\Gamma)\) of all gauge theory graphs of pure Yang-Mills theory which can be achieved from \(\Gamma\) via graph and cycle homology \cite{Kreimer_Sars_Suijlekom} is then obtained via the Corolla differential \(\mathcal D_{\text{\emph{QCD}}} (\Gamma)\) (cf. \defnref{defn:corolla_differential}) and the renormalized parametric integrand with non-standard second Symanzik polynomial \(I^R (\Gamma)\) (cf. \thmref{thm:parametric_integrand_non-standard} and again \cite{Kreimer_Sars_Suijlekom,Sars,Kreimer,Peskin_Schroeder,Itzykson_Zuber,Suijlekom,Brown_Kreimer} for the huge topic of renormalization)
\begin{equation}
	\tilde{I}_{\mathcal F}^R (\Gamma) = \mathcal D_{\text{\emph{QCD}}} (\Gamma) I^R (\Gamma)
\end{equation}
by (where the hat over a variable \(\widehat{A_e}\) means that we omit this particular variable \(A_e\))
\begin{equation}
	\tilde{\mathcal I}^R_{\mathcal F} (\Gamma) = \frac{1}{\operatorname{sym} (\Gamma)} \sum_{k=0}^\infty \sum_{\set{e_1, \ldots, e_k} \subset \Gamma^{[1]}_{\text{\emph{int}}}} \int \dif \underline{A}_{\Gamma \setminus \set{e_1, \ldots, e_k}} \underset{A_1, \ldots, \widehat{A_{e_1}}, \ldots, \widehat{A_{e_k}}, \ldots, A_{e_{\left \vert \Gamma^{[1]} \right \vert}} = 0}{\operatorname{Reg}} \underset{A_{e_1}, \ldots, A_{e_k} = 0}{\operatorname{Res}} \tilde{I}_{\mathcal F}^R (\Gamma) \, .
\end{equation}
\end{thm}

\begin{proof}
We refer to \cite[Page 31]{Kreimer_Sars_Suijlekom} for a proof.
\end{proof}

\vspace{\baselineskip}

\begin{exmp}
We continue with \exref{exmp:1L_Sc_Corolla_polynomial}:

Again, we have
\begin{equation}
\Gamma := 
\vcenter{\hbox{\includegraphics[height=2.5cm]{1L_Sc_labeled.eps}}}
\end{equation}
and
\begin{equation}
\begin{split}
	\mathcal C (\Gamma) = & \mathcal C^0 (\Gamma) - \mathcal C^1 (\Gamma) \\ = & \left ( b_{h_{\alpha}} \left ( a_{h_{\alpha +}} + a_{h_{\alpha -}} \right ) + b_{h_{\beta}} \left ( a_{h_{\beta +}} + a_{h_{\beta -}} \right ) + b_{h_{\gamma}} \left ( a_{h_{\gamma +}} + a_{h_{\gamma -}} \right ) \right ) \\ & \phantom{ ( } \times \left ( b_{h_{\delta}} \left ( a_{h_{\delta +}} + a_{h_{\delta -}} \right ) + b_{h_{\epsilon}} \left ( a_{h_{\epsilon +}} + a_{h_{\epsilon -}} \right ) + b_{h_{\zeta}} \left ( a_{h_{\zeta +}} + a_{h_{\zeta -}} \right )  \right ) \\ & - a_{h_{\alpha +}} a_{h_{\delta +}} - a_{h_{\alpha -}} a_{h_{\delta -}}
\end{split}
\end{equation}

Choosing an anti-clockwise orientation for our embedding oriented manifold, external momenta \(p^{\mu_3}\) and \(p^{\mu_4}\) incoming and also an anti-clockwise orientation for the loop momenta, we receive:
{\allowdisplaybreaks
\begin{subequations}
\begin{alignat}{2}
	& \mathscr A _{h_{\alpha +}} := \frac{1}{2 A_2} \frac{\partial}{\partial \xi_{2 \mu_3}} \, ,   \qquad && \mathscr A _{h_{\alpha -}} := \frac{1}{2 A_1} \frac{\partial}{\partial \xi_{1 \mu_3}} \, , \\
	& \mathscr A _{h_{\beta +}} := - \frac{1}{2 A_1} \frac{\partial}{\partial \xi_{1 \mu_2}} \, , \qquad && \mathscr A _{h_{\beta -}} := \frac{1}{2 A_3} \frac{\partial}{\partial \xi_{3 \mu_2}} \, , \\
	& \mathscr A _{h_{\gamma +}} := - \frac{1}{2 A_3} \frac{\partial}{\partial \xi_{3 \mu_1}} \, , \qquad && \mathscr A _{h_{\gamma -}} := - \frac{1}{2 A_2} \frac{\partial}{\partial \xi_{2 \mu_1}} \, , \\
	& \mathscr A _{h_{\delta +}} := \frac{1}{2 A_1} \frac{\partial}{\partial \xi_{1 \mu_4}} \, ,   \qquad && \mathscr A _{h_{\delta -}} := \frac{1}{2 A_2} \frac{\partial}{\partial \xi_{2 \mu_4}} \, , \\
	& \mathscr A _{h_{\epsilon +}} := - \frac{1}{2 A_2} \frac{\partial}{\partial \xi_{2 \mu_1}} \, , \qquad && \mathscr A _{h_{\epsilon -}} := \frac{1}{2 A_4} \frac{\partial}{\partial \xi_{4 \mu_1}} \, , \\
	& \mathscr A _{h_{\zeta +}} := - \frac{1}{2 A_4} \frac{\partial}{\partial \xi_{4 \mu_2}} \, , \qquad && \mathscr A _{h_{\zeta -}} := - \frac{1}{2 A_1} \frac{\partial}{\partial \xi_{1 \mu_2}} \, ,
\end{alignat}
\begin{align}
	& \mathscr B _{h_{\alpha}} := \eta^{\mu_2 \mu_1} \, , \\
	& \mathscr B _{h_{\beta}} := \eta^{\mu_1 \mu_3} \, , \\
	& \mathscr B _{h_{\gamma}} := \eta^{\mu_3 \mu_2} \, , \\
	& \mathscr B _{h_{\delta}} := \eta^{\mu_1 \mu_2} \, , \\
	& \mathscr B _{h_{\epsilon}} := \eta^{\mu_2 \mu_4} \, , \\
	& \mathscr B _{h_{\zeta}} := \eta^{\mu_4 \mu_1}
\end{align}
\end{subequations}
}

and
\begin{equation}
\begin{split}
	\mathcal D (\Gamma) = & \left ( \mathscr B_{h_{\alpha}} \left ( \mathscr A_{h_{\alpha +}} + \mathscr A_{h_{\alpha -}} \right ) + \mathscr B_{h_{\beta}} \left ( \mathscr A_{h_{\beta +}} + \mathscr A_{h_{\beta -}} \right ) + \mathscr B_{h_{\gamma}} \left ( \mathscr A_{h_{\gamma +}} + \mathscr A_{h_{\gamma -}} \right ) \right ) \\ & \phantom{ ( } \times \left ( \mathscr B_{h_{\delta}} \left ( \mathscr A_{h_{\delta +}} + \mathscr A_{h_{\delta -}} \right ) + \mathscr B_{h_{\epsilon}} \left ( \mathscr A_{h_{\epsilon +}} + \mathscr A_{h_{\epsilon -}} \right ) + \mathscr B_{h_{\zeta}} \left ( \mathscr A_{h_{\zeta +}} + \mathscr A_{h_{\zeta -}} \right )  \right ) \\ & - \mathscr A_{h_{\alpha +}} \mathscr A_{h_{\delta +}} - \mathscr A_{h_{\alpha -}} \mathscr A_{h_{\delta -}} \\ = & \left ( \eta^{\mu_2 \mu_1} \left ( \frac{1}{2 A_2} \frac{\partial}{\partial \xi_{2 \mu_3}} + \frac{1}{2 A_1} \frac{\partial}{\partial \xi_{1 \mu_3}} \right ) + \eta^{\mu_1 \mu_3} \left ( - \frac{1}{2 A_1} \frac{\partial}{\partial \xi_{1 \mu_2}} + \frac{1}{2 A_3} \frac{\partial}{\partial \xi_{3 \mu_2}} \right ) \right . \\ & \phantom{ ( } \left . + \eta^{\mu_3 \mu_2} \left ( - \frac{1}{2 A_3} \frac{\partial}{\partial \xi_{3 \mu_1}} - \frac{1}{2 A_2} \frac{\partial}{\partial \xi_{2 \mu_1}} \right ) \right ) \times \left ( \eta^{\mu_1 \mu_2} \left ( \frac{1}{2 A_1} \frac{\partial}{\partial \xi_{1 \mu_4}} + \frac{1}{2 A_2} \frac{\partial}{\partial \xi_{2 \mu_4}} \right ) \right . \\ & \phantom{ ( } \left . + \eta^{\mu_2 \mu_4} \left ( - \frac{1}{2 A_2} \frac{\partial}{\partial \xi_{2 \mu_1}} + \frac{1}{2 A_4} \frac{\partial}{\partial \xi_{4 \mu_1}} \right ) + \eta^{\mu_4 \mu_1} \left ( - \frac{1}{2 A_4} \frac{\partial}{\partial \xi_{4 \mu_2}} - \frac{1}{2 A_1} \frac{\partial}{\partial \xi_{1 \mu_2}} \right ) \right ) \\ & - \frac{1}{4 A_2 A_1} \frac{\partial}{\partial \xi_{2 \mu_3}} \frac{\partial}{\partial \xi_{1 \mu_4}} - \frac{1}{4 A_1 A_2} \frac{\partial}{\partial \xi_{1 \mu_3}} \frac{\partial}{\partial \xi_{2 \mu_4}}
\, .
\end{split}
\end{equation}

Acting with this differential operator on the scalar QFT parametric integrand with non-standard second Symanzik polynomial \(I (\Gamma)\) (cf. \exref{exmp:1L-parametric_integrand}), we receive \cite[Example 5.15]{Sars}
\begin{equation}
\begin{split}
	\tilde{I}_{\mathcal F} (\Gamma) = & \mathcal D_{\text{QCD}} (\Gamma) I (\Gamma) \\ = & \underline{g_s}^2 f^{c_3 c_2 c_1} f^{c_4 c_2 c_1} \left ( \left ( q^{\mu_3} q^{\mu_4} \left ( 2 A_1^2 + 2 A_2^2 + 12 A_1 A_2 \right ) \right . \right . \\ & \left . \left . \phantom{ \underline{g_s}^2 f^{a_3 a_2 a_1 } f^{a_4 a_2 a_1} \left ( \left ( \right . \right .} - q^2 \eta^{\mu_3 \mu_4} \left ( 5 A_1^2 + 5 A_2^2 + 8 A_1 A_2 \right ) \right ) \frac{1}{\psi_\Gamma^2} 8 + \eta^{\mu_3 \mu_4} \frac{1}{\psi_\Gamma} \right ) I (\Gamma) \, .
\end{split}
\end{equation}

Then, we obtain (after renormalizing the parametric integrand \(I (\Gamma)\)) the renormalized Feynman amplitude with renormalization point \(\mu\) \cite[Example 5.15]{Sars}
\begin{equation}
	\tilde{\mathcal I}^R_{\mathcal F} (\Gamma) = \frac{10}{3} \underline{g_s}^2 f^{c_3 c_2 c_1} f^{c_4 c_2 c_1} \left ( -q^{\mu_3} q^{\mu_4} + q^2 \eta^{\mu_3 \mu_4} \right ) \ln \frac{q^2}{\mu^2}
\end{equation}
which is transversal, as desired.
\end{exmp}

\vspace{\baselineskip}

\begin{rem}
The full \(m\)-loop scattering amplitude can be obtained by applying the Corolla polynomial componentwise to the combinatorial Green's function for connected graphs\footnote{We emphasize that we really mean 1-connected graphs, not only 2-connected (as they are called in mathematics) or 1 PI (as they are called in physics) graphs.}, with loop order \(m\) \cite{Kreimer_Sars_Suijlekom}.
\end{rem}

\subsection{Inclusion of the Gauge Bosons of the Electroweak Sector} \label{ssec:corolla_polynomial_gauge_bosons_electroweak}

The inclusion of the gauge bosons and their corresponding ghosts of the electroweak sector of the Standard Model shows some interesting properties and is therefore the next step to adapt the Corolla polynomial to all bosons of the electroweak sector of the Standard Model. Concretely, we have to add the massive gauge bosons \(W^\pm\) and \(Z\) with their corresponding ghosts \(c^\pm\) and \(c_Z\) and the photon \(A\) with its corresponding ghost \(c_A\).

\vspace{\baselineskip}

\begin{out}
First, recall that we have \(W^\pm\)-particle conservation which implies that every vertex of the electroweak sector of the Standard Model has to consist of a positive even number of \(W^\pm\)-particles, in particular every 3-valent vertex must consist of two \(W^\pm\)-particles. This is precisely the definition of a 2-factor, extended to external edges of a scalar QFT Feynman graph. Therefore, we define two nested sums over a 3-regular scalar QFT Feynman graph \(\Gamma\) in the following way: The first sum runs over all 2-factors of \(\Gamma\), thus creating all possible ways to attach valid \(W^\pm\)-particle labelings to \(\Gamma\). Then, the second sum creates all possible ways to attach \(Z\)-labels or \(A\)-labels to the unlabeled edges of the \(W^\pm\)-labeled graphs. The full details are explained in the proof after \thmref{thm:corolla_polynomial-electroweak_gauge_bosons}.
\end{out}

\vspace{\baselineskip}

Therefore, we define:

\vspace{\baselineskip}

\begin{defn} \label{defn:F_2_P_Z}
	Let \(\Gamma\) be a 3-regular graph. Then:
\begin{enumerate}
\item Let \(\wp \left ( \Gamma^{[1]} \right )\) be the power set of all external edges and internal edges of \(\Gamma\) with all adjacent vertices added. In particular, we define\footnote{Recall \defnref{defn:factor} for the definition of a 2-factor of a scalar QFT Feynman graph.} \(\mathscr P_Z \left ( \mathfrak f \right ) := \wp \left ( \Gamma^{[1]} \setminus \mathfrak f^{[1]} \right )\), i.e. the power set of all external edges and internal edges of \(\Gamma\) which are not in the set \(\mathfrak f\). We denote the elements of \( \mathscr P_Z \left (\mathfrak f \right ) \) by \(P_{Z}\). The edges in each set \(P_Z\) in the set of sets \(\mathscr P_Z \left ( \mathfrak f \right )\) are labeled by a \(Z\)-label and the edges in \(\Gamma \setminus \left ( \mathfrak f \cup P_Z \right )\) are labeled by an \(A\)-label.
\item Let \(\left \vert P_Z^{[0]} \right \vert\) and \(\left \vert \Gamma^{[0]} \setminus P_Z^{[0]} \right \vert\) denote the number of vertices in \(P_Z\) and \(\Gamma \setminus P_Z\), respectively.
\item Let \(\operatorname{iso} \left ( \Gamma_{\text{labeled}} \right )\) be the number of labeled graphs (via the sets \(\mathfrak f\) and \(P_Z\)) in the set \(\set{P_Z \in \mathscr P_Z (\mathfrak f) \vert \mathfrak f \in \mathfrak F_2 \left ( \Gamma \right ) }\) isomorphic to \(\Gamma_{\text{labeled}}\).
\end{enumerate}
\end{defn}

\vspace{\baselineskip}

\begin{exmp} \label{exmp:1L_sc-PW_PZ}
In particular, for the one-loop self-energy graph
\begin{equation}
	\Gamma := \vcenter{\hbox{\includegraphics[height=2.5cm]{1L_Sc_labeled.eps}}}
\end{equation}
we have:
\begin{equation}
\begin{split}
	\mathfrak F_2 \left (
\underbrace{\vcenter{\hbox{\includegraphics[height=2cm]{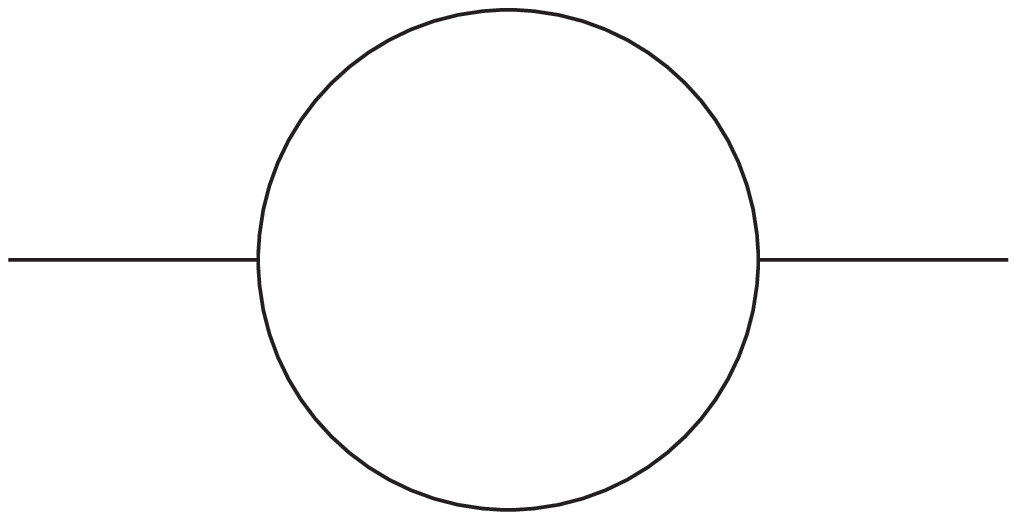}}}}_{=: \Gamma}
\right )
= & \left \{
\underbrace{\vcenter{\hbox{\includegraphics[height=2cm]{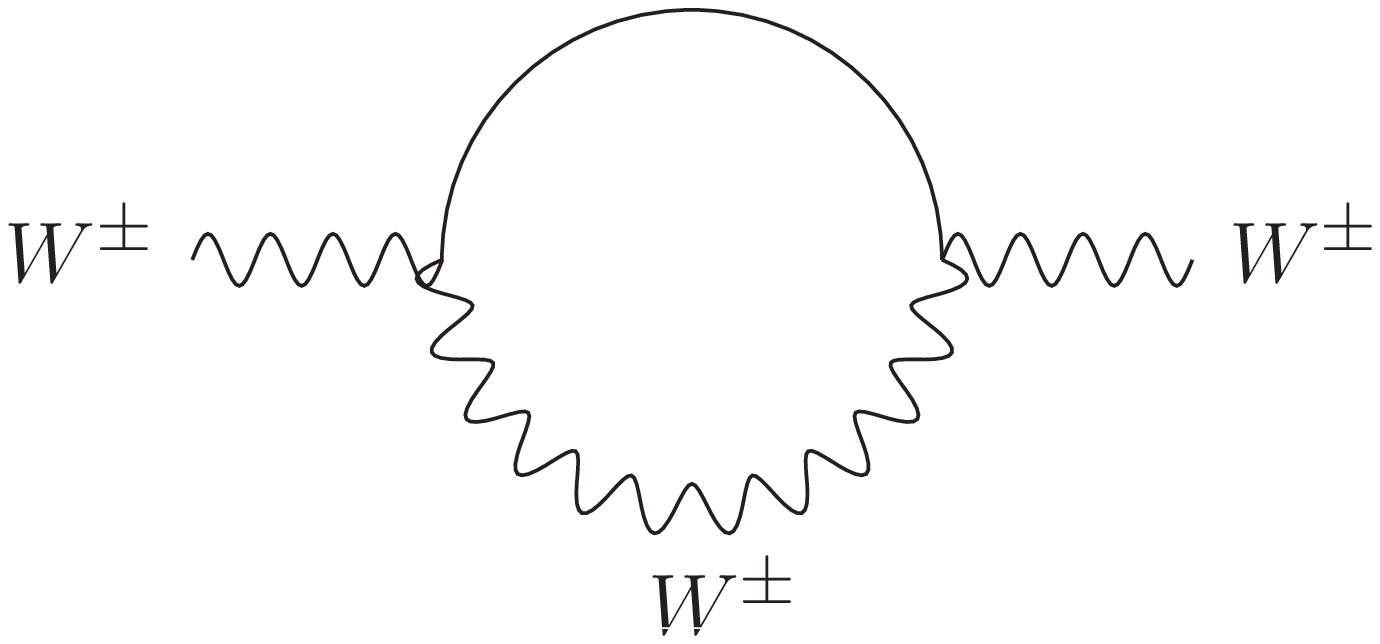}}}}_{=: \mathfrak f^{(1)} (\Gamma)} , \right . \\
& \phantom{ \{ } \left . \underbrace{\vcenter{\hbox{\includegraphics[height=2cm]{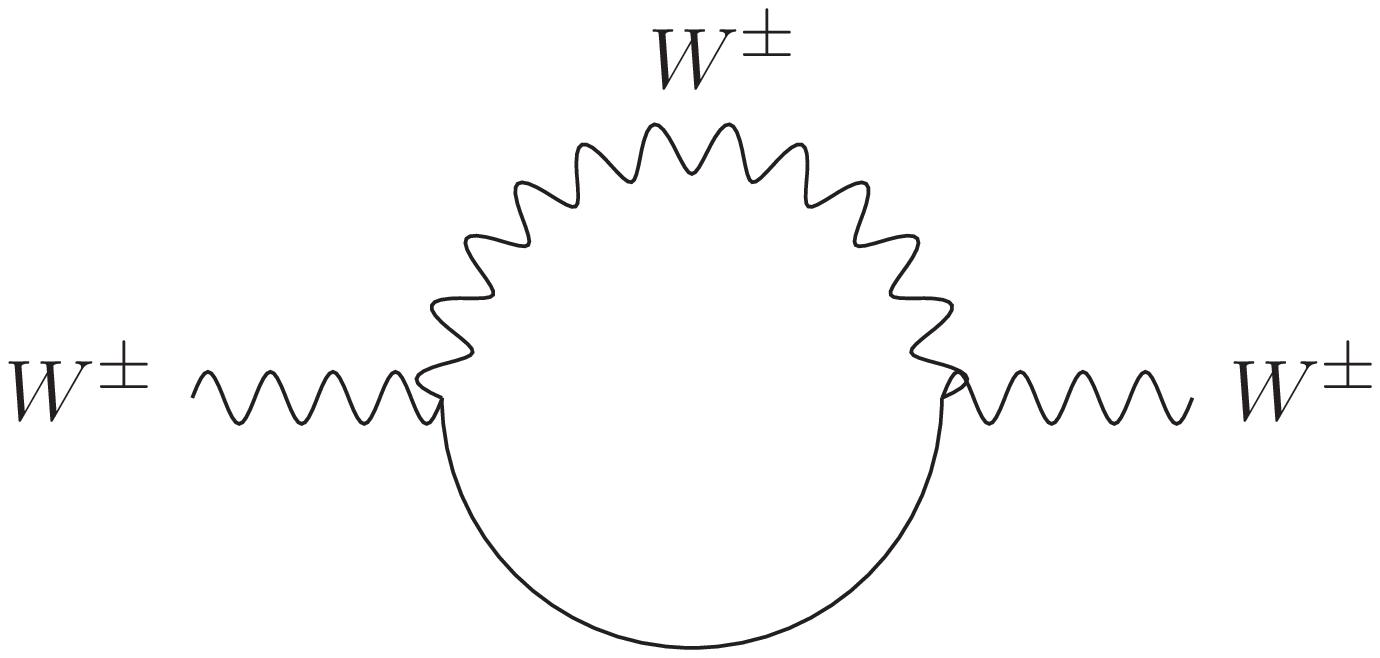}}}}_{=: \mathfrak f^{(2)} (\Gamma)} ,
\underbrace{\vcenter{\hbox{\includegraphics[height=2cm]{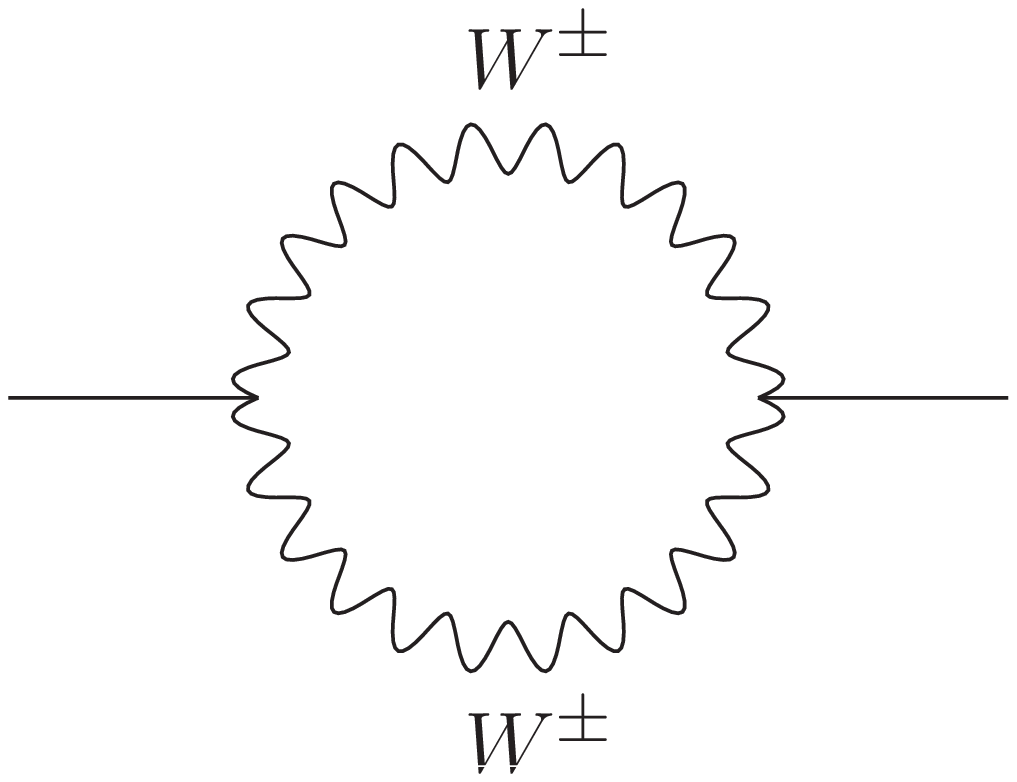}}}}_{=: \mathfrak f^{(3)} (\Gamma)}
\right \}
\end{split}
\end{equation}
and
\begin{subequations}
\begin{equation}
	\mathscr P_Z \left (
\underbrace{\vcenter{\hbox{\includegraphics[height=2cm]{1L_Sc_PW1.eps}}}}_{=: \mathfrak f^{(1)}}
\right )
= \left \{
\underbrace{\vcenter{\hbox{\includegraphics[height=2cm]{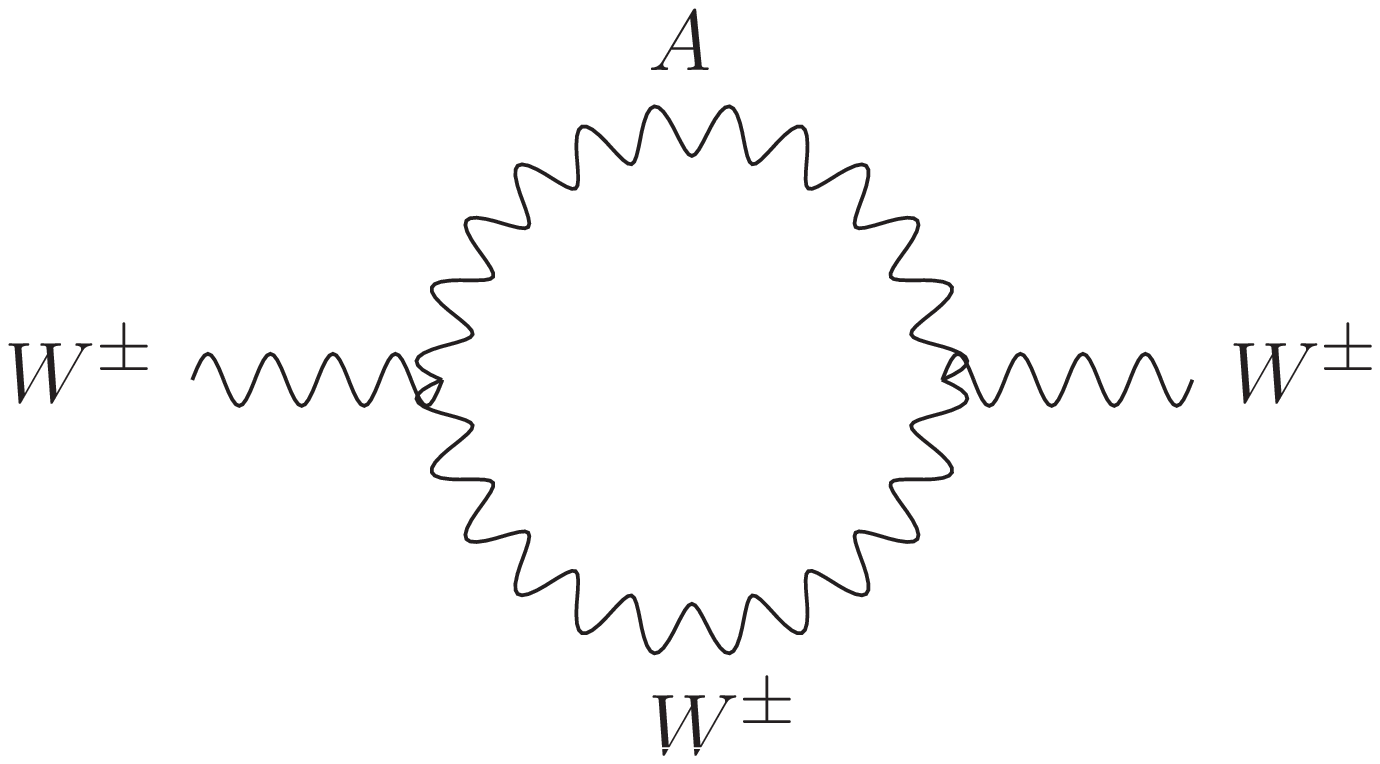}}}}_{=: P_Z^{(1)} \left ( \mathfrak f^{(1)} \right )} ,
\underbrace{\vcenter{\hbox{\includegraphics[height=2cm]{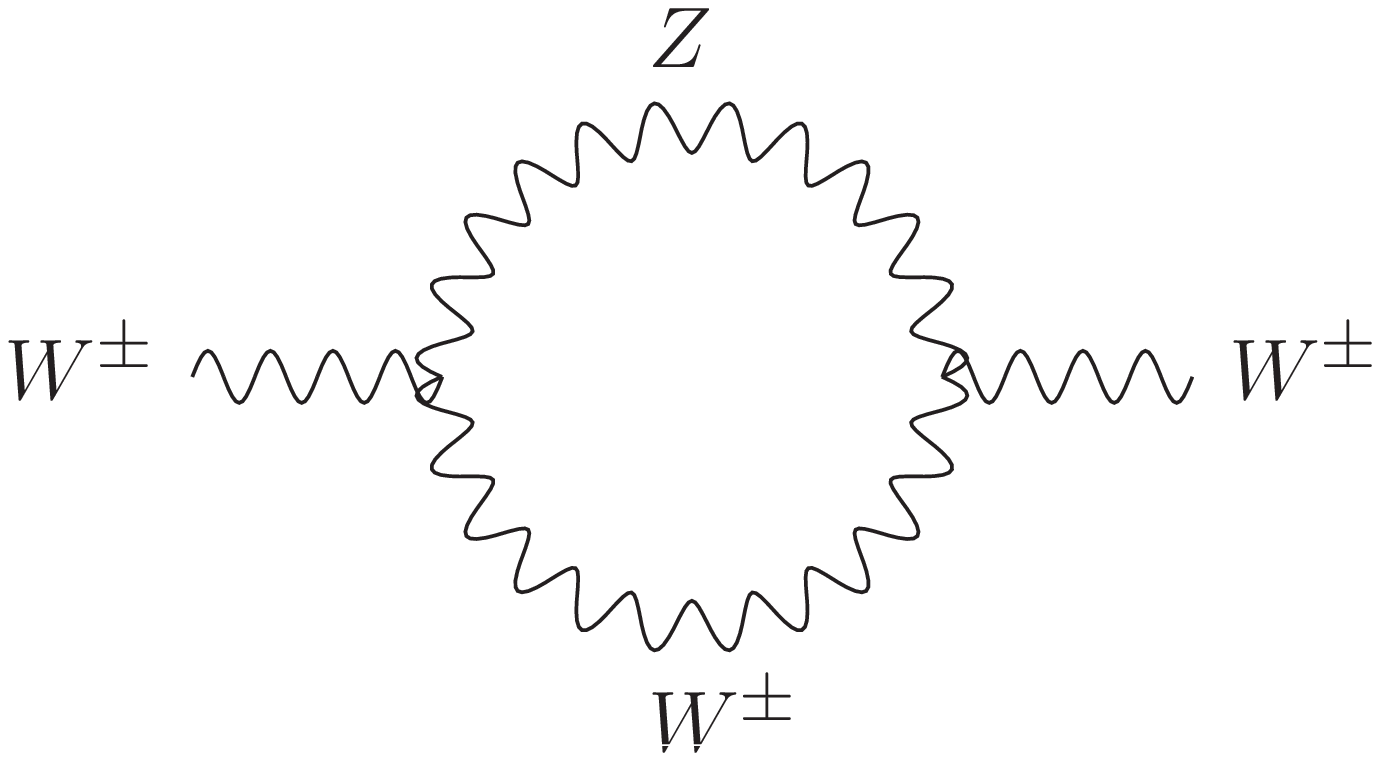}}}}_{=: P_Z^{(2)} \left ( \mathfrak f^{(1)} \right )}
\right \}
\end{equation}
\begin{equation}
	\mathscr P_Z \left (
\underbrace{\vcenter{\hbox{\includegraphics[height=2cm]{1L_Sc_PW2.eps}}}}_{=: \mathfrak f^{(2)}}
\right )
= \left \{
\underbrace{\vcenter{\hbox{\includegraphics[height=2cm]{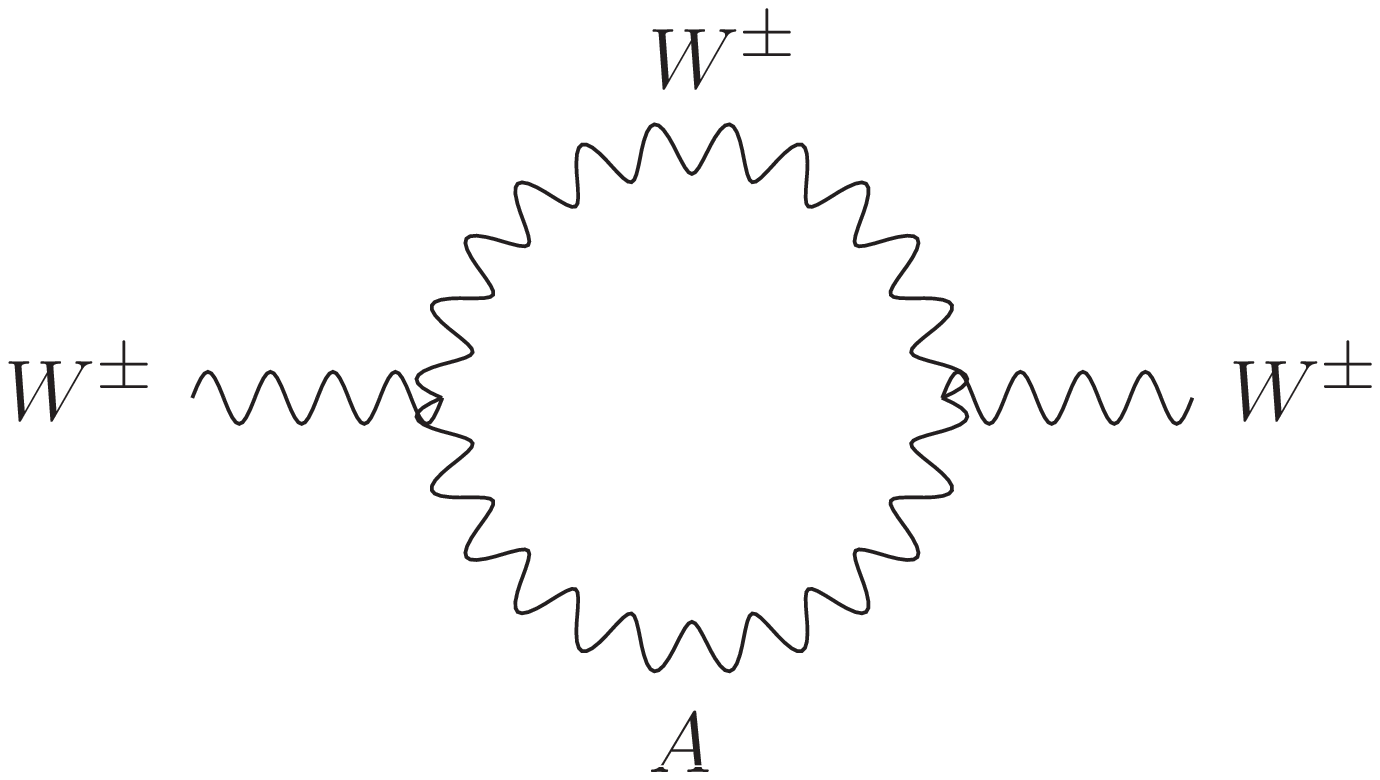}}}}_{=: P_Z^{(1)} \left ( \mathfrak f^{(2)} \right )} ,
\underbrace{\vcenter{\hbox{\includegraphics[height=2cm]{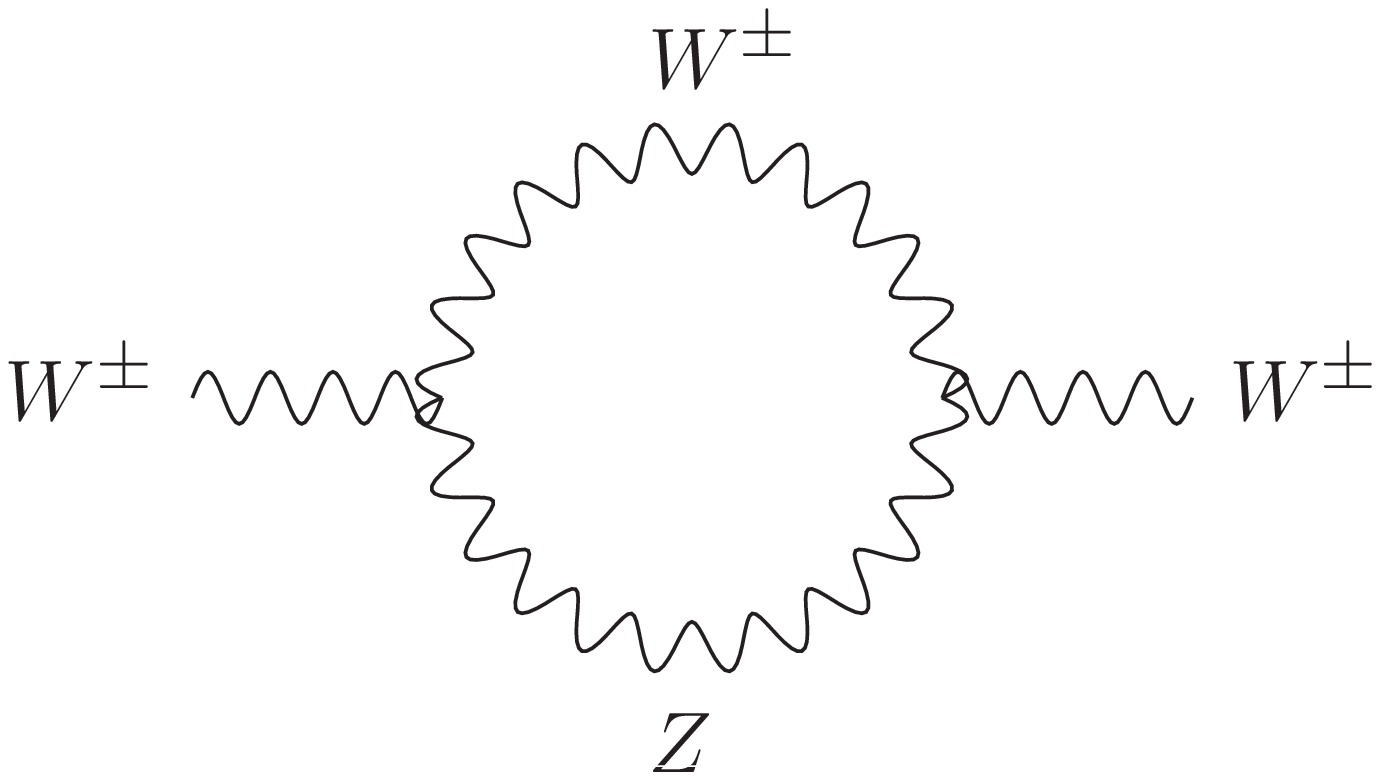}}}}_{=: P_Z^{(2)} \left ( \mathfrak f^{(2)} \right )}
\right \}
\end{equation}
\begin{equation}
\begin{split}
	\mathscr P_Z \left (
\underbrace{\vcenter{\hbox{\includegraphics[height=2cm]{1L_Sc_PW3.eps}}}}_{=: \mathfrak f^{(3)}}
\right )
= & \left \{
\underbrace{\vcenter{\hbox{\includegraphics[height=2cm]{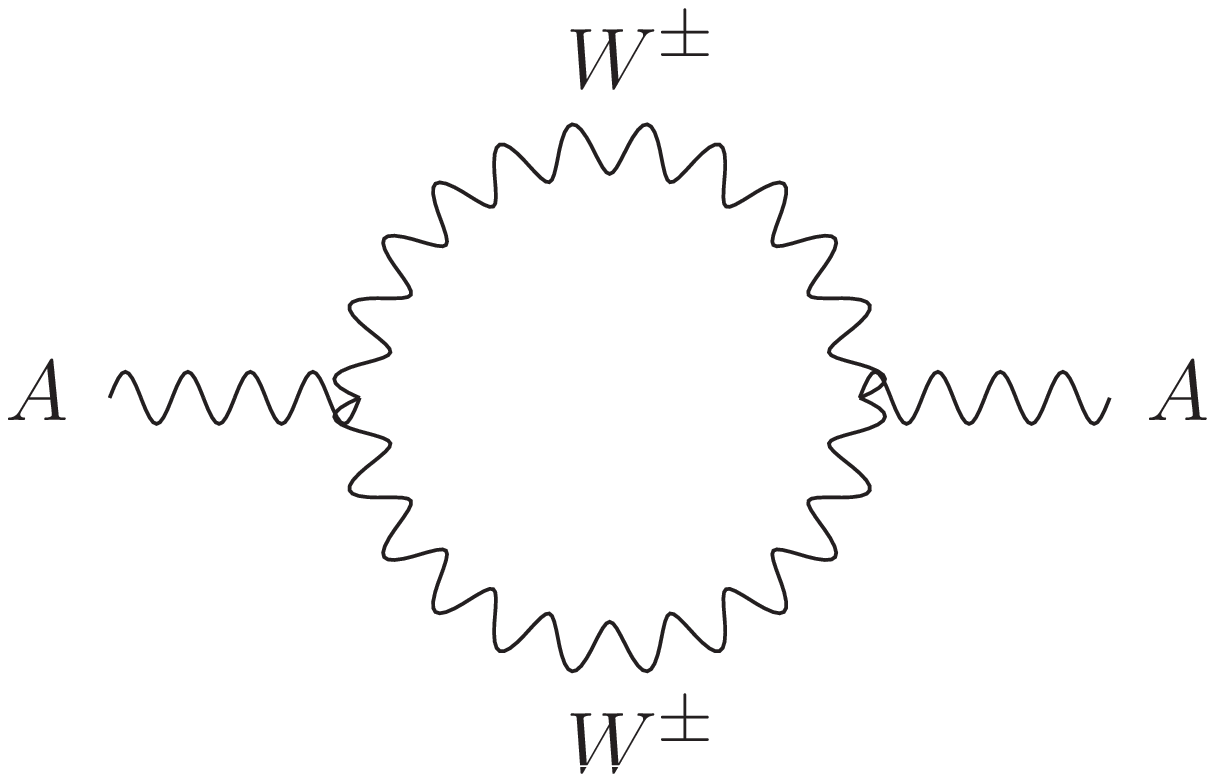}}}}_{=: P_Z^{(1)} \left ( \mathfrak f^{(3)} \right )} ,
\underbrace{\vcenter{\hbox{\includegraphics[height=2cm]{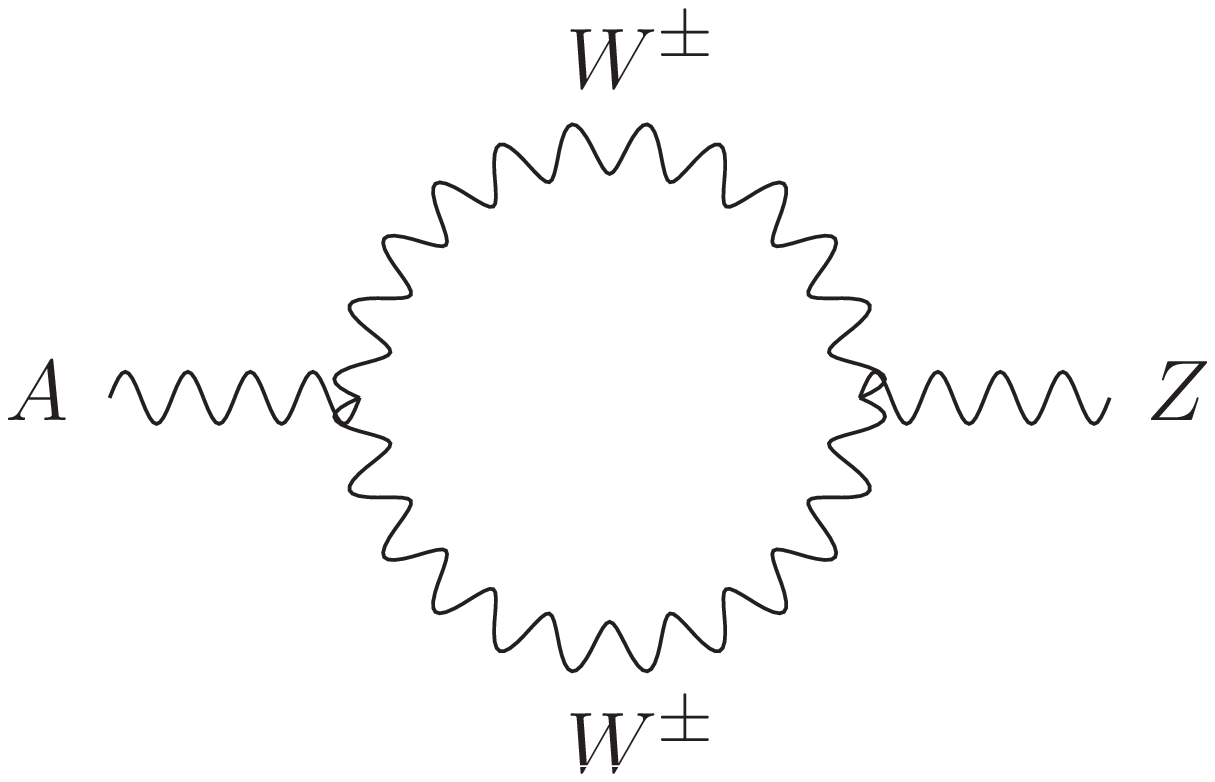}}}}_{=: P_Z^{(2)} \left ( \mathfrak f^{(3)} \right )} \right . ,\\
& \phantom{ \{ } \left . \underbrace{\vcenter{\hbox{\includegraphics[height=2cm]{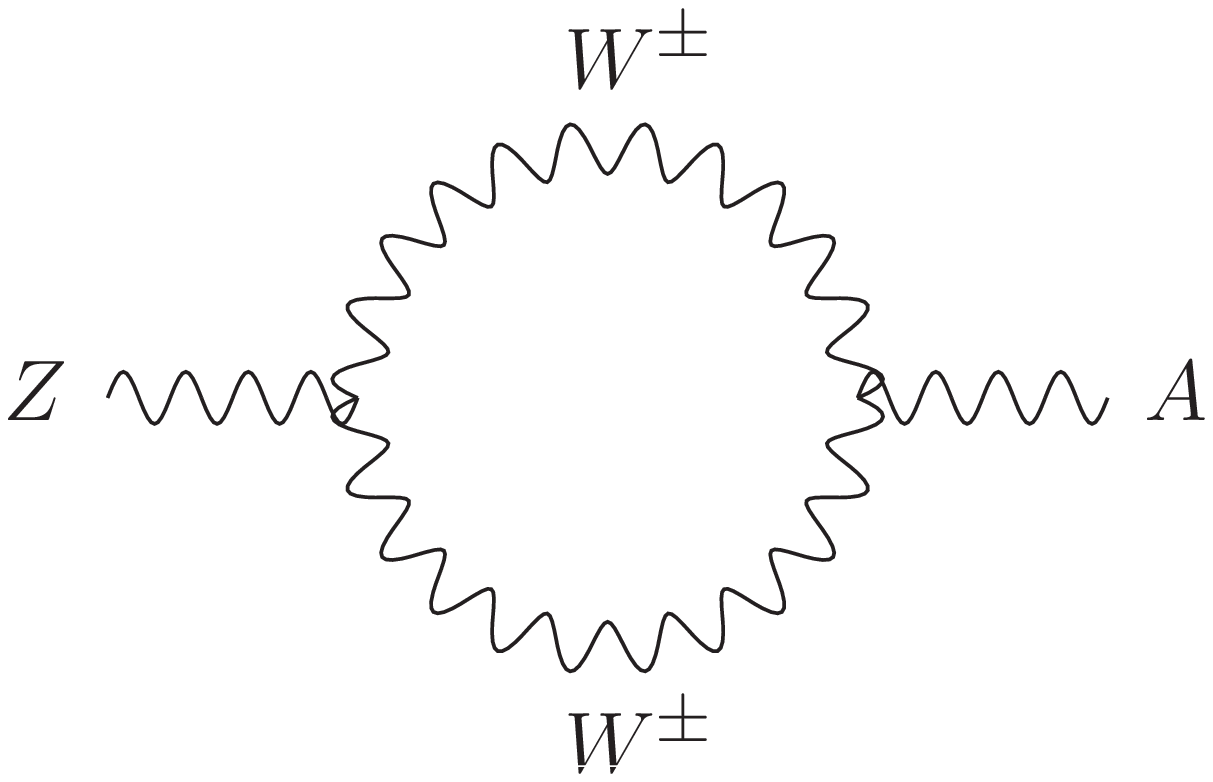}}}}_{=: P_Z^{(3)} \left ( \mathfrak f^{(3)} \right )} ,
\underbrace{\vcenter{\hbox{\includegraphics[height=2cm]{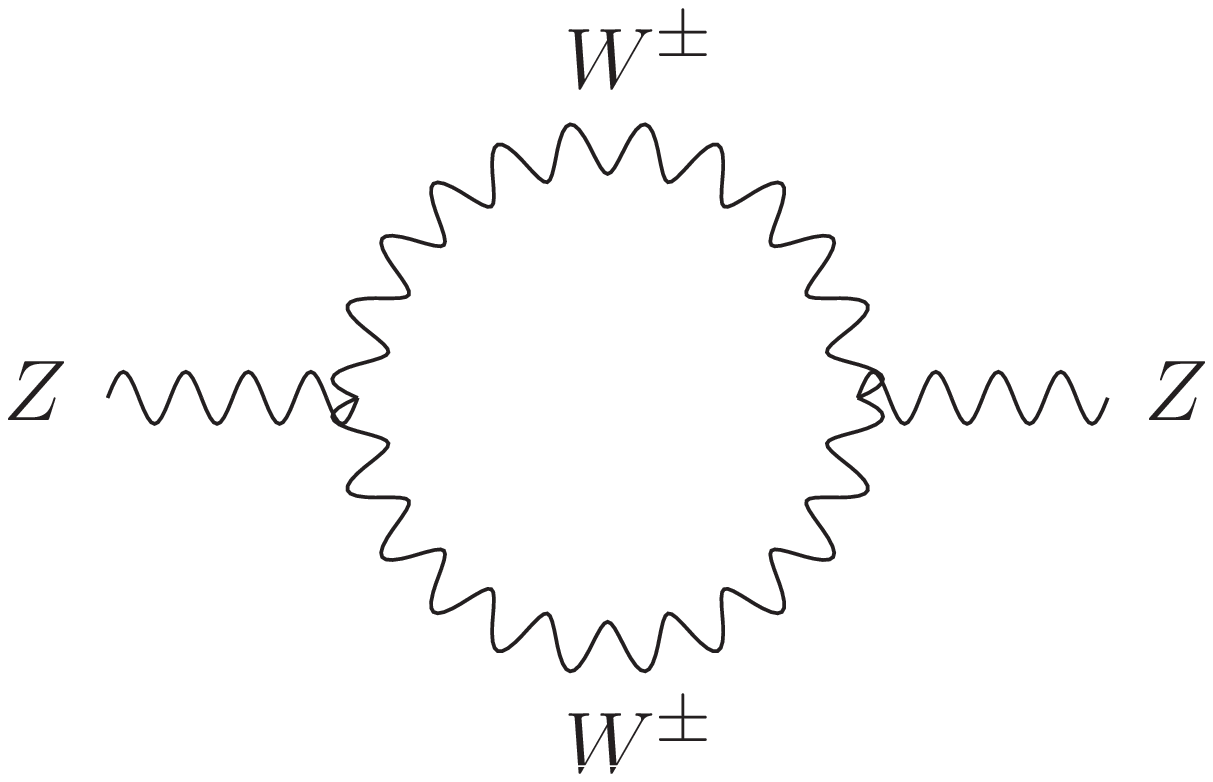}}}}_{=: P_Z^{(4)} \left ( \mathfrak f^{(3)} \right )}
\right \}
\end{split}
\end{equation}
\end{subequations}
Moreover, the above defined factors read
\begin{subequations}
\begin{equation}
\begin{split}
	\operatorname{sym} \left ( \vcenter{\hbox{\includegraphics[height=2cm]{1L_Sc.eps}}} \right ) = & \operatorname{sym} \left ( \vcenter{\hbox{\includegraphics[height=2cm]{1L_Sc_PW3_PZ1.eps}}} \right ) \\ = & \operatorname{sym} \left ( \vcenter{\hbox{\includegraphics[height=2cm]{1L_Sc_PW3_PZ2.eps}}} \right ) \\ = & \operatorname{sym} \left ( \vcenter{\hbox{\includegraphics[height=2cm]{1L_Sc_PW3_PZ3.eps}}} \right ) \\ = & \operatorname{sym} \left ( \vcenter{\hbox{\includegraphics[height=2cm]{1L_Sc_PW3_PZ4.eps}}} \right ) = 2 \, ,
\end{split}
\end{equation}
\begin{equation}
\begin{split}
	\operatorname{sym} \left ( \vcenter{\hbox{\includegraphics[height=2cm]{1L_Sc_PW1_PZ1.eps}}} \right ) = & \operatorname{sym} \left ( \vcenter{\hbox{\includegraphics[height=2cm]{1L_Sc_PW1_PZ2.eps}}} \right ) \\ = & \operatorname{sym} \left ( \vcenter{\hbox{\includegraphics[height=2cm]{1L_Sc_PW2_PZ1.eps}}} \right ) \\ = & \operatorname{sym} \left ( \vcenter{\hbox{\includegraphics[height=2cm]{1L_Sc_PW2_PZ2.eps}}} \right ) = 1 \, ,
\end{split}
\end{equation}
\begin{equation}
\begin{split}
	\operatorname{iso} \left ( \vcenter{\hbox{\includegraphics[height=2cm]{1L_Sc_PW1_PZ1.eps}}} \right ) = & \operatorname{iso} \left ( \vcenter{\hbox{\includegraphics[height=2cm]{1L_Sc_PW1_PZ2.eps}}} \right ) \\ = & \operatorname{iso} \left ( \vcenter{\hbox{\includegraphics[height=2cm]{1L_Sc_PW2_PZ1.eps}}} \right ) \\ = & \operatorname{iso} \left ( \vcenter{\hbox{\includegraphics[height=2cm]{1L_Sc_PW2_PZ2.eps}}} \right ) = 2
\end{split}
\end{equation}
and
\begin{equation}
\begin{split}
	\operatorname{iso} \left ( \vcenter{\hbox{\includegraphics[height=2cm]{1L_Sc_PW3_PZ1.eps}}} \right ) = & \operatorname{iso} \left ( \vcenter{\hbox{\includegraphics[height=2cm]{1L_Sc_PW3_PZ2.eps}}} \right ) \\ = & \operatorname{iso} \left ( \vcenter{\hbox{\includegraphics[height=2cm]{1L_Sc_PW3_PZ3.eps}}} \right ) \\ = & \operatorname{iso} \left ( \vcenter{\hbox{\includegraphics[height=2cm]{1L_Sc_PW3_PZ4.eps}}} \right ) = 1 \, ,
\end{split}
\end{equation}
such that
\begin{equation}
	\frac{\operatorname{sym} \left ( \Gamma \right )}{\operatorname{sym} \left ( \Gamma_{\text{labeled}} \right ) \operatorname{iso} \left ( \Gamma_{\text{labeled}} \right )} = 1 \, , \qquad \forall P_Z \in \mathscr P_Z \left ( \mathfrak f \right ) \, , \quad \forall \mathfrak f \in \mathfrak F_2 \left ( \Gamma \right ) \, ,
\end{equation}
\end{subequations}
where we remind the reader that we consider the one-loop self-energy graph
\begin{equation}
	\Gamma := \vcenter{\hbox{\includegraphics[height=2.5cm]{1L_Sc_labeled.eps}}} \, .
\end{equation}
\end{exmp}

Then, we have:

\vspace{\baselineskip}

\begin{thm}[\cite{Prinz}] \label{thm:corolla_polynomial-electroweak_gauge_bosons}
Let \(\Gamma\) be a 3-regular scalar QFT Feynman graph. Then the gauge bosons and their corresponding ghosts of the electroweak sector of the Standard Model can be included into the definition of the Corolla polynomial (\defnref{defn:corolla_polynomial}) by multiplying the parametric integrand \(I (\Gamma)\) of the corresponding scalar QFT amplitude by \(J (\Gamma; m_W, m_Z)\), defined as
\begin{equation} \label{eqn:j}
\begin{split}
	J \left ( \Gamma; m_W, m_Z \right ) := & \left ( \vphantom{\phantom{ ( + } \times i^{\left \vert \Gamma^{[0]} \right \vert + \left \vert \Gamma^{[1]} \right \vert} \left ( \underline{g} \cos \theta_W \right ) ^{\left \vert P_Z^{[0]} \right \vert} \underline{e}^{\left \vert \Gamma^{[0]} \setminus P_Z^{[0]} \right \vert} e^{-\left ( \sum_{e \in \mathfrak f^{[1]}} A_e m_W^2 + \sum_{e \in P_Z^{[1]}} A_e m_Z^2 \right )}} i^{\left \vert \Gamma^{[1]} \right \vert} \underline{g_s}^{\left \vert \Gamma^{[0]} \right \vert} \operatorname{color} (\Gamma) \right . \\ & \left . \phantom{ ( } + \sum_{\mathfrak f \in \mathfrak F_2 \left ( \Gamma \right )} \sum_{P_Z \in \mathscr{P}_Z (\mathfrak f)} \frac{\operatorname{sym} \left ( \Gamma \right )}{\operatorname{sym} \left ( \Gamma_{\text{\emph{labeled}}} \right ) \operatorname{iso} \left ( \Gamma_{\text{\emph{labeled}}} \right )} \right . \\ & \left . \phantom{ ( + } \times i^{\left \vert \Gamma^{[0]} \right \vert + \left \vert \Gamma^{[1]} \right \vert} \left ( \underline{g} \cos \theta_W \right ) ^{\left \vert P_Z^{[0]} \right \vert} \underline{e}^{\left \vert \Gamma^{[0]} \setminus P_Z^{[0]} \right \vert} e^{-\left ( \sum_{e \in \mathfrak f^{[1]}} A_e m_W^2 + \sum_{e \in P_Z^{[1]}} A_e m_Z^2 \right )} \right ) \, ,
\end{split}
\end{equation}
such that the corresponding parametric integrand for gauge theory including gluons, ghosts, electroweak gauge bosons and their corresponding ghosts \(\tilde{I}_{\mathcal F} (\Gamma)\) can be written as\footnote{Note that here the Corolla differential \(\mathcal D (\Gamma)\) is used instead of \(\mathcal D_{\text{QCD}} (\Gamma)\), since the additional constants are moved into the definition of \(J (\Gamma; m_W, m_Z)\).}
\begin{equation} \label{eqn:amplitude_electroweak_gauge_bosons}
	\tilde{I}_{\mathcal F} (\Gamma) = \mathcal D (\Gamma) I (\Gamma) J \left ( \Gamma; m_W, m_Z \right ) \, .
\end{equation}
\end{thm}

\begin{proof}
Recall that the Feynman rules do not allow gluons to couple to the gauge bosons of the electroweak sector of the Standard Model. Therefore, we can consider these two parts of the Standard Model separately. Additionally note that the Feynman rules for the gauge bosons of the electroweak sector of the Standard Model have precisely the same tensor structure as the Feynman rules for the gluons. The only difference lies in different factors of \(i\), different coupling constants, the structure constants for the gluons and the mass terms for the massive gauge bosons and their corresponding ghosts of the electroweak sector of the Standard Model. These are taken in account by the factors \(i^{\left \vert \Gamma^{[1]} \right \vert} \underline{g_s}^{\left \vert \Gamma^{[0]} \right \vert} \operatorname{color} (\Gamma)\) and \(i^{\left \vert \Gamma^{[0]} \right \vert + \left \vert \Gamma^{[1]} \right \vert} \left ( \underline{g} \cos \theta_W \right ) ^{\left \vert P_Z^{[0]} \right \vert} \underline{e}^{\left \vert \Gamma^{[0]} \setminus P_Z^{[0]} \right \vert}\), respectively. Moreover, notice that due to the conservation of \(W^\pm\)-particles each 3- or 4-valent vertex has to consist of an even number of \(W^\pm\)-particles. In particular, since the 4-valent vertices are created by the Corolla polynomial in the same fashion than for the gluons we only need to consider the 3-valent Feynman rules, i.e. we have to sum over all possibilities to label edges in \(\Gamma\) in such a way, that each 3-valent vertex consists of exactly two half-edges with label \(W^\pm\) and one unlabeled half-edge. Additionally the edge type has to be respected while gluing these 3-valent vertices together. This is precisely fulfilled by the sets \(\mathfrak f \in \mathfrak F_2 \left ( \Gamma \right )\). Furthermore, notice that the Feynman rules of the electroweak sector of the Standard Model allow each unlabeled edge of \(\Gamma\) to be turned into either a \(Z\)-particle or an \(A\)-particle. This is precisely fulfilled by the sets \(P_Z \in \mathscr P_Z (\mathfrak f)\) for each set \(\mathfrak f \in \mathfrak F_2 \left ( \Gamma \right )\). Moreover, the corresponding ghost edges have the same masses than their corresponding gauge bosons, meaning that each component \(\mathcal D^i (\Gamma) I (\Gamma)\) gets the same factor \(J \left ( \Gamma; m_W, m_Z \right )\) and thus the whole contribution created by the Corolla differential can be multiplied by the factor given in \mbox{Equation \eqref{eqn:j}}, such that \mbox{Equation \eqref{eqn:amplitude_electroweak_gauge_bosons}} holds. Since no derivatives of \(\mathcal D (\Gamma)\) act on the corresponding mass terms in the parametric representation of a scalar QFT, \(I (\Gamma)\) can be simply multiplied by \(J \left ( \Gamma; m_W, m_Z \right )\). Finally, the possibly different symmetry factors of \(\Gamma_{\text{labeled}}\) compared to \(\Gamma\) are compensated by the factor \(\operatorname{sym} \left ( \Gamma \right ) / \operatorname{sym} \left ( \Gamma_{\text{labeled}} \right )\) and the redundant number of isomorphic graphs is divided out by \(\operatorname{iso} \left ( \Gamma_{\text{labeled}} \right )\).
\end{proof}

\vspace{\baselineskip}

\begin{lem}[\cite{Prinz}] \label{lem:symsymiso}
Let the assumptions of \thmref{thm:corolla_polynomial-electroweak_gauge_bosons} hold. Then we have
\begin{equation}
	\frac{\operatorname{sym} \left ( \Gamma \right )}{\operatorname{sym} \left ( \Gamma_{\text{\emph{labeled}}} \right ) \operatorname{iso} \left ( \Gamma_{\text{\emph{labeled}}} \right )} = 1, \qquad \forall \Gamma \, .
\end{equation}
\end{lem}

\begin{proof}
Note that the symmetry factors of pure electroweak gauge boson Feynman graphs are similar to the symmetry factors of quantum electrodynamic Feynman graphs, since every vertex possesses exactly two orientable edges of the same type and one unoriented edge of another type. Since we have chosen to work in this paper with unoriented \(W^\pm\)-labeled edges, the symmetry factors of pure electroweak gauge boson Feynman graphs are not all equal to 1, as in quantum electrodynamics for unoriented fermion edges \cite{Kreimer}. Moreover, note that the labeling with \(Z\)-labels and \(A\)-labels does not change the symmetry-factors, i.e. the element \(\mathfrak f \in \mathfrak F_2 \left ( \Gamma \right )\) has the same symmetry factor as all elements \(P_Z \in \mathscr P_Z (\mathfrak f)\). This is true because the symmetry factor can be defined as the number of ways half-edges of adjacent vertices can be interchanged without changing the graph. But since two edges of a vertex are of the same particle type, the remaining particle type is fixed and hence it does not matter if it is of \(Z\)-type or \(A\)-type, the only contribution comes from the unoriented two edges of the same particle type. Furthermore, we remark that as in quantum electrodynamics the sum of all possible orientations of unoriented edges (which is given by the factor \(\operatorname{iso} \left ( \Gamma_{\text{labeled}} \right )\)) times their symmetry factor \(\operatorname{sym} \left ( \Gamma_{\text{labeled}} \right )\) is equal to 1 \cite{Kreimer}. Moreover, when also ghost edges come into play, the symmetry factors may also change, but in the same way changes the number of isomorphic graphs with the same argument. Finally, in \cite[Lemma 5.1]{Sars} it is proven that the symmetry factors work out correct while passing from 3-valent to 4-valent vertices by shrinking internal edges.
\end{proof}

\vspace{\baselineskip}

\begin{cor}[\cite{Prinz}] \label{cor:inclusion_ew_gauge_bosons}
In particular, \(J \left ( \Gamma; m_W, m_Z \right )\) simplifies to
\begin{equation}
\begin{split}
	J \left ( \Gamma; m_W, m_Z \right ) := & \left ( \vphantom{\phantom{ ( + } \times e^{-\left ( \sum_{e \in \mathfrak f^{[1]}} A_e m_W^2 + \sum_{e \in P_Z^{[1]}} A_e m_Z^2 \right )}} i^{\left \vert \Gamma^{[1]} \right \vert} \underline{g_s}^{\left \vert \Gamma^{[0]} \right \vert} \operatorname{color} (\Gamma) \right . \\ & \left . \phantom{ ( } + \sum_{\mathfrak f \in \mathfrak F_2 \left ( \Gamma \right )} \sum_{P_Z \in \mathscr{P}_Z (\mathfrak f)} i^{\left \vert \Gamma^{[0]} \right \vert + \left \vert \Gamma^{[1]} \right \vert} \left ( \underline{g} \cos \theta_W \right ) ^{\left \vert P_Z^{[0]} \right \vert} \underline{e}^{\left \vert \Gamma^{[0]} \setminus P_Z^{[0]} \right \vert} \right . \\ & \left . \phantom{ ( + } \times e^{-\left ( \sum_{e \in \mathfrak f^{[1]}} A_e m_W^2 + \sum_{e \in P_Z^{[1]}} A_e m_Z^2 \right )} \right ) \, .
\end{split}
\end{equation}
\end{cor}

\begin{proof}
\thmref{thm:corolla_polynomial-electroweak_gauge_bosons} and \lemref{lem:symsymiso}.
\end{proof}

\vspace{\baselineskip}

\begin{exmp}
Consider
\begin{equation}
\Gamma := 
\vcenter{\hbox{\includegraphics[height=2.5cm]{1L_Sc_labeled.eps}}} \, .
\end{equation}
Then, we have (cf. \exref{exmp:1L_sc-PW_PZ}):
\begin{equation}
\begin{split}
	J \left ( \vcenter{\hbox{\includegraphics[height=2.5cm]{1L_Sc_labeled.eps}}}; m_W, m_Z \right ) = & \left ( \vphantom{\phantom{ ( } + e^{- \left ( (A_1 + A_2) m_W^2 + (A_3 + A_4) m_Z^2 \right )} \left ( \underline{g} \cos \theta_W \right )^2} - \underline{g_s}^2 f^{c_3 c_2 c_1} f^{c_4 c_2 c_1} \right . \\
& \left . \phantom{ ( } + e^{- \left ( (A_3 + A_2 + A_4) m_W^2 \right )} \underline{e}^2 \right . \\
& \left . \phantom{ ( } + e^{- \left ( (A_3 + A_2 + A_4) m_W^2 + A_1 m_Z^2 \right )} \left ( \underline{g} \cos \theta_W \right )^2 \right . \\
& \left . \phantom{ ( } + e^{- \left ( (A_3 + A_1 + A_4) m_W^2 \right )} \underline{e}^2 \right . \\
& \left . \phantom{ ( } + e^{- \left ( (A_3 + A_1 + A_4) m_W^2 + A_2 m_Z^2 \right )} \left ( \underline{g} \cos \theta_W \right )^2 \right . \\
& \left . \phantom{ ( } + e^{- \left ( (A_1 + A_2) m_W^2 \right )} \underline{e}^2 \right . \\
& \left . \phantom{ ( } + e^{- \left ( (A_1 + A_2) m_W^2 + A_4 m_Z^2 \right )} \underline{e} \underline{g} \cos \theta_W \right . \\
& \left . \phantom{ ( } + e^{- \left ( (A_1 + A_2) m_W^2 + A_3 m_Z^2 \right )} \underline{e} \underline{g} \cos \theta_W \right . \\
& \left . \phantom{ ( } + e^{- \left ( (A_1 + A_2) m_W^2 + (A_3 + A_4) m_Z^2 \right )} \left ( \underline{g} \cos \theta_W \right )^2 \right )
\end{split}
\end{equation}
\end{exmp}

\subsection{Inclusion of the scalar Bosons of the Electroweak Sector} \label{ssec:corolla_polynomial_scalar_particles}

For the inclusion of the scalar bosons from the electroweak sector of the Standard Model we have to add four scalar particles, bringing new possible tensor structures in their Feynman rules with them. Concretely, we have to add the Higgs boson \(h\) and the Goldstone bosons \(\varphi^\pm\) and \(\varphi_Z\).

\vspace{\baselineskip}

\begin{out}
Again, the idea will be to define two nested sums over a 3-regular scalar QFT Feynman graph in the following way: The first sum creates all possible ways to attach scalar particle labels to edges of \(\Gamma\) that are allowed to become a scalar edge (i.e. edges that will not become a ghost edge) and the second sum will create all possible Feynman graphs with 4-valent scalar vertices by shrinking suitable scalar labeled edges. Furthermore, we define a third sum which creates all particle labelings which are allowed by the Standard Model Feynman rules. The full details are explained in the proof after \thmref{thm:corolla_polynomial-electroweak_scalar_bosons}.
\end{out}

\vspace{\baselineskip}

Therefore, we define:

\vspace{\baselineskip}

\begin{defn} \label{defn:scalarcorolla}
Let \(\Gamma\) be a 3-regular scalar QFT Feynman graph. Then:
\begin{enumerate}
\item Again, as in \defnref{defn:F_2_P_Z}, let \(\wp \left ( \Gamma^{[1]} \right )\) be the power set of all external edges and internal edges of \(\Gamma\) with all adjacent vertices added. In particular, we are interested in the set \(\wp \left (\Gamma^{[1]} \setminus \left ( \bigcup_{k=1}^i C_i^{[1]} \right ) \right )\), the power set of all edges of \(\Gamma\) that will not be turned into ghost edges and are free to possibly become a Higgs or Goldstone edge. We denote the elements of \(\wp \left (\Gamma^{[1]} \setminus \left ( \bigcup_{k=1}^i C_i^{[1]} \right ) \right )\) by \(P_{\text{H/G}}\).
\item Let \(\wp_{(4)} \left (P_{\text{H/G}} \right )\) be the power set of all internal edges of \(\Gamma\) in the set \(P_{\text{H/G}}\) having \(2\) or \(4\) neighbors, that do not share a vertex with an edge which will be turned into a ghost edge (also edges that will be turned into fermions edges are not allowed either, if included). Also adjacent edges are not allowed to be in the same set \(P_{(4)} \in \wp_{(4)} \left ( P_{\text{H/G}} \right )\). The set of sets \(\wp_{(4)} \left ( P_{\text{H/G}} \right )\) consists therefore of all edges in the set \(P_{\text{H/G}}\) that are free to create a valid \(4\)-valent vertex (either 2-scalar-2-gauge boson or 4-scalar boson vertices). Again, add all adjacent vertices in each edge set.
\item Let \(e \in P_{(4)} \left ( P_{\text{H/G}} \right )\). Then we define the set of adjacent half-edges to \(e\) as \(H_{(4)} (e) := \set{h_1, h_2, h_3, h_4}\). In particular, we are interested in the set \(H_{(2)} (e) \subset H_{(4)} (e)\), defined as \[H_{(2)} (e) := \begin{cases} \set{h_1, h_2} & \text{if } h_1, h_2 \notin P_{\text{H/G}} \text{ and } h_3, h_4 \in P_{\text{H/G}} \\ \emptyset & \text{if } h_1, h_2, h_3, h_4 \in P_{\text{H/G}} \end{cases} \] for some arbitrary numbering \(1, \ldots, 4\), since in the case of a 2-scalar-2-gauge boson vertex we have to add a metric tensor with indices depending on the gauge boson edges (where the product below in the modified Corolla polynomial in \defnref{defn:corolla_polynomial-electroweak_scalar_bosons} over the empty set is defined to be 0). For the 4-scalar boson vertex we just receive a coupling constant, cf. part 5. of this definition.
\item Let \(\mathscr L \left ( \Gamma, P_{\text{H/G}}, P_{(4)} \right )\) denote the set of all possible allowed particle type labelings of the graph \(\Gamma\), with scalar particle edges \(P_{\text{H/G}}\) and 4-valent scalar vertices \(P_{(4)}\), compatible with the Standard Model Feynman rules (cf. \cite[Appendix A]{Prinz} and \cite{Romao_Silva}), such that\footnote{We denote here the gauge bosons as well as their corresponding ghosts, by \(\set{W^\pm, Z, A}\) since in the following we are only interested in their pairwise coincident masses.} \[ L(e) \in \left \{ W^\pm, Z, A, h, \varphi^\pm, \varphi_Z \right \} \, , \qquad e \in \Gamma^{[1]} \, . \] In the previous \mbox{Subsection \ref{ssec:corolla_polynomial_gauge_bosons_electroweak}} this was created by the sets of sets \(\mathfrak F_2 \left ( \Gamma \right )\) and \(\mathscr P_Z \left ( \mathfrak f \right )\) for the special cases \(P_{\text{H/G}} = P_{(4)} = \emptyset\), defined in \defnref{defn:F_2_P_Z}.
\item Let \(\operatorname{coupling} \left ( P_{\text{H/G}}, P_{(4)} , L \right )\) be the product over all factors of \(\pm i\) and all coupling constants given by the Feynman rules for \(\Gamma\) with labeling \(L\) (cf. \cite[Appendix A]{Prinz} and \cite{Romao_Silva}). The coupling constants include all 3-valent vertices and 4-scalar boson vertices (the 4-gauge boson vertex coupling constants are created from the corresponding 3-valent ones by the Corolla polynomial via graph homology as usual).
\end{enumerate}
\end{defn}

\vspace{\baselineskip}

We start with examples for 1., 2. and 3. from \defnref{defn:scalarcorolla}:

\vspace{\baselineskip}

\begin{exmp}
Again, we consider the one-loop graph
\begin{subequations}
\begin{equation}
\Gamma := 
\vcenter{\hbox{\includegraphics[height=2.5cm]{1L_Sc_labeled.eps}}}
\end{equation}
and
\begin{equation}
\Gamma^\prime := 
\vcenter{\hbox{\includegraphics[height=2.5cm]{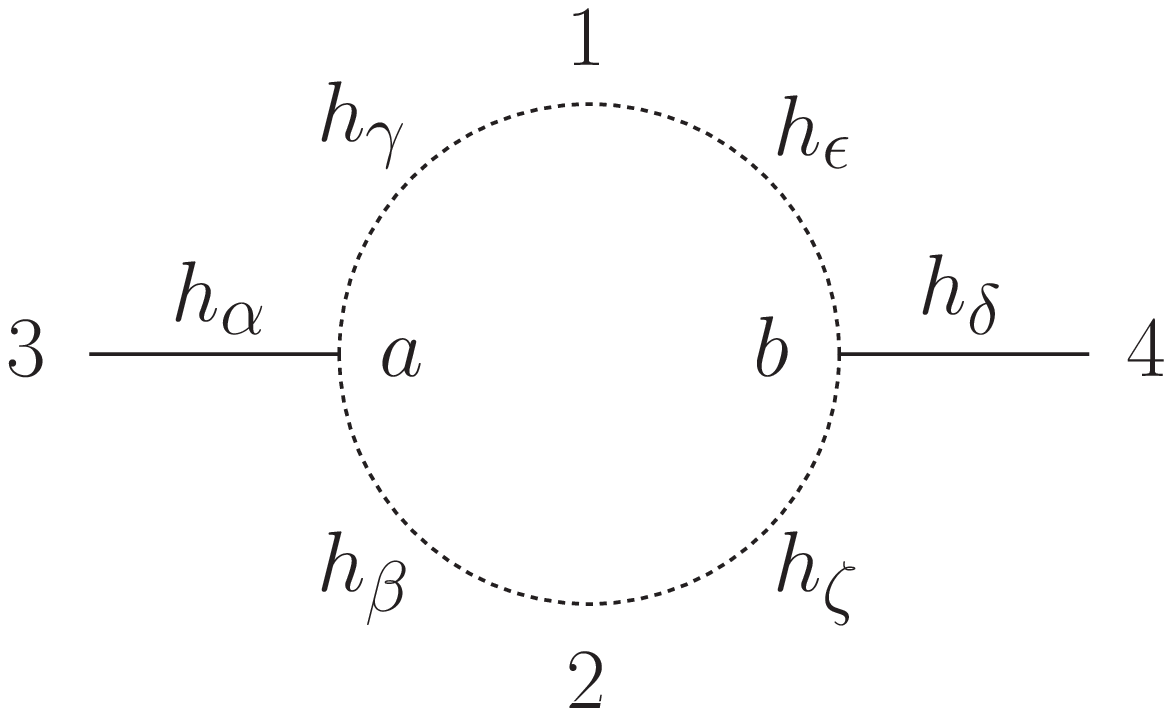}}} \, .
\end{equation}
\end{subequations}
Then, we have:
\begin{subequations}
\begin{equation}
\begin{split}
\wp \left ( \Gamma^{[1]} \right ) = & \left \{ \underbrace{\emptyset}_{P_{\text{H/G}}^{(1)} (\Gamma)}, \underbrace{\set{a,b;1}}_{P_{\text{H/G}}^{(2)} (\Gamma)}, \underbrace{\set{a,b;2}}_{P_{\text{H/G}}^{(3)} (\Gamma)}, \underbrace{\set{a;3}}_{P_{\text{H/G}}^{(4)} (\Gamma)}, \underbrace{\set{b;4}}_{P_{\text{H/G}}^{(5)} (\Gamma)}, \underbrace{\set{a,b;1,2}}_{P_{\text{H/G}}^{(6)} (\Gamma)}, \underbrace{\set{a,b;1,3}}_{P_{\text{H/G}}^{(7)} (\Gamma)}, \right . \\ & \phantom{ \{ } \left . \underbrace{\set{a,b;1,4}}_{P_{\text{H/G}}^{(8)} (\Gamma)}, \underbrace{\set{a,b;2,3}}_{P_{\text{H/G}}^{(9)} (\Gamma)}, \underbrace{\set{a,b;2,4}}_{P_{\text{H/G}}^{(10)} (\Gamma)}, \underbrace{\set{a,b;3,4}}_{P_{\text{H/G}}^{(11)} (\Gamma)}, \underbrace{\set{a,b;1,2,3}}_{P_{\text{H/G}}^{(12)} (\Gamma)}, \right . \\ & \phantom{ \{ } \left . \underbrace{\set{a,b;1,2,4}}_{P_{\text{H/G}}^{(13)} (\Gamma)}, \underbrace{\set{a,b;1,3,4}}_{P_{\text{H/G}}^{(14)} (\Gamma)}, \underbrace{\set{a,b;2,3,4}}_{P_{\text{H/G}}^{(15)} (\Gamma)}, \underbrace{\set{a,b;1,2,3,4}}_{P_{\text{H/G}}^{(16)} (\Gamma)} \right \}
\end{split}
\end{equation}
\begin{equation}
\wp \left ( {\Gamma^\prime}^{[1]} \right ) = \wp \left ( \Gamma^{[1]} \setminus C_1^{[1]} \right ) = \left \{ \underbrace{\emptyset}_{P_{\text{H/G}}^{(1) \prime} (\Gamma^\prime)}, \underbrace{\set{a;3}}_{P_{\text{H/G}}^{(2) \prime} (\Gamma^\prime)}, \underbrace{\set{b;4}}_{P_{\text{H/G}}^{(3) \prime} (\Gamma^\prime)}, \underbrace{\set{a,b;3,4}}_{P_{\text{H/G}}^{(4) \prime} (\Gamma^\prime)} \right \}
\end{equation}
\end{subequations}
and
\begin{subequations}
{\allowdisplaybreaks
\begin{align}
\wp_{(4)} \left ( P_{\text{H/G}}^{(l)} \right ) = & \left \{ \underbrace{\emptyset}_{P_{(4)}^{(1)} \left (P_{\text{H/G}}^{(l)} \right )} \right \} \, , \qquad 1 \leq l \leq 13 \\
\wp_{(4)} \left ( P_{\text{H/G}}^{(14)} \right ) = & \left \{ \underbrace{\emptyset}_{P_{(4)}^{(1)} \left (P_{\text{H/G}}^{(14)} \right )}, \underbrace{\set{a,b;1}}_{P_{(4)}^{(2)} \left (P_{\text{H/G}}^{(14)} \right )} \right \} \\
\wp_{(4)} \left ( P_{\text{H/G}}^{(15)} \right ) = & \left \{ \underbrace{\emptyset}_{P_{(4)}^{(1)} \left (P_{\text{H/G}}^{(15)} \right )}, \underbrace{\set{a,b;2}}_{P_{(4)}^{(2)} \left (P_{\text{H/G}}^{(15)} \right )} \right \} \\
\wp_{(4)} \left ( P_{\text{H/G}}^{(16)} \right ) = & \left \{ \underbrace{\emptyset}_{P_{(4)}^{(1)} \left (P_{\text{H/G}}^{(16)} \right )}, \underbrace{\set{a,b;1}}_{P_{(4)}^{(2)} \left (P_{\text{H/G}}^{(16)} \right )}, \underbrace{\set{a,b;2}}_{P_{(4)}^{(3)} \left (P_{\text{H/G}}^{(16)} \right )} \right \} \\
\wp_{(4)} \left ( P_{\text{H/G}}^{(l) \prime} \right ) = & \left \{ \underbrace{\emptyset}_{P_{(4)}^{(1)} \left (P_{\text{H/G}}^{(l) \prime} \right )} \right \} \, , \qquad 1 \leq l \leq 4
\end{align}
}
\end{subequations}

And for our sets \(H_{(2)} (e)\), \(e \in P_{(4)} \left ( P_{\text{H/G}} \right )\), we have (we just list the non-empty sets):
\begin{subequations}
{\allowdisplaybreaks
\begin{align}
1 \in P_{(4)}^{(1)} \left (P_{\text{H/G}}^{(6)} \right ) : \qquad H_{(2)} (1) = & \set{h_{\alpha} , h_{\delta}} \\
2 \in P_{(4)}^{(1)} \left (P_{\text{H/G}}^{(6)} \right ) : \qquad H_{(2)} (2) = & \set{h_{\alpha} , h_{\delta}} \\
1 \in P_{(4)}^{(1)} \left (P_{\text{H/G}}^{(14)} \right ) : \qquad H_{(2)} (1) = & \set{h_{\beta} , h_{\zeta}} \\
2 \in P_{(4)}^{(1)} \left (P_{\text{H/G}}^{(15)} \right ) : \qquad H_{(2)} (2) = & \set{h_{\gamma} , h_{\epsilon}}
\end{align}
}
\end{subequations}

These create, after applying the Corolla differential for the electroweak sector of the Standard Model (cf. \defnref{defn:corolla_polynomial-electroweak_scalar_bosons} and \thmref{thm:corolla_polynomial-electroweak_scalar_bosons}) to the scalar parametric integrand \(I (\Gamma) \), the tensor structure contributions of the following gauge theory graphs (which is indicated by the arrow ``\(\rightsquigarrow\)''):
\begin{subequations}
{\allowdisplaybreaks
\begin{align}
P_{(4)}^{(1)} \left (P_{\text{H/G}}^{(1)} \right ) \rightsquigarrow & \vcenter{\hbox{\includegraphics[height=1.5cm]{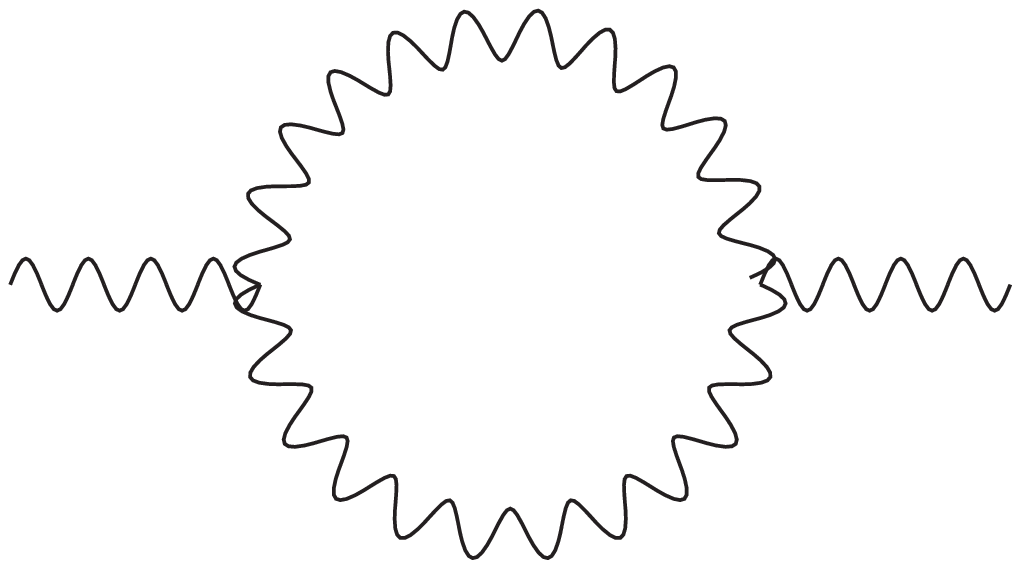}}} \\
P_{(4)}^{(1)} \left (P_{\text{H/G}}^{(2)} \right ) \rightsquigarrow & \vcenter{\hbox{\includegraphics[height=1.5cm]{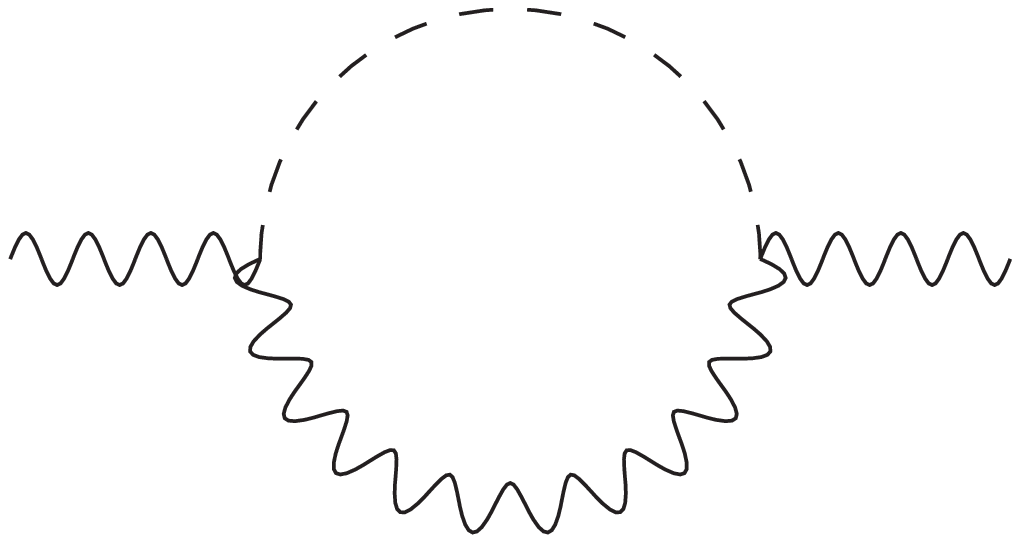}}} \\
P_{(4)}^{(1)} \left (P_{\text{H/G}}^{(3)} \right ) \rightsquigarrow & \vcenter{\hbox{\includegraphics[height=1.5cm]{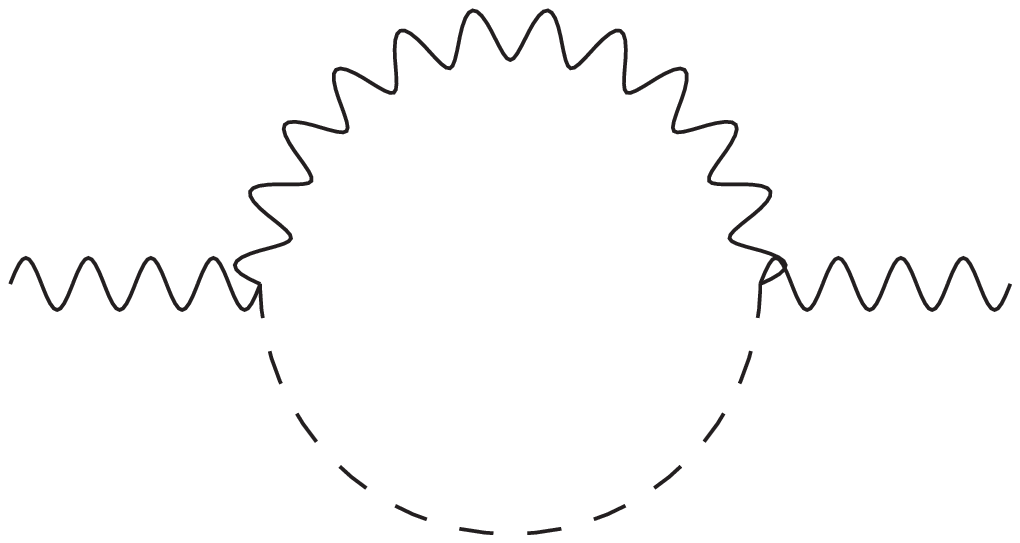}}} \\
P_{(4)}^{(1)} \left (P_{\text{H/G}}^{(4)} \right ) \rightsquigarrow & \vcenter{\hbox{\includegraphics[height=1.5cm]{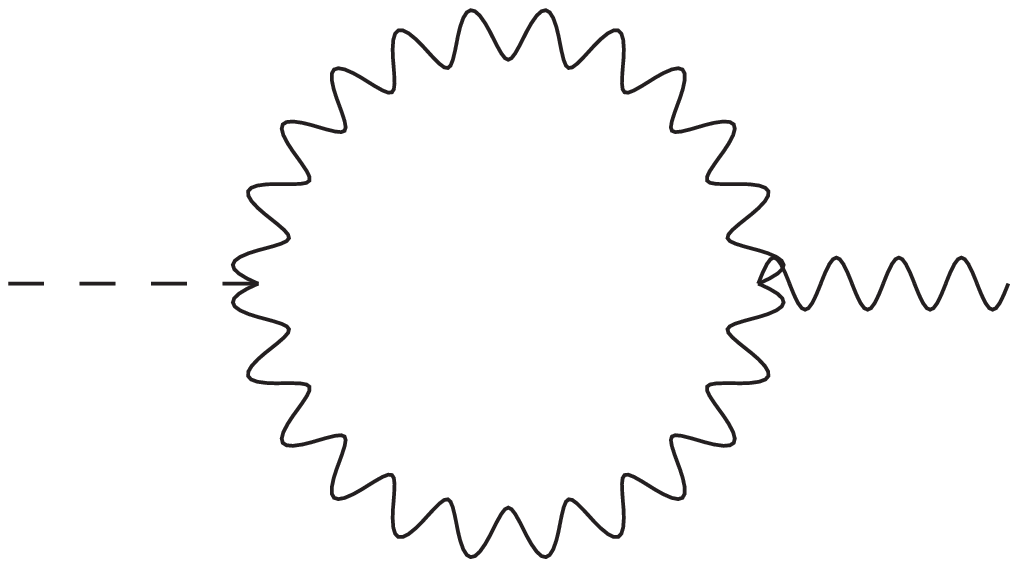}}} \\
P_{(4)}^{(1)} \left (P_{\text{H/G}}^{(5)} \right ) \rightsquigarrow & \vcenter{\hbox{\includegraphics[height=1.5cm]{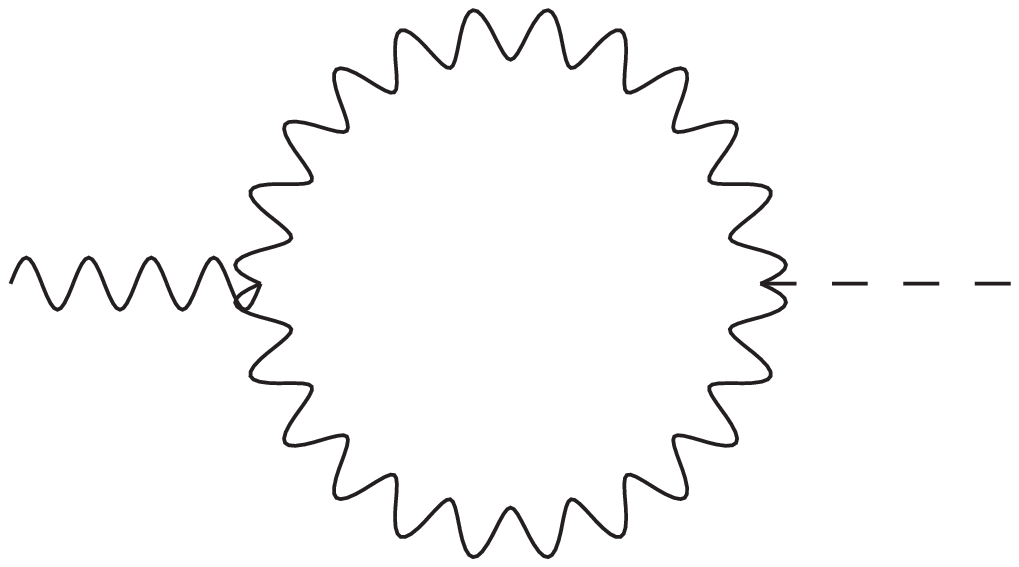}}} \\
P_{(4)}^{(1)} \left (P_{\text{H/G}}^{(6)} \right ) \rightsquigarrow & \vcenter{\hbox{\includegraphics[height=1.5cm]{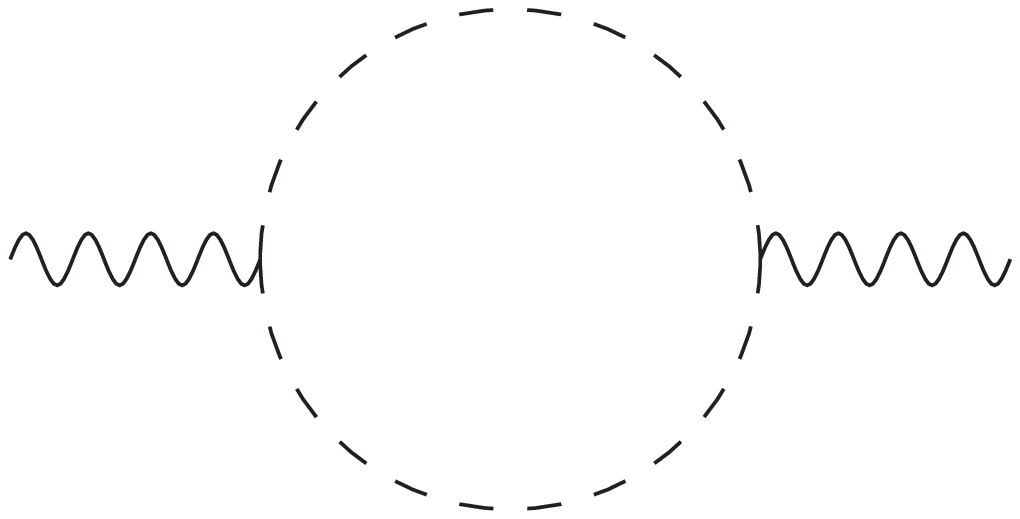}}} \\
P_{(4)}^{(1)} \left (P_{\text{H/G}}^{(7)} \right ) \rightsquigarrow & \vcenter{\hbox{\includegraphics[height=1.5cm]{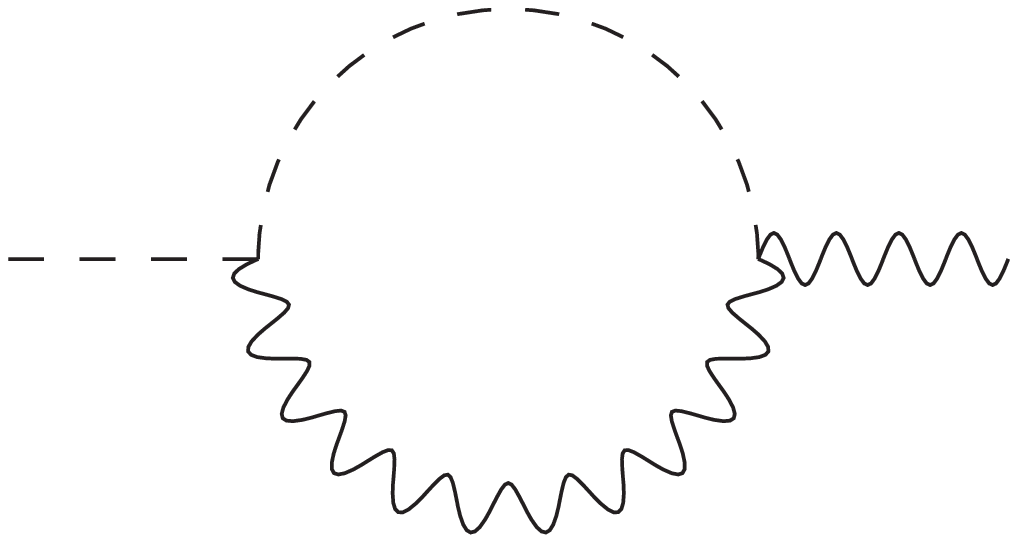}}} \\
P_{(4)}^{(1)} \left (P_{\text{H/G}}^{(8)} \right ) \rightsquigarrow & \vcenter{\hbox{\includegraphics[height=1.5cm]{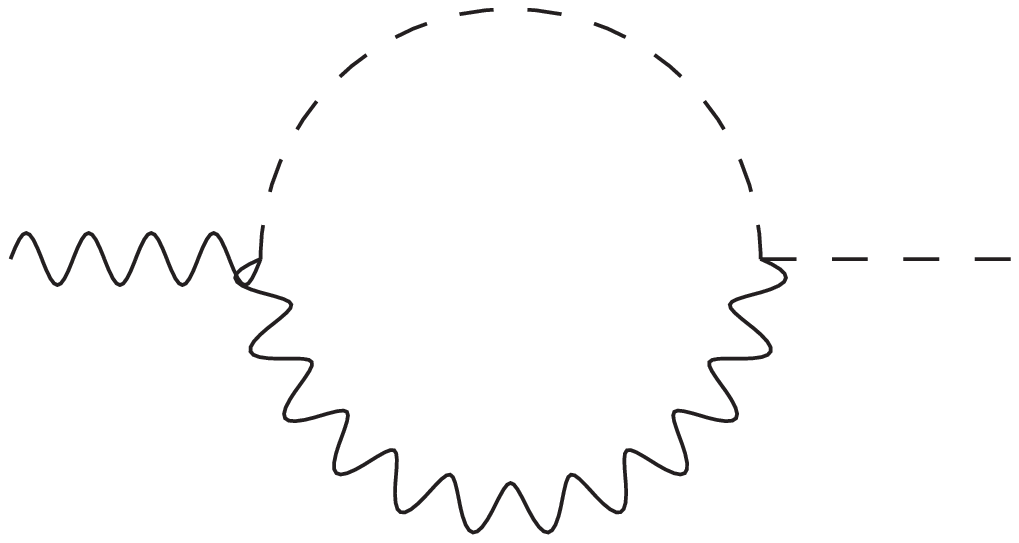}}} \\
P_{(4)}^{(1)} \left (P_{\text{H/G}}^{(9)} \right ) \rightsquigarrow & \vcenter{\hbox{\includegraphics[height=1.5cm]{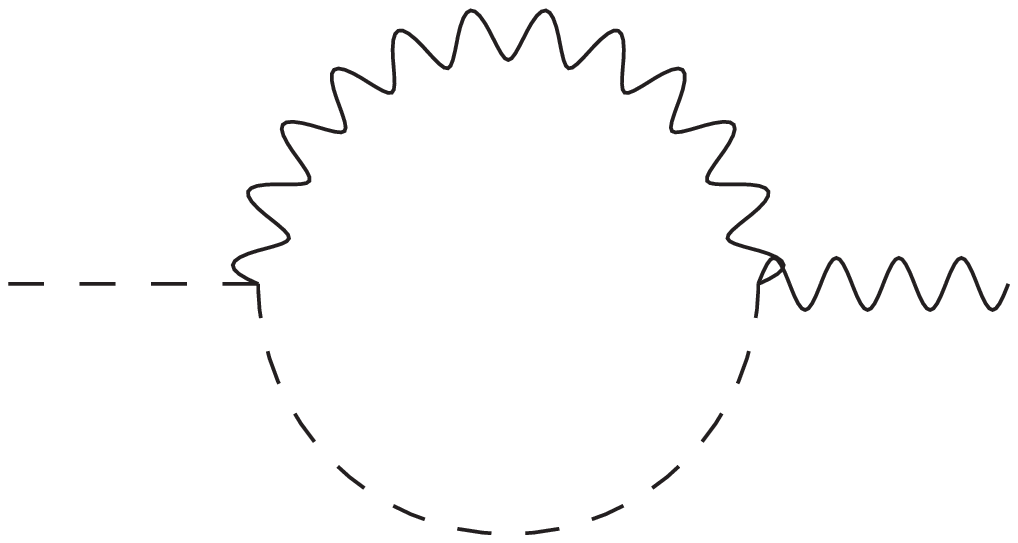}}} \\
P_{(4)}^{(1)} \left (P_{\text{H/G}}^{(10)} \right ) \rightsquigarrow & \vcenter{\hbox{\includegraphics[height=1.5cm]{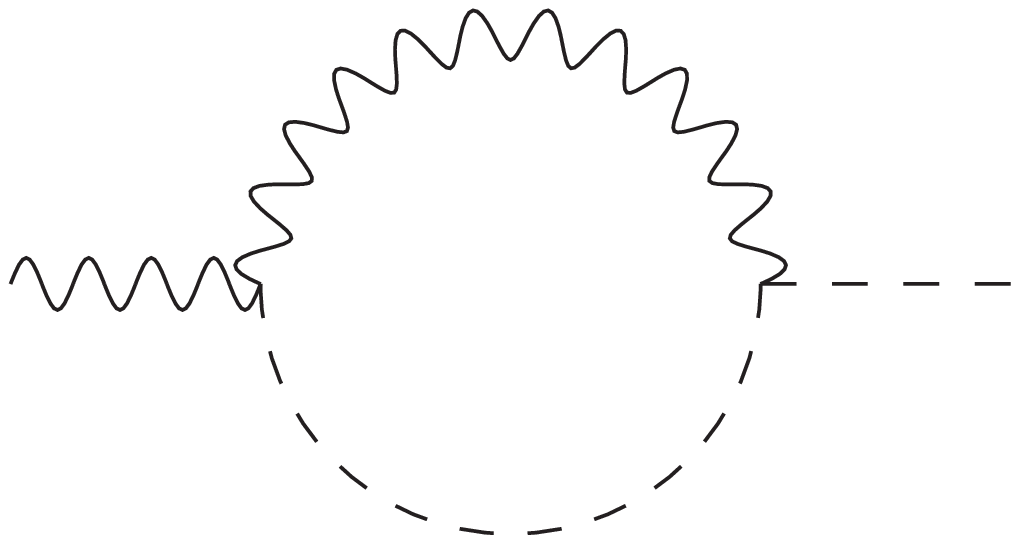}}} \\
P_{(4)}^{(1)} \left (P_{\text{H/G}}^{(11)} \right ) \rightsquigarrow & \vcenter{\hbox{\includegraphics[height=1.5cm]{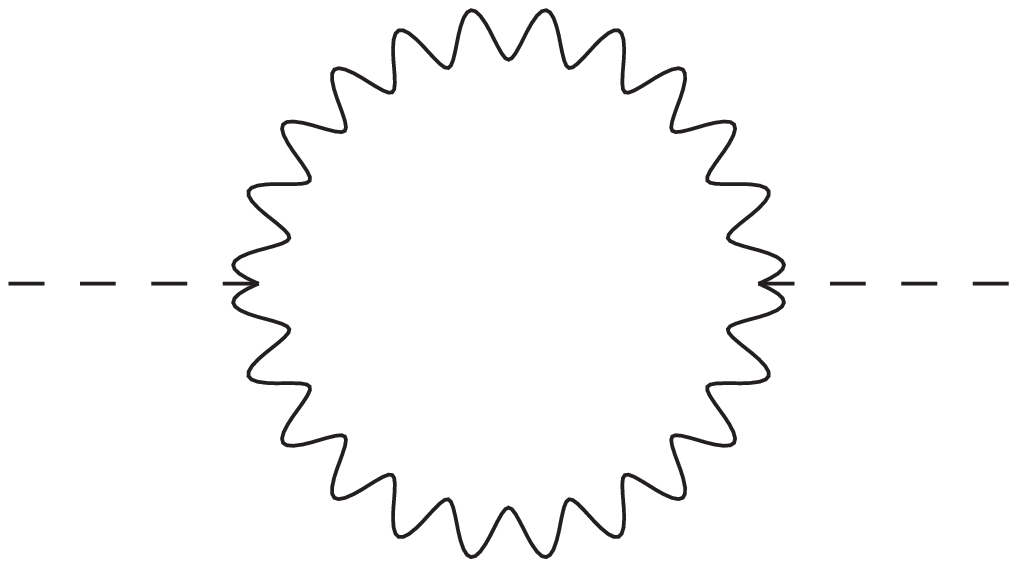}}} \\
P_{(4)}^{(1)} \left (P_{\text{H/G}}^{(12)} \right ) \rightsquigarrow & \vcenter{\hbox{\includegraphics[height=1.5cm]{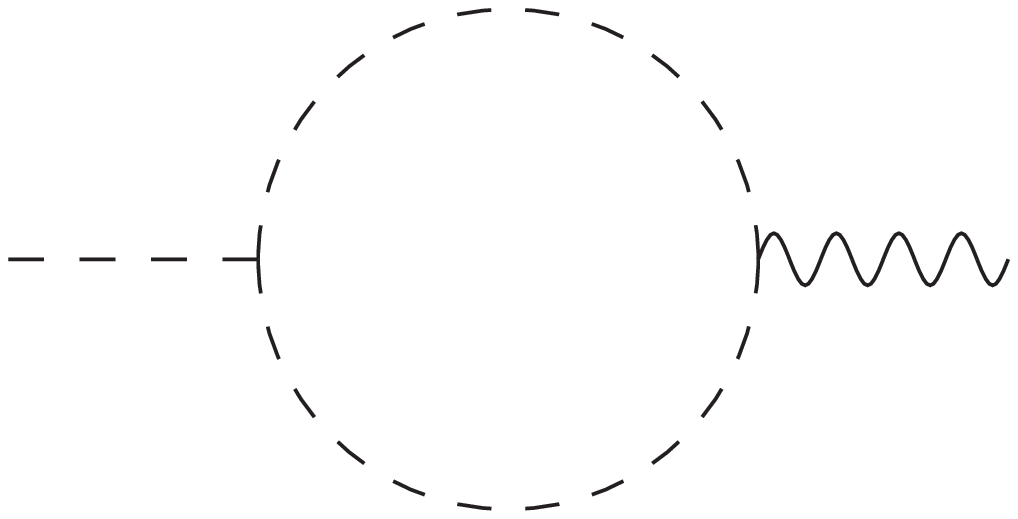}}} \\
P_{(4)}^{(1)} \left (P_{\text{H/G}}^{(13)} \right ) \rightsquigarrow & \vcenter{\hbox{\includegraphics[height=1.5cm]{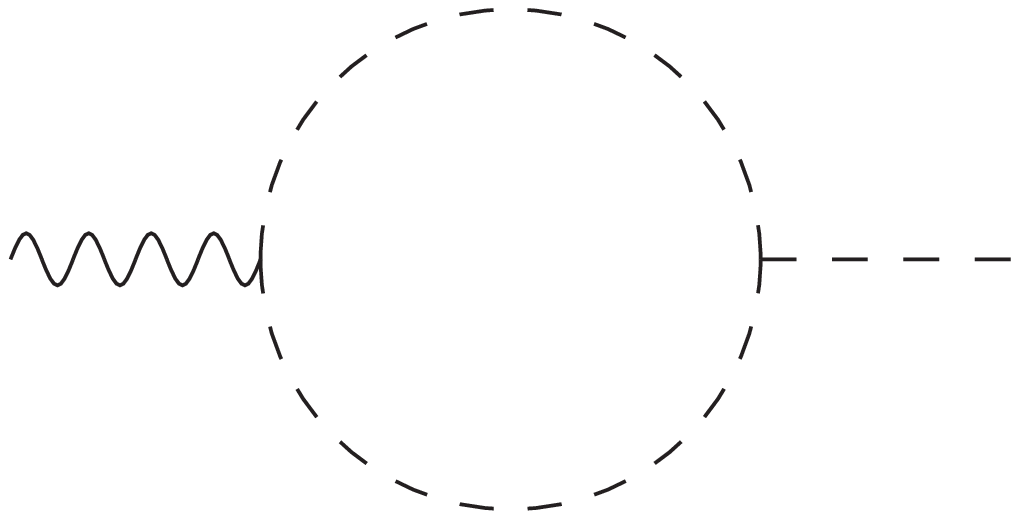}}} \\
P_{(4)}^{(1)} \left (P_{\text{H/G}}^{(14)} \right ) \rightsquigarrow & \vcenter{\hbox{\includegraphics[height=1.5cm]{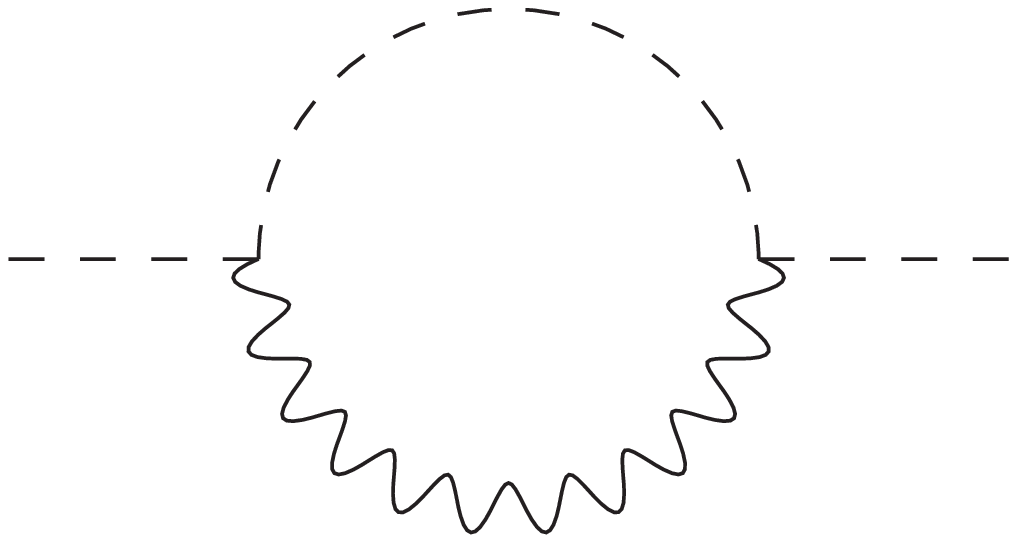}}} \\
P_{(4)}^{(2)} \left (P_{\text{H/G}}^{(14)} \right ) \rightsquigarrow & \vcenter{\hbox{\includegraphics[height=1.5cm]{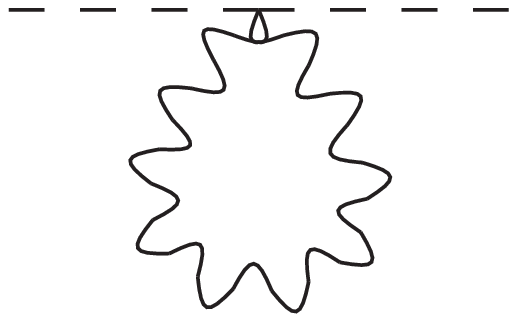}}} \\
P_{(4)}^{(1)} \left (P_{\text{H/G}}^{(15)} \right ) \rightsquigarrow & \vcenter{\hbox{\includegraphics[height=1.5cm]{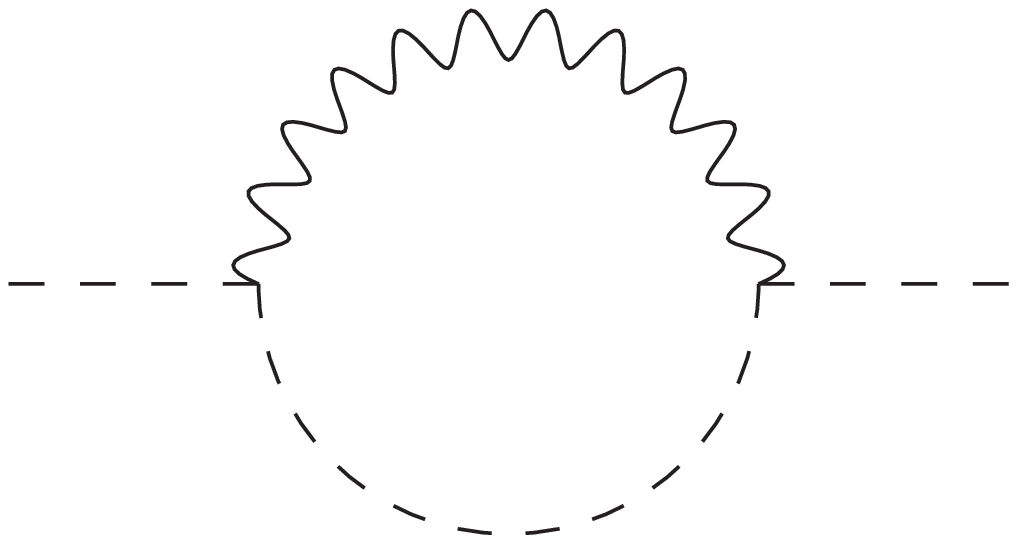}}} \\
P_{(4)}^{(2)} \left (P_{\text{H/G}}^{(15)} \right ) \rightsquigarrow & \vcenter{\hbox{\includegraphics[height=1.5cm]{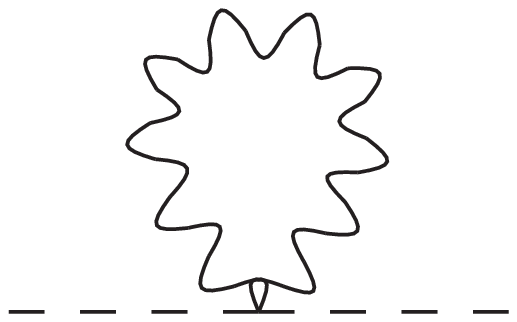}}} \\
P_{(4)}^{(1)} \left (P_{\text{H/G}}^{(16)} \right ) \rightsquigarrow & \vcenter{\hbox{\includegraphics[height=1.5cm]{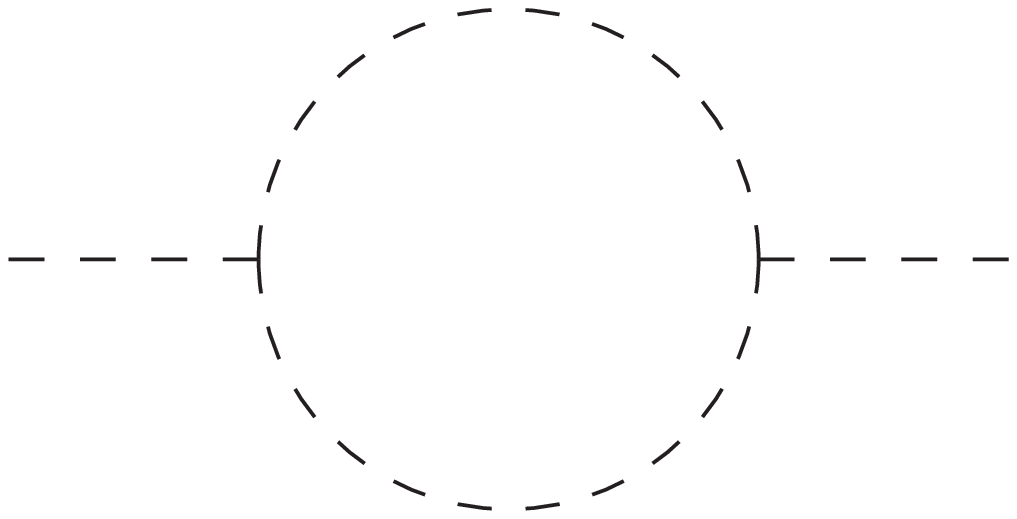}}} \\
P_{(4)}^{(2)} \left (P_{\text{H/G}}^{(16)} \right ) \rightsquigarrow & \vcenter{\hbox{\includegraphics[height=1.5cm]{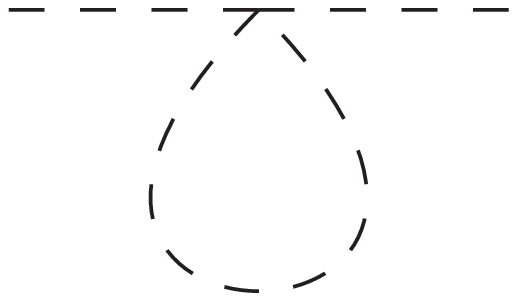}}} \\
P_{(4)}^{(3)} \left (P_{\text{H/G}}^{(16)} \right ) \rightsquigarrow & \vcenter{\hbox{\includegraphics[height=1.5cm]{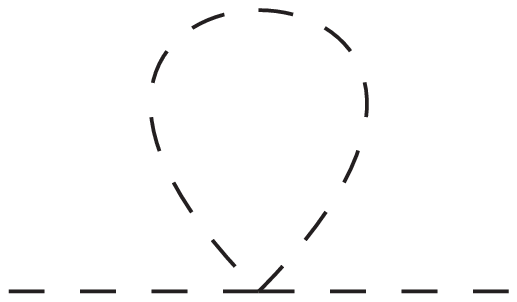}}} \\
P_{(4)}^{(1)} \left (P_{\text{H/G}}^{(1) \prime} \right ) \rightsquigarrow & \vcenter{\hbox{\includegraphics[height=1.5cm]{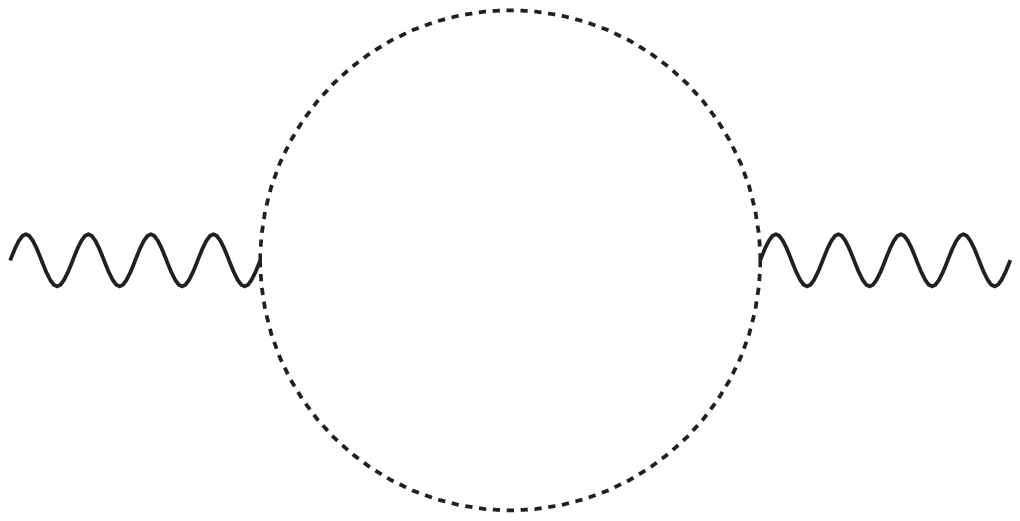}}} \\
P_{(4)}^{(1)} \left (P_{\text{H/G}}^{(2) \prime} \right ) \rightsquigarrow & \vcenter{\hbox{\includegraphics[height=1.5cm]{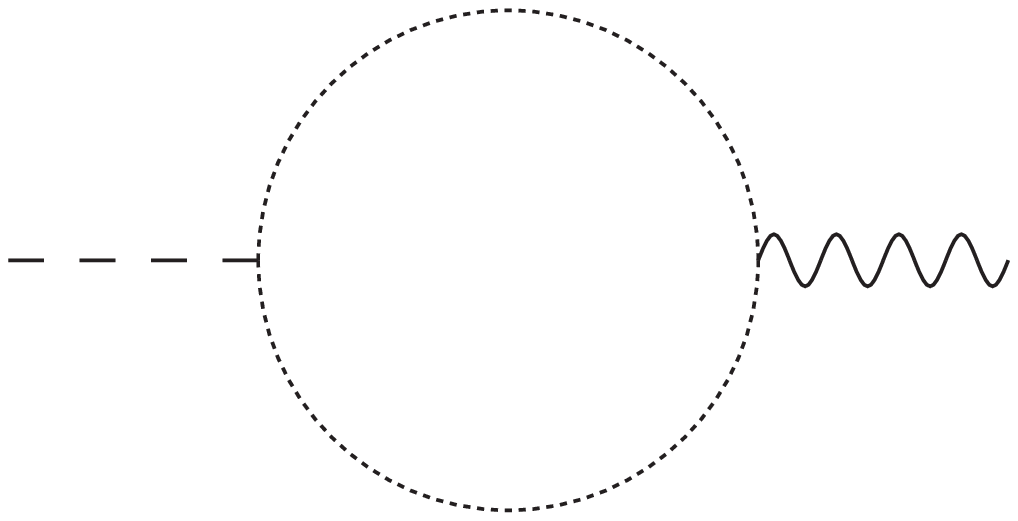}}} \\
P_{(4)}^{(1)} \left (P_{\text{H/G}}^{(3) \prime} \right ) \rightsquigarrow & \vcenter{\hbox{\includegraphics[height=1.5cm]{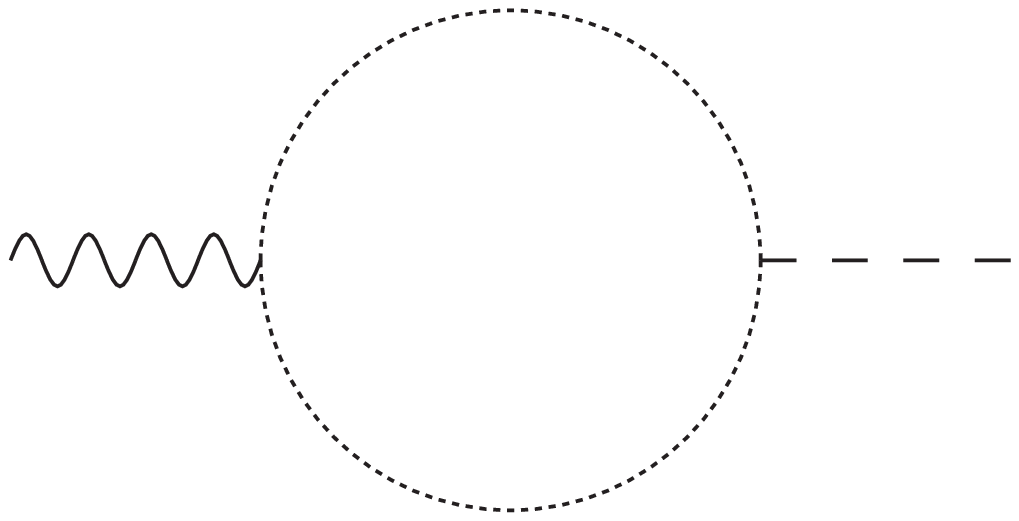}}} \\
P_{(4)}^{(1)} \left (P_{\text{H/G}}^{(4) \prime} \right ) \rightsquigarrow & \vcenter{\hbox{\includegraphics[height=1.5cm]{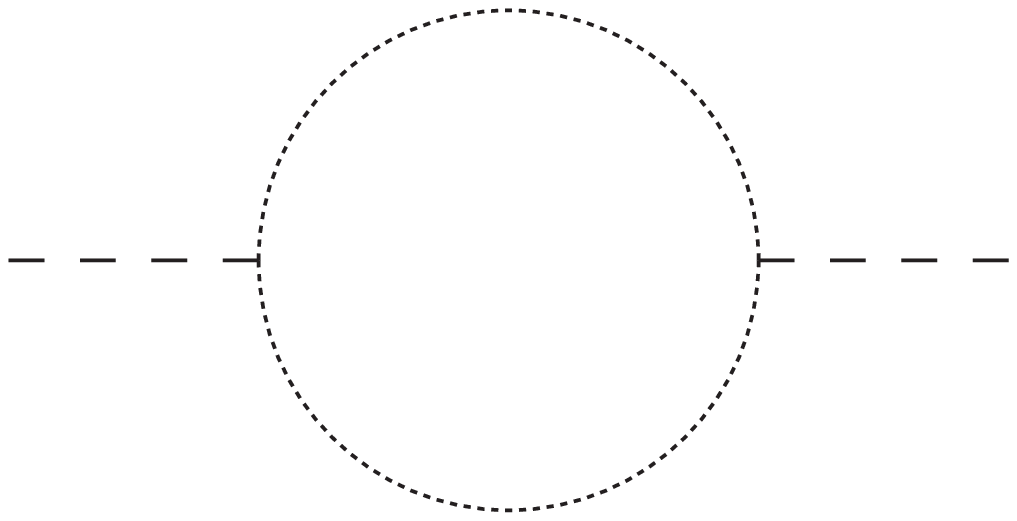}}}
\end{align}
}
\end{subequations}
\end{exmp}

\vspace{\baselineskip}

Furthermore, we have the following examples for 4. and 5. of \defnref{defn:scalarcorolla}:

\vspace{\baselineskip}

\begin{exmp}
We have:
\begin{equation}
\begin{split}
	\mathscr L \left ( \vcenter{\hbox{\includegraphics[height=2cm]{1L_Sc.eps}}} , \set{1} , \emptyset \right ) = & \left \{
\vcenter{\hbox{\includegraphics[height=2cm]{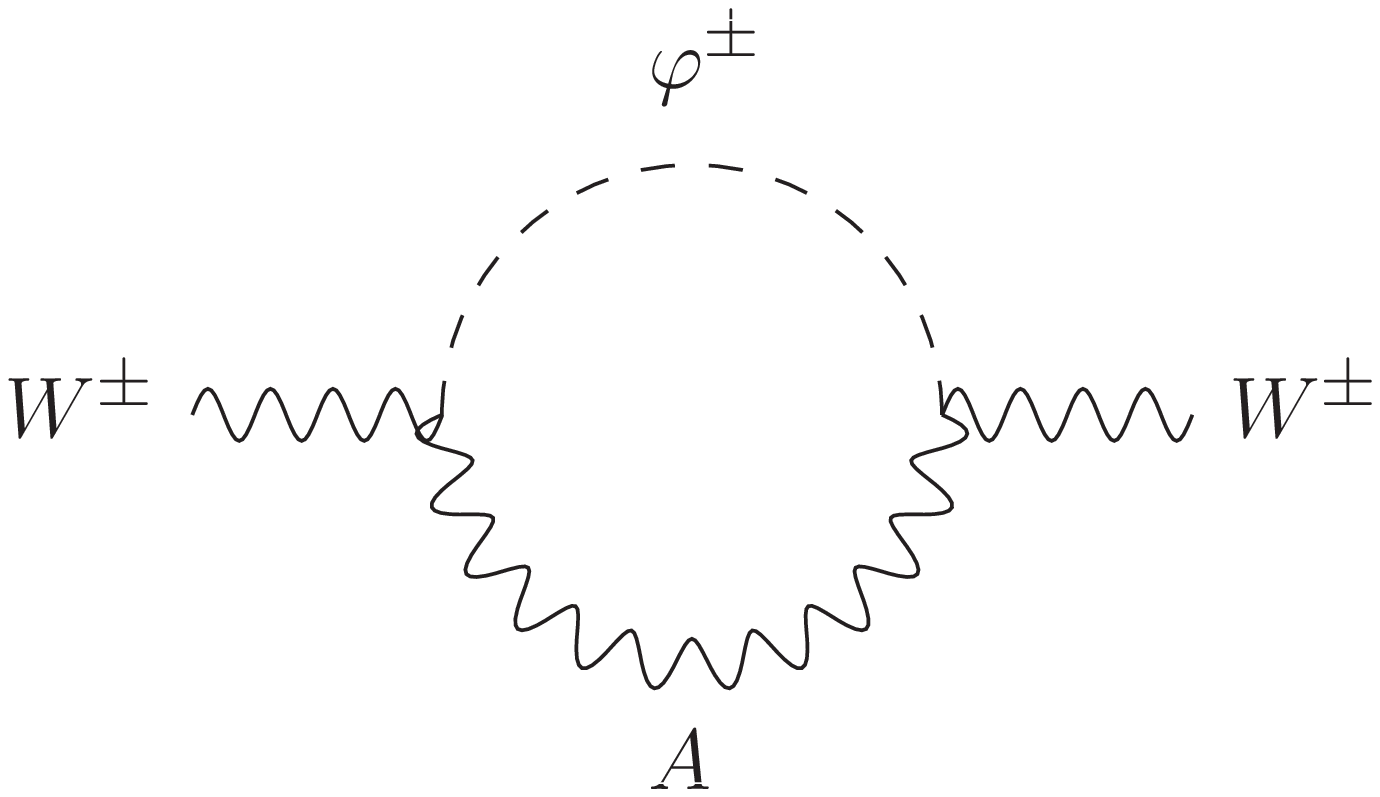}}} ,
\vcenter{\hbox{\includegraphics[height=2cm]{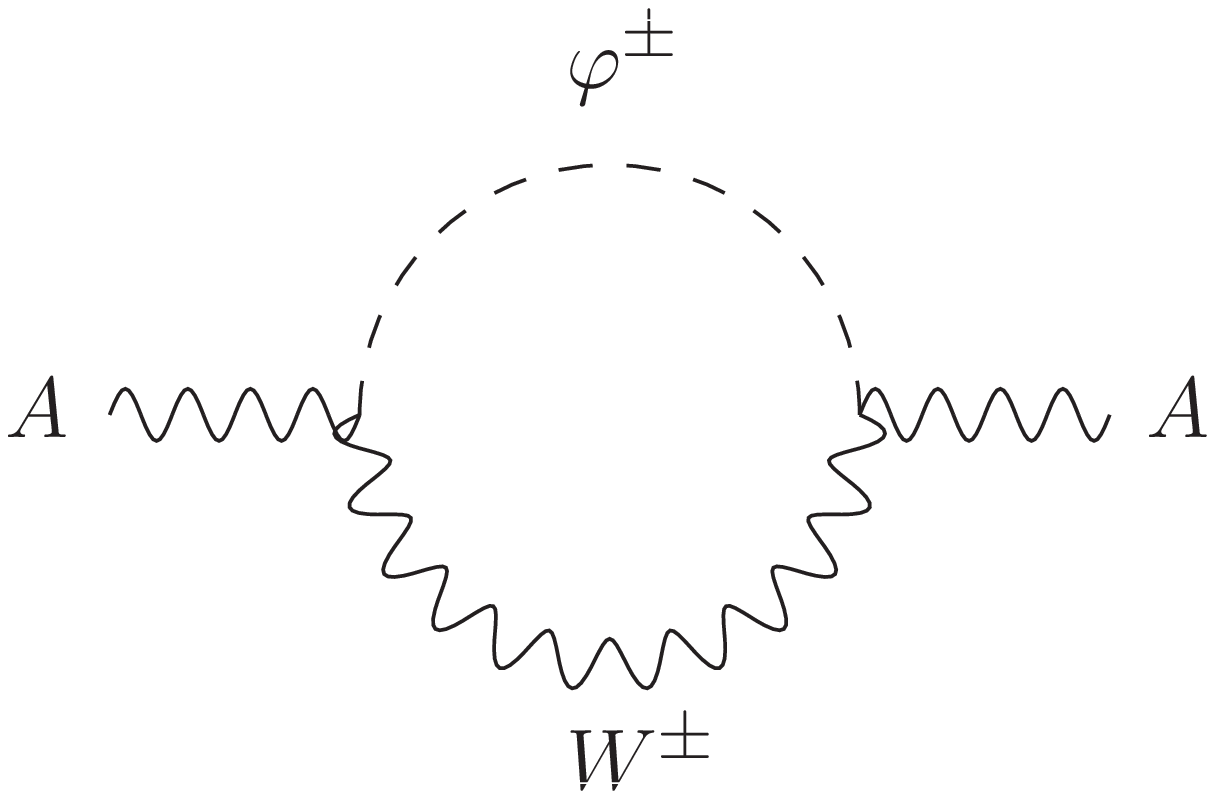}}} , \right . \\
& \left . \vcenter{\hbox{\includegraphics[height=2cm]{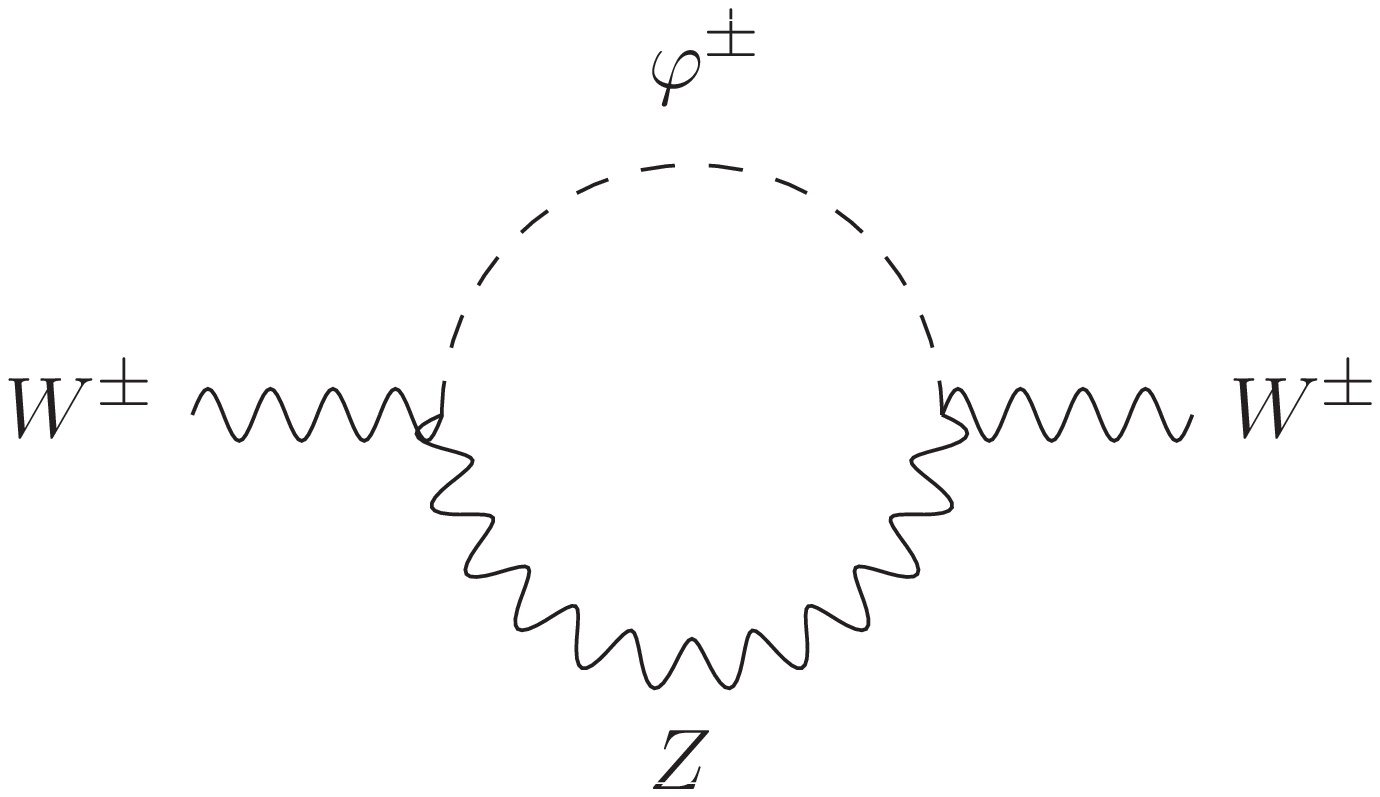}}} ,
\vcenter{\hbox{\includegraphics[height=2cm]{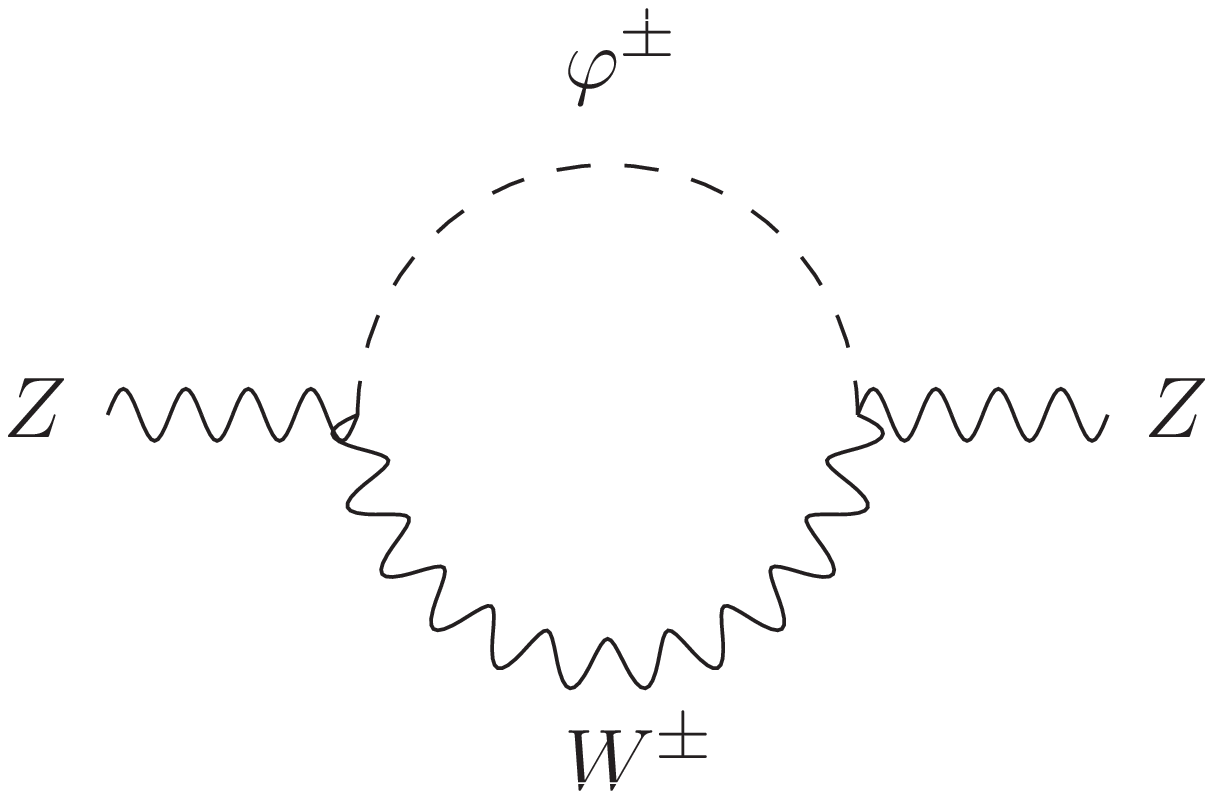}}} , \right . \\
& \left . \vcenter{\hbox{\includegraphics[height=2cm]{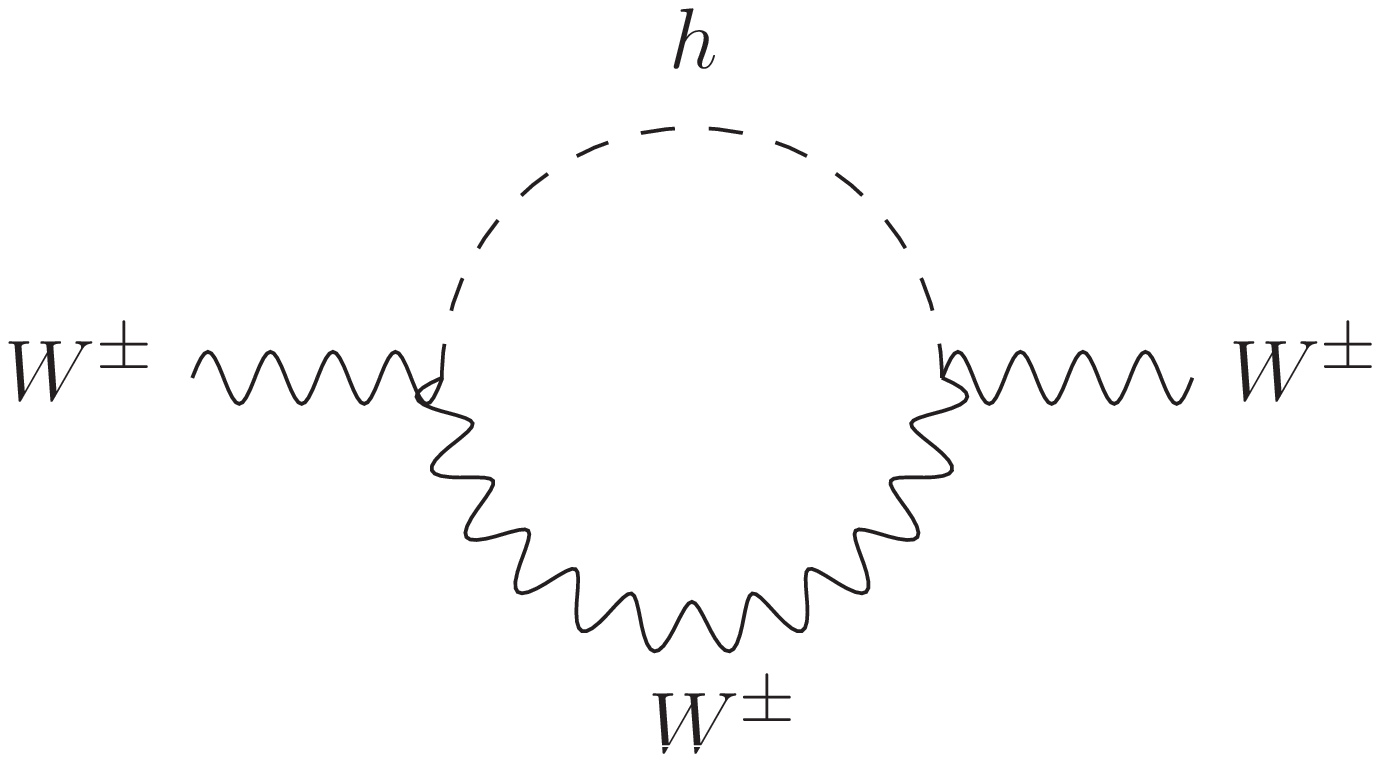}}} ,
\vcenter{\hbox{\includegraphics[height=2cm]{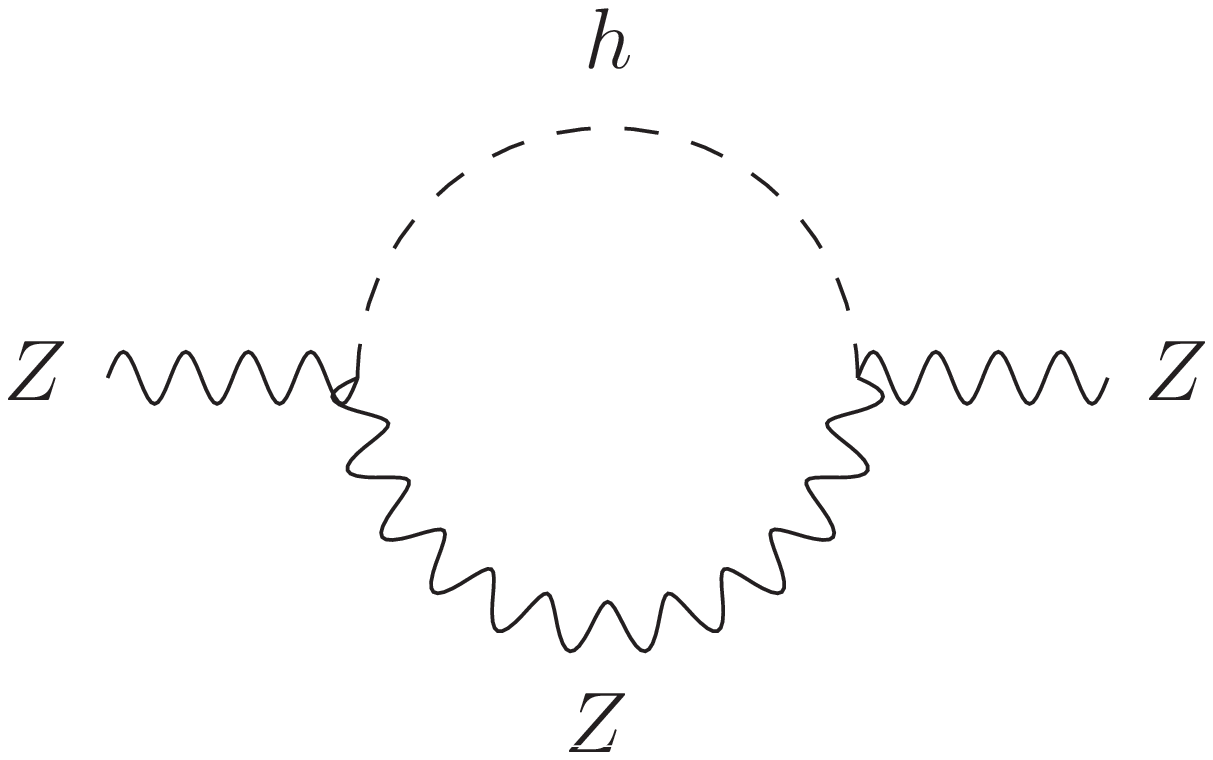}}}
\right \}
\end{split}
\end{equation}
and:
\begin{subequations}
{\allowdisplaybreaks
\begin{align}
\operatorname{coupling} \left ( \set{1} , \emptyset , \vcenter{\hbox{\includegraphics[height=2cm]{1L_Sc_labeled_1.eps}}} \right ) = & \underline{e}^2 m_W^2 \\
\operatorname{coupling} \left ( \set{1} , \emptyset , \vcenter{\hbox{\includegraphics[height=2cm]{1L_Sc_labeled_2.eps}}} \right ) = & \underline{e}^2 m_W^2 \\
\operatorname{coupling} \left ( \set{1} , \emptyset , \vcenter{\hbox{\includegraphics[height=2cm]{1L_Sc_labeled_3.eps}}} \right ) = & \underline{g}^2 m_Z^2 \sin^4 \theta_W \\
\operatorname{coupling} \left ( \set{1} , \emptyset , \vcenter{\hbox{\includegraphics[height=2cm]{1L_Sc_labeled_4.eps}}} \right ) = & \underline{g}^2 m_Z^2 \sin^4 \theta_W \\
\operatorname{coupling} \left ( \set{1} , \emptyset , \vcenter{\hbox{\includegraphics[height=2cm]{1L_Sc_labeled_5.eps}}} \right ) = & \underline{g}^2 m_W^2 \\
\operatorname{coupling} \left ( \set{1} , \emptyset , \vcenter{\hbox{\includegraphics[height=2cm]{1L_Sc_labeled_6.eps}}} \right ) = & \frac{\underline{g}^2}{\cos^2 \theta_W} m_Z^2
\end{align}
}
\end{subequations}
\end{exmp}

\vspace{\baselineskip}

Then we can define the Corolla polynomial for the gauge bosons and the scalar particles of the electroweak sector of the Standard Model as follows:

\vspace{\baselineskip}

\begin{defn}[Corolla polynomial for the electroweak sector of the Standard Model (omitting fermions)] \label{defn:corolla_polynomial-electroweak_scalar_bosons}
Let \(\Gamma\) be a 3-regular scalar QFT Feynman graph. Then we define the various summands of the Corolla polynomial for the electroweak sector of the Standard Model (omitting fermions) by
\begin{subequations}
\begin{equation}
\begin{split}
	\mathcal C ^0 _{\text{EW}} (\Gamma) := & \sum_{P_{\text{H/G}} \in \wp \left ( \Gamma^{[1]} \right )} \sum_{P_{(4)} \in \wp_{(4)} \left (P_{\text{H/G}} \right )} \sum_{L \in \mathscr L \left ( \Gamma, P_{\text{H/G}}, P_{(4)} \right )} \\ & \left [ \left ( \frac{\operatorname{sym} \left ( \Gamma \right )}{\operatorname{sym} (L) \operatorname{iso} (L)} \right ) \operatorname{coupling} \left ( P_{\text{H/G}}, P_{(4)} , L \right ) \right . \\
& \phantom{ [ } \left . e^{\left ( \sum_{e \in P_{(4)}^{[1]}} A_e \xi_e^{\prime 2} - \sum_{e \in \Gamma^{[1]} \setminus P_{(4)}^{[1]}} A_e m^2_{L(e)} \right )} \right . \\
& \phantom{ [ } \left . \left ( \prod_{\substack{v \in \Gamma^{[0]},\\v \notin P_{\text{H/G}}^{[0]}}} \mathscr V_v \right ) \left ( \prod_{\substack{h \in P_{\text{H/G}}^{[1/2]} \setminus P_{(4)}^{[1/2]}, \\ h_+, h_- \notin P_{\text{H/G}}^{[1/2]} \setminus P_{(4)}^{[1/2]}}} b_h \right ) \right . \\ & \phantom{ [ } \left . \left ( \prod_{\substack{h_+, h_- \in P_{\text{H/G}}^{[1/2]} \setminus P_{(4)}^{[1/2]} , \\ h \notin P_{\text{H/G}}^{[1/2]} \setminus P_{(4)}^{[1/2]}}} \left ( a_{h_+} + a_{h_-} \right ) \right ) \left ( \prod_{\substack{e \in P_{(4)}^{[1]} , \\ h_1, h_2 \in H_{(2)} (e)}} b_{h_1} b_{h_2} \right ) \right ]\, ,
\end{split}
\end{equation}
\begin{equation}
\begin{split}
	\mathcal C ^i _{\text{EW}} (\Gamma) := & \sum_{\substack{C_1, C_2, \ldots, C_i \in \mathscr{C}_\Gamma,\\C_j \text{ pairwise disjoint}}} \sum_{P_{\text{H/G}} \in \wp \left ( \Gamma^{[1]} \setminus \bigcup_{k=1}^i C_k^{[1]} \right )} \sum_{P_{(4)} \in \wp_{(4)} \left (P_{\text{H/G}} \right )} \sum_{L \in \mathscr L \left ( \Gamma, P_{\text{H/G}}, P_{(4)} \right )} \\ & \left [ \left ( \frac{\operatorname{sym} \left ( \Gamma \right )}{\operatorname{sym} (L) \operatorname{iso} (L)} \right ) \operatorname{coupling} \left ( P_{\text{H/G}}, P_{(4)} , L \right ) \right . \\
& \phantom{ [ } \left . e^{\left ( \sum_{e \in P_{(4)}^{[1]}} A_e \xi_e^{\prime 2} - \sum_{e \in \Gamma^{[1]} \setminus P_{(4)}^{[1]}} A_e m^2_{L(e)} \right )} \right . \\ & \phantom{ [ } \left . \left ( \prod_{j = 1}^i \mathscr G_{C_j} \right ) \left ( \prod_{\substack{v \in \Gamma^{[0]},\\v \notin \bigcup_{k=1}^i C_k^{[0]} \cup P_{\text{H/G}}^{[0]}}} \mathscr V_v \right ) \left ( \prod_{\substack{h \in P_{\text{H/G}}^{[1/2]} \setminus P_{(4)}^{[1/2]}, \\ h_+, h_- \notin P_{\text{H/G}}^{[1/2]} \setminus P_{(4)}^{[1/2]}}} b_h \right ) \right . \\ & \phantom{ [ } \left . \left ( \prod_{\substack{h_+, h_- \in P_{\text{H/G}}^{[1/2]} \setminus P_{(4)}^{[1/2]} , \\ h \notin P_{\text{H/G}}^{[1/2]} \setminus P_{(4)}^{[1/2]}}} \left ( a_{h_+} + a_{h_-} \right ) \right ) \left ( \prod_{\substack{e \in P_{(4)}^{[1]} , \\ h_1, h_2 \in H_{(2)} (e)}} b_{h_1} b_{h_2} \right ) \right ]
\end{split}
\end{equation}
and the Corolla polynomial by
\begin{equation}
	\mathcal C _{\text{EW}} (\Gamma) := \sum_{i=0}^\infty (-1)^i \mathcal C^i _{\text{EW}} (\Gamma) \, .
\end{equation}
\end{subequations}
\end{defn}

\vspace{\baselineskip}

\begin{thm}[\cite{Prinz}] \label{thm:corolla_polynomial-electroweak_scalar_bosons}
Let \(\Gamma\) be a connected\footnote{Again, we emphasize that we really mean 1-connected graphs, not only 2-connected (as they are called in mathematics) or 1 PI (as they are called in physics) graphs.} 3-regular scalar QFT Feynman graph. Then, the parametric integrand of all Feynman graphs from quantum chromodynamics and the electroweak sector of the Standard Model which are related to \(\Gamma\) via graph and cycle homology including all possible particle type labelings allowed by the Standard Model Feynman rules (omitting fermions) \(\tilde{I}_{\mathcal F} (\Gamma)\) can be obtained via acting with the sum of the two Corolla differentials for quantum chromodynamics and the electroweak sector of the Standard Model \(\left ( \mathcal D_{\text{\emph{QCD}}} (\Gamma) + \mathcal D_{\text{\emph{EW}}} (\Gamma) \right )\) on the corresponding parametric integrand \(I (\Gamma)\) of \(\Gamma\)\footnote{Recall \(\mathcal D_{\text{QCD}} (\Gamma) :=  i^{\left \vert \Gamma^{[1]} \right \vert} \underline{g_s}^{\left \vert \Gamma^{[0]} \right \vert} \operatorname{color} (\Gamma) \mathcal D (\Gamma)\) from \defnref{defn:corolla_polynomial} and \defnref{defn:corolla_differential}.},
\begin{equation}
	\tilde{I}_{\mathcal F} (\Gamma) = \left ( \mathcal D_{\text{\emph{QCD}}} (\Gamma) + \mathcal D_{\text{\emph{EW}}} (\Gamma) \right ) I(\Gamma) \, .
\end{equation}
\end{thm}

\begin{proof}
First note that the scalar particles \(h\), \(\varphi^\pm\) and \(\varphi_Z\) only couple to gauge bosons of the electroweak sector of the Standard Model so that we are allowed to treat it separately from quantum chromodynamics. The contribution for quantum chromodynamics via the action from \(\mathcal D_{\text{QCD}} (\Gamma)\) on \(I(\Gamma)\) was already proven in \thmref{thm:gauge_theory_amplitude_pure_yang-mills}. For the contribution of the electroweak sector of the Standard Model notice that for a given graph \(\Gamma\) every edge that is not going to be turned into a ghost edge is allowed to become a scalar edge (possibly with a restriction on the edge labeling). This is precisely created by the power set of all edges of the graph \(\Gamma\) without the edges that will be turned into ghost edges, i.e. \(\wp \left ( \Gamma^{[1]} \setminus \bigcup_{k=1}^i C_k^{[1]} \right )\) for the contribution with \(i\) ghost loops. Moreover, the creation of 4-valent scalar vertices (2-scalar-2-gauge boson or 4-scalar boson vertices) can be constructed via shrinking suitable internal scalar labeled edges between two 3-valent vertices. Concrete, the requirements are having \(2\) or \(4\) scalar edge neighbors, i.e. edges in the set \(P_{\text{H/G}}\). Furthermore, they are not allowed to share a vertex with an edge which will be turned into a ghost edge (again, we remark that also edges which will be turned into fermion edges are not allowed either, if included). And finally, also adjacent edges are not allowed to shrink simultaneously. Therefore, \(\wp_{(4)} \left ( P_{\text{H/G}} \right )\) consists of all sets of edges that are allowed to shrink in the set \(P_{\text{H/G}}\) in order to produce valid 4-valent scalar vertices. Furthermore, the Feynman rules for the 3-valent 1-scalar-2-gauge boson and 2-scalar-1-gauge boson vertices are created by the additional products of the Corolla polynomial for the electroweak sector of the Standard Model (cf. \defnref{defn:corolla_polynomial-electroweak_scalar_bosons} compared to \defnref{defn:corolla_polynomial}). Moreover, since the original Corolla polynomial acts only on the part(s) of \(\Gamma\) which will be turned into pure gauge boson vertices or their corresponding ghosts, the 4-valent 4-gauge boson vertices are created just like with the standard Corolla differential and can be received as the corresponding residues in the shrinked edge Schwinger parameters, cf. \thmref{thm:gauge_theory_amplitude_pure_yang-mills}. Additionally, the right symmetry factors are obtained by the multiplication with the fraction \(\operatorname{sym} (\Gamma) / \operatorname{sym} (L)\) and possible redundancies are divided out by the factor \(\operatorname{iso} (L)\). The corresponding particle masses are included via the exponential term which also removes the Schwinger parameters for shrinked scalar labeled edges. Finally all corresponding coupling constants for the labeling \(L \in \mathscr L \left ( \Gamma , P_{\text{H/G}}, P_{(4)} \right )\) are given in \(\operatorname{coupling} \left ( P_{\text{H/G}}, P_{(4)} , L \right )\).
\end{proof}

\section{Conclusion} \label{sec:conclusion}
The aim of this paper was to review the generalization of the Corolla polynomial defined in \cite{Kreimer_Yeats,Kreimer_Sars_Suijlekom,Sars} to the bosons of the electroweak sector of the Standard Model, which was first worked out in \cite{Prinz}. Therefore, all the relevant graph theoretic notions (cf. \mbox{Section \ref{sec:graph_theoretic_notions}}), the parametric representation of scalar quantum field theories (cf. \mbox{Section \ref{sec:parametric_representations_of_scalar_quantum_field_theories}}) and the Corolla polynomial and differential for pure Yang-Mills theory (cf. \mbox{Subsection \ref{ssec:corolla_polynomial_pure_yang-mills}}) were reviewed. The inclusion of the bosons of the electroweak sector of the Standard Model was made possible in two steps. First by working out the combinatorics of labeling a 3-regular scalar QFT Feynman graph with labels of the gauge bosons of the electroweak sector of the Standard Model (cf. \mbox{Subsection \ref{ssec:corolla_polynomial_gauge_bosons_electroweak}}). Then secondly by working out the additional tensor structures arising from the inclusion of the Feynman rules for the scalar particles of the electroweak sector of the Standard Model (cf. \mbox{Subsection \ref{ssec:corolla_polynomial_scalar_particles}}). We showed that at least for the gauge bosons of the electroweak sector of the Standard Model the symmetry factors work out correct (cf. \colref{cor:inclusion_ew_gauge_bosons}). Furthermore, a relatively compact notation was given for the inclusion of the scalar particles of the Standard Model, compared to the standard Standard Model Feynman rules (cf. \defnref{defn:corolla_polynomial-electroweak_scalar_bosons}, \defnref{defn:corolla_differential}, \thmref{thm:corolla_polynomial-electroweak_scalar_bosons}, \cite[Appendix A]{Prinz} and \cite{Romao_Silva}).

As projects for future work, several topics related to the Corolla polynomial are of particular interest: First of all, although it is in principle clear \cite{Kreimer_Sars_Suijlekom}, the combinatorics for the inclusion of fermions to the Corolla polynomial for the electroweak sector of the Standard Model should be explicitly worked out. Secondly, as we believe that this approach is useful for computer calculations --- since derivations can be done with much less computational afford than integrations --- it would be useful to bring this approach on a computer. And thirdly, since it is believed that the Corolla polynomial can also be generalized to the case of quantum gravity \cite{Kreimer_Sars_Suijlekom}, it would be interesting to see what we could learn from this approach and its underlying combinatorics.

\section*{Acknowledgments}
It is my pleasure to thank Prof. Dr. Dirk Kreimer, Dr. Christian Bogner and the rest of the group for the welcoming atmosphere, their support and illuminating and helpful discussions!

\bibliography{References}{}

\begin{thebibliography}{10}
  \providebibliographyfont{name}{}%
  \providebibliographyfont{lastname}{}%
  \providebibliographyfont{title}{\emph}%
  \providebibliographyfont{jtitle}{\btxtitlefont}%
  \providebibliographyfont{etal}{\emph}%
  \providebibliographyfont{journal}{}%
  \providebibliographyfont{volume}{}%
  \providebibliographyfont{ISBN}{\MakeUppercase}%
  \providebibliographyfont{ISSN}{\MakeUppercase}%
  \providebibliographyfont{url}{\url}%
  \providebibliographyfont{numeral}{}%
  \expandafter\btxselectlanguage\expandafter {\btxfallbacklanguage}

\btxselectlanguage {english}
\bibitem {Kreimer_Yeats}
\btxnamefont {\btxlastnamefont {{D. Kreimer}}} \btxandlong {}\ \btxnamefont
  {\btxlastnamefont {{K. Yeats}}}\btxauthorcolon\ \btxjtitlefont
  {\btxifchangecase {{P}roperties of the {C}orolla {P}olynomial of a 3-regular
  {G}raph}{{P}roperties of the {C}orolla {P}olynomial of a 3-regular {G}raph}}.
\newblock \btxjournalfont {Electr. J. Comb.}, \btxpageslong {}\ 41--41, 2013.
\newblock {\latintext \btxurlfont{https://arxiv.org/abs/1207.5460v1}}.

\bibitem {Kreimer_Sars_Suijlekom}
\btxnamefont {\btxlastnamefont {{D. Kreimer}}}, \btxnamefont {\btxlastnamefont
  {{M. Sars}}}\btxandcomma {} \btxandlong {}\ \btxnamefont {\btxlastnamefont
  {{W. D. van Suijlekom}}}\btxauthorcolon\ \btxjtitlefont {\btxifchangecase
  {{{Q}uantization of gauge fields, graph polynomials and graph
  homology}}{{{Q}uantization of gauge fields, graph polynomials and graph
  homology}}}.
\newblock \btxjournalfont {Annals Phys.}, 336:180--222, 2013.
\newblock {\latintext \btxurlfont{https://arxiv.org/abs/1208.6477v4}}.

\bibitem {Sars}
\btxnamefont {\btxlastnamefont {{M. Sars}}}\btxauthorcolon\ \btxtitlefont
  {{P}arametric {R}epresentation of {F}eynman {A}mplitudes in {G}auge
  {T}heories}.
\newblock \btxphdthesis {}, Humboldt-Universit\"at zu Berlin,
  \btxprintmonthyear{.}{January}{2015}{long}.
\newblock {\latintext
  \btxurlfont{https://www2.mathematik.hu-berlin.de/~kreimer/wp-content/uploads/SarsThesisNeu.pdf}}.

\bibitem {Prinz}
\btxnamefont {\btxlastnamefont {{D. Prinz}}}\btxauthorcolon\ \btxjtitlefont
  {\btxifchangecase {{T}he {C}orolla {P}olynomial for spontaneously broken
  {G}auge {T}heories}{{T}he {C}orolla {P}olynomial for spontaneously broken
  {G}auge {T}heories}}.
\newblock \btxjournalfont {Master thesis},
  \btxprintmonthyear{.}{September}{2015}{long}.
\newblock {\latintext
  \btxurlfont{https://www2.mathematik.hu-berlin.de/~kreimer/wp-content/uploads/prinz.pdf}}.

\bibitem {Kreimer}
\btxnamefont {\btxlastnamefont {{D. Kreimer}}}\btxauthorcolon\ \btxtitlefont
  {\btxifchangecase {{D}irk {K}reimer's lecture series ({QFT I}, {QFT II},
  {H}opf {A}lgebras and the renormalization group, {D}yson–{S}chwinger
  equations and quantization of gauge theories).}{{D}irk {K}reimer's lecture
  series ({QFT I}, {QFT II}, {H}opf {A}lgebras and the renormalization group,
  {D}yson–{S}chwinger equations and quantization of gauge theories).}}
\newblock {\latintext
  \btxurlfont{https://www2.mathematik.hu-berlin.de/~kreimer/teaching/}},
  Scripts for the latter two by Lutz Klaczynski can be found here
  \url{https://www2.mathematik.hu-berlin.de/~klacz/SkriptRGE.pdf} and here
  \url{https://www2.mathematik.hu-berlin.de/~klacz/DSE.pdf}.

\bibitem {Peskin_Schroeder}
\btxnamefont {\btxlastnamefont {{M. E. Peskin}}} \btxandlong {}\ \btxnamefont
  {\btxlastnamefont {{D. V. Schroeder}}}\btxauthorcolon\ \btxtitlefont {An
  Introduction to Quantum Field Theory}.
\newblock Advanced book classics. \btxpublisherfont {Addison-Wesley Publishing
  Company}, 1995\ifbtxprintISBN {, \mbox{\btxISBN~\btxISBNfont
  {9780201503975}}}.
\newblock {\latintext
  \btxurlfont{https://books.google.de/books?id=i35LALN0GosC}}.

\bibitem {Itzykson_Zuber}
\btxnamefont {\btxlastnamefont {{C. Itzykson}}} \btxandlong {}\ \btxnamefont
  {\btxlastnamefont {{J. B. Zuber}}}\btxauthorcolon\ \btxtitlefont {Quantum
  Field Theory}.
\newblock Dover Books on Physics. \btxpublisherfont {Dover Publications},
  2012\ifbtxprintISBN {, \mbox{\btxISBN~\btxISBNfont {9780486134697}}}.
\newblock {\latintext
  \btxurlfont{https://books.google.de/books?id=CxYCMNrUnTEC}}.

\bibitem {Romao_Silva}
\btxnamefont {\btxlastnamefont {{J. C. Rom{\~a}o}}} \btxandlong {}\
  \btxnamefont {\btxlastnamefont {{J. P. Silva}}}\btxauthorcolon\
  \btxjtitlefont {\btxifchangecase {{A} resource for signs and {F}eynman
  diagrams of the {S}tandard {M}odel}{{A} resource for signs and {F}eynman
  diagrams of the {S}tandard {M}odel}}.
\newblock \btxjournalfont {arXiv}, arXiv:1209.6213v2,
  \btxprintmonthyear{.}{October}{2012}{long}.
\newblock {\latintext \btxurlfont{https://arxiv.org/abs/1209.6213v2}}.

\bibitem {Suijlekom}
\btxnamefont {\btxlastnamefont {{W. D. van Suijlekom}}}\btxauthorcolon\
  \btxjtitlefont {\btxifchangecase {{R}enormalization of {G}auge {F}ields: {A}
  {H}opf {A}lgebra {A}pproach}{{R}enormalization of {G}auge {F}ields: {A}
  {H}opf {A}lgebra {A}pproach}}.
\newblock \btxjournalfont {Commun. Math. Phys.}, 276:773–798,
  \btxprintmonthyear{.}{April}{2007}{long}.
\newblock {\latintext \btxurlfont{https://arxiv.org/abs/hep-th/0610137v1}},
  Digital Object Identifier (DOI) 10.1007/s00220-007-0353-9.

\bibitem {Diestel}
\btxnamefont {\btxlastnamefont {{R. Diestel}}}\btxauthorcolon\ \btxtitlefont
  {{G}raph {T}heory, 4th {E}dition}, \btxvolumelong {}\ \btxvolumefont {173}
  \btxofserieslong {}\ \btxtitlefont {Graduate texts in mathematics}.
\newblock \btxpublisherfont {Springer}, 2012\ifbtxprintISBN {,
  \mbox{\btxISBN~\btxISBNfont {978-3-642-14278-9}}}.
\newblock {\latintext \btxurlfont{https://diestel-graph-theory.com/}}.

\bibitem {Brown_Kreimer}
\btxnamefont {\btxlastnamefont {{F. Brown}}} \btxandlong {}\ \btxnamefont
  {\btxlastnamefont {{D. Kreimer}}}\btxauthorcolon\ \btxjtitlefont
  {\btxifchangecase {{A}ngles, {S}cales and {P}arametric
  {R}enormalization}{{A}ngles, {S}cales and {P}arametric {R}enormalization}}.
\newblock \btxjournalfont {Letters in Mathematical Physics}, 103(9):933--1007,
  2013\ifbtxprintISSN {, \mbox{\btxISSN~\btxISSNfont {0377-9017}}}.
\newblock {\latintext \btxurlfont{https://arxiv.org/abs/1112.1180v1}}.

\end{thebibliography}
\bibliographystyle{babunsrt}

\end{document}